\documentclass[11pt,a4paper]{article}
\pdfoutput=1

\usepackage{amssymb}
\usepackage{amsmath}
\usepackage{amsfonts}
\usepackage{bbm}
\usepackage{amsthm}
\usepackage{mathrsfs}
\usepackage{hyperref}
\usepackage{color}
\usepackage[margin=2.41cm]{geometry}
\usepackage[all,cmtip]{xy}
\usepackage[utf8]{inputenc}
\usepackage{graphicx}
\usepackage{varwidth}
\usepackage{comment}
\usepackage[shortlabels]{enumitem}
\usepackage{todonotes}

\usepackage{mathtools}

\usepackage{upgreek}
\usepackage{rotating}

\usepackage{tikz}
\usetikzlibrary{shapes.geometric}

\usepackage{tikz-cd}
\usetikzlibrary{matrix,arrows,calc,decorations.pathmorphing,fit,shapes.geometric,decorations.pathreplacing,positioning}

\usepackage{amssymb}
\tikzset{curve/.style={settings={#1},to path={(\tikztostart)
    .. controls ($(\tikztostart)!\pv{pos}!(\tikztotarget)!\pv{height}!270:(\tikztotarget)$)
    and ($(\tikztostart)!1-\pv{pos}!(\tikztotarget)!\pv{height}!270:(\tikztotarget)$)
    .. (\tikztotarget)\tikztonodes}},
    settings/.code={\tikzset{quiver/.cd,#1}
        \def\pv##1{\pgfkeysvalueof{/tikz/quiver/##1}}},
    quiver/.cd,pos/.initial=0.35,height/.initial=0}
\definecolor{darkred}{rgb}{0.8,0.1,0.1}

\hypersetup{
	colorlinks=true,         % false: boxed links; true: colored links,false is default
	urlcolor=darkred, %blue??
	linkcolor=darkred,
	citecolor=blue,
}

\theoremstyle{plain}
\newtheorem{theo}{Theorem}[section]
\newtheorem{lem}[theo]{Lemma}
\newtheorem{propo}[theo]{Proposition}
\newtheorem{cor}[theo]{Corollary}

\theoremstyle{definition}
\newtheorem{defi}[theo]{Definition}

\newtheorem{open}[theo]{Open Problem}

\newtheorem{exx}[theo]{Example}
\newenvironment{ex}
{\pushQED{\qed}\exx}
{\popQED\endexx}

\newtheorem{remm}[theo]{Remark}
\newenvironment{rem}
{\pushQED{\qed}\remm}
{\popQED\endremm}

\numberwithin{equation}{section}

\def\nn{\nonumber}

\def\bbK{\mathbb{K}}
\def\bbR{\mathbb{R}}
\def\bbC{\mathbb{C}}

\def\bbZ{\mathbb{Z}}

\def\id{\mathrm{id}}

\def\1{I}

\def\op{\mathrm{op}}

\def\Mor{\operatorname{Mor}}

\def\Loc{\mathbf{Loc}}

\def\Set{\mathbf{Set}}
\def\Alg{\mathbf{Alg}}

\def\Vec{\mathbf{Vec}}
\def\Ch{\mathbf{Ch}}

\def\CC{\mathbf{C}}

\def\EE{\mathbf{E}}

\def\TT{\mathbf{T}}

\def\CAT{\mathbf{CAT}}

\def\AQFT{\mathbf{AQFT}}

\def\tPFA{\mathbf{tPFA}}

\def\Fun{\mathbf{Fun}}

\def\CombR{\mathbf{Comb}^R_{\mathrm{tr}}}
\def\SectR{\mathbf{Sect}^R}

\def\AAA{\mathfrak{A}}

\def\BBB{\mathfrak{B}}
\def\GGG{\mathfrak{G}}

\def\FFF{\mathfrak{F}}
\def\GGG{\mathfrak{G}}

\def\O{\mathcal{O}}
\def\P{\mathcal{P}}
\def\Q{\mathcal{Q}}

\def\HK{\mathsf{HK}}
\def\CG{\mathsf{CG}}

\def\tP{\mathsf{t}\mathcal{P}}

\def\colim{\mathrm{colim}}

\def\bilim{\mathrm{bilim}}

\def\rc{\mathrm{rc}}
\def\pt{\mathrm{pt}}
\def\dc{\mathrm{dc}}
\def\as{\mathrm{as}}

\newcommand\und[1]{\underline{#1}}
\newcommand\ovr[1]{\overline{#1}}

\def\sk{\vspace{2mm}}

\makeatletter
\let\@fnsymbol\@alph
\makeatother

\title{%
On the equivalence of AQFTs and prefactorization algebras
}

\author{%
Marco Benini$^{1,2,a}$, Victor Carmona$^{3,b}$, Alastair Grant-Stuart$^{4,c}$\ and\ 
Alexander Schenkel$^{4,d}$\vspace{4mm}\\
{\small ${}^1$ Dipartimento di Matematica, Dipartimento di Eccellenza 2023-27, Universit\`a di Genova,}\\
{\small Via Dodecaneso 35, 16146 Genova, Italy.}\vspace{2mm}\\
{\small ${}^2$ INFN, Sezione di Genova,}\\
{\small Via Dodecaneso 33, 16146 Genova, Italy.}\vspace{2mm}\\
{\small ${}^3$ Max Planck Institut f\"ur Mathematik in den Naturwissenschaften,}\\
{\small Inselstra\ss e 22, 04103 Leipzig, Germany.}\vspace{2mm}\\
{\small ${}^4$ School of Mathematical Sciences, University of Nottingham,}\\
{\small University Park, Nottingham NG7 2RD, United Kingdom.}\vspace{4mm}\\
{\small \begin{tabular}{ll}
Email: & ${}^a$~\href{mailto:marco.benini@unige.it}{\texttt{marco.benini@unige.it}}\\
& ${}^b$~\href{mailto:victor.carmona@mis.mpg.de}{\texttt{victor.carmona@mis.mpg.de}}\\
&${}^c$~\href{mailto:alastair.grant-stuart@nottingham.ac.uk}{\texttt{alastair.grant-stuart@nottingham.ac.uk}}\\
& ${}^d$~\href{mailto:alexander.schenkel@nottingham.ac.uk}{\texttt{alexander.schenkel@nottingham.ac.uk}}
\vspace{2mm}
\end{tabular}
}
}

\date{January 2026}

\begin{document}

\maketitle

\begin{abstract}
\noindent This paper revisits the equivalence problem between algebraic quantum field theories and prefactorization algebras defined over globally hyperbolic Lorentzian manifolds. We develop a radically new approach whose main innovative features are 1.)~a structural implementation of the additivity property used in earlier approaches and 2.)~a reduction of the global equivalence problem to a family of simpler spacetime-wise problems. When applied to the case where the target category is a symmetric monoidal $1$-category, this yields a generalization of the equivalence theorem from [Commun.\ Math.\ Phys.\ \textbf{377}, 971 (2019)]. In the case where the target is the symmetric monoidal $\infty$-category of cochain complexes, we obtain a reduction of the global $\infty$-categorical equivalence problem to simpler, but still challenging, spacetime-wise problems. The latter would be solved by showing that certain functors between $1$-categories exhibit $\infty$-localizations, however the available detection criteria are inconclusive in our case.
\end{abstract}
\vspace{-1mm}

\paragraph*{Keywords:} algebraic quantum field theory, factorization algebras, homotopical algebra, operads, localizations, Lorentzian geometry
\vspace{-2mm}

\paragraph*{MSC 2020:} 81Txx, 18Nxx, 53C50
\vspace{-2mm}

\tableofcontents

%%%%%%%%%%%%%%%%%%%%%%%%%%%%%%%%%%%%%%%%%%%%%%%%
%%%%%%%%%%%%%%%%%%%%%%%%%%%%%%%%%%%%%%%%%%%%%%%%

\section{Introduction and summary}
Quantum field theory (QFT) is a highly successful and versatile 
concept which, in addition to its intrinsic relevance for physics,
has inspired countless influential developments in mathematics.
In order to facilitate these interdisciplinary exchanges,
it is essential to capture the key features of QFT in terms of mathematical axioms,
which then can be analyzed and developed further by purely mathematical techniques.
The desire to lay the mathematical foundations of QFT has triggered
significant developments over the past decades,
leading to a whole zoo of different axiomatic frameworks for QFT.
Among the most prominent recent frameworks are \textit{algebraic QFT (AQFT)}, 
going back to works of Haag and Kastler \cite{HaagKastler} and Brunetti, 
Fredenhagen and Verch \cite{BFV}, \textit{functorial QFT},
initiated by Witten \cite{Witten}, Atiyah \cite{Atiyah} and Segal \cite{Segal},
and the \textit{factorization algebras} of Costello and Gwilliam \cite{CG1,CG2}.
\sk

The availability of different mathematical axiomatizations 
for the same concept of a QFT prompts a very fundamental and important question:
Are the different axiomatizations compatible with each other? 
Or even better: Are they equivalent? Finding precise answers to
these questions is important for various reasons. 
On the one hand, major incompatibilities between different 
frameworks would indicate that they have failed their main task 
to capture the essential content of a QFT and hence they should be revised.
On the other hand, finding equivalences would show unity between different frameworks
and it would open up new pathways for a fruitful exchange of ideas and techniques across different 
research communities. For example, using equivalence theorems, one can
analyze the same QFT problem from different mathematical perspectives,
which extends the range of available techniques and increases 
the chances of successfully finding a solution.
\sk

We will now provide a brief overview of the existing comparison results
for the case of QFTs which are defined over \textit{Lorentzian spacetime manifolds}, 
which is also the context of our present paper. A first link between algebraic QFTs
and factorization algebras was pointed out in \cite{GwilliamRejzner1}, where it
is shown that a certain class of examples, called the free (i.e.\ non-interacting) QFTs, 
admit compatible descriptions in both frameworks. In particular, it was observed that the
factorization products of the factorization algebra encode the 
time-ordered products of the AQFT, and that the time-slice axiom (encoding
a concept of time-evolution) is the key ingredient to recover the spacetime-wise
unital associative algebra structures of the AQFT from the factorization algebra.
These example-based comparison results were then generalized later in \cite{BMScomparison}
to the case of free (higher) gauge theories and in \cite{GwilliamRejzner2} to the case of 
perturbatively interacting examples.
\sk

While such example-based comparisons provide important evidence that 
the different approaches are compatible with each other, they are not
sufficient to compare the underlying axiomatic frameworks and in particular
to establish equivalence theorems among them. A first categorical 
equivalence theorem between algebraic QFTs and factorization algebras 
was proven in \cite{BPScomparison}. This equivalence theorem assumes 
the very restrictive hypothesis that the target category $\TT$ is an 
ordinary symmetric monoidal $1$-category, which
excludes on both sides the key examples of gauge theories as these
take values in the symmetric monoidal $\infty$-category $\Ch$ of cochain complexes.
Despite this oversimplified context, proving a general equivalence 
theorem has led to the following important observations and 
lessons about the finer details of both axiomatic frameworks:
\begin{enumerate}
\item The notion of tuples of mutually disjoint open subsets in the factorization 
products from \cite{CG1,CG2} has to be refined in the Lorentzian geometric
context to the stronger concept of time-orderable tuples (see Definition \ref{def:Lorentz})
in order to obtain a variant of prefactorization algebras which is comparable to AQFTs.
These Lorentzian geometric variants have been called \textit{time-orderable prefactorization algebras (tPFAs)}
in \cite{BPScomparison} and we will follow this terminology in our present work.

\item The idea in \cite{GwilliamRejzner1} to use the 
time-slice axiom to construct the spacetime-wise
unital associative algebra structures of the AQFT from the 
prefactorization algebra can be realized in a model-independent context,
but it is not sufficient to construct an AQFT from a tPFA. 
The reason for this is that there exist spacetime morphisms $f:M\to N$
for which these algebra structures are not natural, hence these spacetime-wise
algebras do not assemble into a functor. This problem has 
a geometric origin, which can be traced back to the obstruction 
to the extension of Cauchy surfaces along generic spacetime 
morphisms $f:M\to N$. It has been solved in \cite{BPScomparison} 
by introducing an \textit{additivity property} 
for both AQFTs and tPFAs which demands that the global observables 
on a spacetime $M$ can be recovered from a colimit over the directed set 
$\{U\subseteq M\}$ of all relatively compact open regions in $M$.
\end{enumerate}
The main result of \cite{BPScomparison} is the construction
of two functors $\AQFT^{\mathrm{add},W} \rightleftarrows \tPFA^{\mathrm{add},W}$
which exhibit an equivalence between the category 
$\AQFT^{\mathrm{add},W}$ of additive AQFTs satisfying the time-slice axiom
and the category $\tPFA^{\mathrm{add},W}$ of additive tPFAs satisfying the time-slice axiom.
\sk

For completeness, we would also like to point the reader to
\cite{Bunk,Bunk2} where comparison theorems between AQFTs and Lorentzian
functorial QFTs have been proven in the case where the target $\TT$ is 
a symmetric monoidal $1$-category. These constructions are not directly
relevant for our present paper, but the new techniques we develop 
might be useful to generalize these results to QFTs taking values
in the symmetric monoidal $\infty$-category $\Ch$ of cochain complexes.
\sk

The aim of this paper is to present significant progress towards
an $\infty$-categorical\footnote{All $\infty$-categories that appear 
in our paper will be presented concretely by (semi-)model categories.} equivalence theorem between AQFTs and tPFAs taking
values in the symmetric monoidal $\infty$-category $\Ch$ 
of cochain complexes.\footnote{Our constructions and results
are extendable to the case where the target is any presentably
symmetric monoidal $\infty$-category, which by the results of \cite{NikolausSagave} 
can be presented by a combinatorial and tractable symmetric monoidal model category.
We however restrict our attention to the symmetric monoidal $\infty$-category $\Ch$ 
of cochain complexes because this is the most relevant example in the context of QFT.}
We would like to emphasize that this problem is considerably more abstract and 
challenging than its $1$-categorical version in \cite{BPScomparison}. In particular,
the rather explicit proof strategy from the latter work is not directly
transferable to an $\infty$-categorical context for the following reasons:
Firstly, $\Ch$-valued AQFTs and tPFAs satisfy an $\infty$-categorical relaxation
of the time-slice axiom, called the \textit{homotopy time-slice axiom}, which involves 
quasi-isomorphisms of cochain complexes instead of isomorphisms.
Since quasi-isomorphisms can in general be inverted only weakly (i.e.\ up to homotopy),
the construction of the spacetime-wise algebra structures
in \cite{BPScomparison} would have to be supplemented by a tower of homotopy 
coherence data, which is expected to assemble into spacetime-wise 
$A_\infty$-algebra structures that depend homotopy-coherent-functorially on the spacetime.
Such constructions are difficult to control, which makes it practically impossible to
construct an $\infty$-categorical analogue of the functor 
$\tPFA^{\mathrm{add},W}\to \AQFT^{\mathrm{add},W}$ from \cite{BPScomparison}.
Secondly, the additivity property, which is manifestly used in 
the constructions of \cite{BPScomparison}, 
is a kind of descent (i.e.\ local-to-global) condition which is difficult
to work with in an $\infty$-categorical context. 
\sk

As a consequence of these obstacles, a key preliminary step 
towards an $\infty$-categorical comparison between 
AQFTs and tPFAs is to develop a new proof strategy in the $1$-categorical case
which is better transferable to the $\infty$-categorical context than the techniques 
of \cite{BPScomparison}. In our work we present such a new proof strategy
which has the following innovative features: 
\begin{enumerate}
\item We provide a \textit{structural}
implementation of the relevant features of the additivity \textit{property} for AQFTs and tPFAs.
This is achieved by restricting the usual spacetime category $\Loc$ (see Definition \ref{def:Lorentz})
to the wide subcategory $\Loc^\rc\subseteq \Loc$ from Definition \ref{def:Locrc}
which has the same objects but fewer morphisms, namely only those $\Loc$-morphisms $f:M\to N$
whose image $f(M)\subseteq N$ is either Cauchy or relatively compact. The main effect
of this restriction is that Cauchy surfaces can be extended along all $\Loc^\rc$-morphisms.
We show in Remarks \ref{rem:additivity} and \ref{rem:additivitytPFA}
that the additive AQFTs/tPFAs over $\Loc$ from \cite{BPScomparison} form a full subcategory of
the AQFTs/tPFAs over $\Loc^\rc$ (which we characterize fully by a similar 
additivity property), hence our new structural approach
generalizes the old one based on the additivity property. 

\item We employ, and also generalize to the $\infty$-categorical context, 
the recent Haag-Kastler $2$-functor techniques from \cite{BGSHaagKastler}
in order to reduce the global equivalence problem for AQFTs and tPFAs over $\Loc^\rc$
to a family of simpler local equivalence problems on individual spacetimes. While
this provides only a minor simplification in the $1$-categorical context,
there is much more to gain from the reduction to individual spacetimes in the 
$\infty$-categorical context because the spacetime-wise 
homotopy-coherence questions associated with the homotopy time-slice axiom 
are considerably simpler than the ones on $\Loc^\rc$. See Section \ref{sec:spacetimewise}
for details on this crucial point.
\end{enumerate}

The advantage of our new proof strategy is that the reduction of
the global equivalence problem for AQFTs and tPFAs over $\Loc^\rc$
to a family of simpler spacetime-wise equivalence problems
carries over to the $\infty$-categorical context,
where both AQFTs and tPFAs satisfy the homotopy time-slice axiom.
This is the content of Theorems \ref{theo:dcQuillenequivalenceW}
and \ref{theo:homotopicalreduction}.
Therefore, one is left with the still challenging task of solving
the spacetime-wise $\infty$-categorical equivalence problems.
This task simplifies further due to the fact that,
in the spacetime-wise setting, the homotopy time-slice axiom
for AQFTs turns out \textit{not} to be richer than the strict one,
see Theorem \ref{theo:O_MCLF}. Taking these simplifications into account, 
there remains the Open Problem \ref{open:openproblem} 
of spacetime-wise comparing the $\infty$-category of AQFTs 
satisfying the strict time-slice axiom and the $\infty$-category 
of tPFAs satisfying the homotopy time-slice axiom. 
We hope to answer the remaining Open Problem 
\ref{open:openproblem} in future work.
\sk

We will now explain our results in more detail by outlining the content of this paper.
In Section \ref{sec:prelim}, we collect some relevant preliminaries
about AQFTs and tPFAs. To obtain a structural
implementation of the key features of the additivity property from \cite{BPScomparison},
we consider AQFTs and tPFAs which are defined over the wide subcategory
$\Loc^\rc\subseteq \Loc$ from Definition \ref{def:Locrc} whose spacetime 
morphisms $f:M\to N$ are either Cauchy or relatively compact.
(See Remarks \ref{rem:additivity} and \ref{rem:additivitytPFA} for a precise comparison
to the additivity property.) We review the concept of a Haag-Kastler $2$-functor $\HK$
and its category of points $\HK(\pt)$ from \cite{BGSHaagKastler}. 
The latter is relevant for us mainly because 
there exists a functor $\dc : \AQFT^\rc \to \HK(\pt)$ which decomposes
an AQFT over $\Loc^\rc$ into a compatible family of AQFTs over individual spacetimes,
as well as a functor $\as : \HK(\pt)\to \AQFT^\rc$ which assembles a compatible
family of AQFTs over individual spacetimes into an AQFT over $\Loc^\rc$.
These two functors are quasi-inverse to each other (see Theorem \ref{theo:dcasAQFT}),
hence they allow us to break complicated problems over $\Loc^\rc$ into 
families of simpler problems over individual spacetimes.
We introduce a similar concept for tPFAs, which we call
a Costello-Gwilliam $2$-functor $\CG$. Also in this case there exist
a decomposition functor $\dc:\tPFA^\rc \to \CG(\pt)$ and a quasi-inverse
assembly functor $\as : \CG(\pt)\to \tPFA^\rc$ which 
relate tPFAs over $\Loc^\rc$ and compatible families of tPFAs over individual spacetimes,
see Theorem \ref{theo:dcastPFA}.
\sk

In Section \ref{sec:1categorical}, we present a new $1$-categorical equivalence theorem
between AQFTs and tPFAs which is not only more general than the previous one in \cite{BPScomparison},
but also better transferable to the $\infty$-categorical context. Our key ingredient
for comparing AQFTs with tPFAs is given by the tPFA/AQFT-comparison multifunctor
$\Phi : \tP_{\Loc^\rc}\to \O_{\ovr{\Loc^\rc}}$ from Definition \ref{def:comparisonoperadmorphism}
which compares these two concepts at the fundamental level of their underlying operads.
This induces a comparison functor $\Phi^\ast : \AQFT^{\rc,W}\to\tPFA^{\rc,W}$ between the categories
of AQFTs and tPFAs over $\Loc^\rc$ satisfying the time-slice axiom. We prove that this functor
is an equivalence of categories, hence it provides an equivalence theorem between
AQFTs and tPFAs, by implementing the following new proof strategy: In Lemma \ref{lem:comparisondiagram}, 
we use the decomposition functors to the categories of points of the Haag-Kastler and Costello-Gwilliam 
$2$-functors to reduce the global equivalence problem on $\Loc^\rc$ 
to a family of simpler spacetime-wise equivalence problems. The latter are solved
in Theorem \ref{theo:equivalenceM} by proving a spacetime-wise equivalence theorem between AQFTs and tPFAs, 
which is a simple adaption of the results in \cite{BPScomparison}. 
In Theorem \ref{theo:1categoricalequivalence}, we combine these two steps in order to
prove a new $1$-categorical equivalence theorem between 
AQFTs and tPFAs over $\Loc^\rc$ satisfying the time-slice axiom.
\sk

In Section \ref{sec:modelcategorical}, we generalize this new proof strategy 
to the $\infty$-categorical context of AQFTs and tPFAs taking values in the 
symmetric monoidal $\infty$-category $\Ch$ of cochain complexes. 
Whenever possible, we work with explicit presentations of the $\infty$-categories
of $\Ch$-valued AQFTs and tPFAs (satisfying the homotopy time-slice axiom)
in terms of model categories or, more generally, semi-model categories obtained via left Bousfield localizations.
We review these (semi-)model categorical presentations in Subsection \ref{subsec:modelstructures}.
In Subsection \ref{subsec:homotopicalpoints}, we extend Barwick's results \cite{Barwick}
about homotopy limits of right Quillen diagrams of (semi-)model categories, which allow us to define
homotopical refinements of the categories of points of the Haag-Kastler and Costello-Gwilliam $2$-functors.
In Subsection \ref{subsec:homotopicaldcas}, we develop homotopical generalizations of
the decomposition and assembly functors and prove that they induce right Quillen equivalences
between the semi-model categories of AQFTs/tPFAs over $\Loc^\rc$  (satisfying the homotopy time-slice axiom)
and the semi-model categories of homotopical points, see in particular Corollary 
\ref{cor:dcQuillenequivalence} and Theorem \ref{theo:dcQuillenequivalenceW}. 
The main result of this section is Theorem \ref{theo:homotopicalreduction},
which is an $\infty$-categorical generalization of Lemma \ref{lem:comparisondiagram}.
It implies that the $\infty$-categorical equivalence problem for $\Ch$-valued
AQFTs and tPFAs over $\Loc^\rc$ satisfying the homotopy time-slice axiom
can be reduced to a family of simpler spacetime-wise $\infty$-categorical equivalence problems.
\sk

In Section \ref{sec:spacetimewise}, we present some non-trivial progress towards
solving the remaining spacetime-wise $\infty$-categorical equivalence problem for 
$\Ch$-valued AQFTs and tPFAs satisfying the homotopy time-slice axiom. 
In Subsection \ref{subsec:simplification}, we prove that
in this spacetime-wise context the homotopy 
time-slice axiom for AQFTs can be strictified (see Theorem \ref{theo:O_MCLF}), 
which yields a simplification of the problem as stated in Corollary \ref{cor:simplerspacetimewise}.
Despite the remarkable simplifications arising from our reduction to individual
spacetimes (Theorem \ref{theo:homotopicalreduction}) and the spacetime-wise AQFT
strictification result (Theorem \ref{theo:O_MCLF}),
the remaining Open Problem \ref{open:openproblem} is still challenging.
In Subsection \ref{subsec:openproblem}, we discuss some possible strategies
to attack it through techniques from homotopical algebra, homotopical
localizations of operads or $\infty$-categorical localizations. The latter approach
seems to be the most promising one, but unfortunately all existing detection criteria
for $\infty$-localizations are inconclusive in our case, see 
Proposition \ref{prop:HinichCriterion} and Remark \ref{rem:othercriteria}.
The way how these criteria fail seems to be rather mild (see Remark \ref{rem:outlook}),
which makes it plausible that there might be slight variations of the criteria
in \cite[Key Lemma 1.3.6]{HinichDK} and \cite{HinichDKnew} that apply to our example.
\sk

The paper includes two appendices. Appendix \ref{app:operadiclocalization} introduces
and studies a concept of operadic calculus of left fractions which is used for proving our
strictification Theorem \ref{theo:O_MCLF}. Appendix \ref{app:Lorentz} 
establishes some technical results of Lorentzian geometric nature which are needed
in the proofs of Theorem \ref{theo:O_MCLF} and Proposition \ref{prop:O_Mlocalized}.

%%%%%%%%%%%%%%%%%%%%%%%%%%%%%%%%%%%%%%%%%%%%%%%%
%%%%%%%%%%%%%%%%%%%%%%%%%%%%%%%%%%%%%%%%%%%%%%%%

\section{\label{sec:prelim}Preliminaries}
In this section we recall some relevant definitions
and constructions for algebraic quantum field theories (AQFTs) \cite{BFV,FewsterVerch,BSWoperad}
and prefactorization algebras \cite{CG1,CG2} which will be essential
for our work. Our focus will be on quantum field theories
that are defined on globally hyperbolic Lorentzian manifolds and 
satisfy the time-slice axiom. 
The following definition collects some basic concepts
from Lorentzian geometry which are needed below, see also
\cite{ONeill,BGP,Minguzzi} for more complete introductions to this subject.
\begin{defi}\label{def:Lorentz}
We denote by $\Loc$ the category whose objects are all oriented
and time-oriented globally hyperbolic Lorentzian manifolds $M$ (of a fixed dimension $m\geq 1$)
and whose morphisms are all orientation and time-orientation preserving isometric
open embeddings $f:M\to N$ with causally convex 
image $f(M)\subseteq N$.\footnote{The category $\Loc$ is not small, but it is equivalent
to a small category. This follows as usual by using Whitney's embedding theorem to
realize (up to diffeomorphism) all $m$-dimensional manifolds $M$ as 
submanifolds of $\bbR^{2m+1}$. To avoid size issues, 
we will always implicitly replace $\Loc$ by an equivalent 
small category, which we denote with abuse of notation also by $\Loc$.} 
Causal convexity means that every causal curve $\gamma : [0,1]\to N$ whose endpoints
$\gamma(0),\gamma(1)\in f(M)$ lie in the image is contained entirely in $f(M)\subseteq N$.
The following distinguished (tuples of) $\Loc$-morphisms will play a prominent role:
\begin{itemize}
\item[(a)] A $\Loc$-morphism $f:M\to N$ is called \textit{Cauchy}
if its image $f(M)\subseteq N$ contains a Cauchy surface of $N$.

\item[(b)] A $\Loc$-morphism $f:M\to N$ is called \textit{relatively compact}
if its image $f(M)\subseteq N$ is relatively compact, i.e.\ the 
closure $\overline{f(M)}\subseteq N$ is a compact subset.

\item[(c)] A pair of $\Loc$-morphisms $(f_1:M_1\to N,f_2:M_2\to N)$ to a common target
$N\in\Loc$ is called \textit{causally disjoint}, denoted by $f_1\perp f_2$, 
if there exists no causal curve connecting the images $f_1(M_1)\subseteq N$ and $f_2(M_2)\subseteq N$.

\item[(d)] A tuple of $\Loc$-morphisms $\und{f} = (f_1: M_1\to N,\dots,f_n:M_n\to N)$
to a common target $N\in\Loc$ is called \textit{time-ordered} if $J^+_N(f_i(M_i))\cap f_j(M_j)=\varnothing$,
for all $i<j$. The causal future $J^+_N(f_i(M_i))\subseteq N$ is defined as
the union of $f_i(M_i)\subseteq N$ and the set of all points in $N$ which 
can be reached by future-pointing causal
curves emanating from $f_i(M_i)\subseteq N$. 
(In particular, the images of the morphisms 
of a time-ordered tuple are mutually disjoint.)

\item[(e)] A tuple of $\Loc$-morphisms $\und{f} = (f_1: M_1\to N,\dots,f_n:M_n\to N)$
to a common target $N\in\Loc$ is called \textit{time-orderable} if there exists
a permutation $\rho\in\Sigma_n$, called \textit{time-ordering permutation}, such that
the permuted tuple $\und{f}\rho = (f_{\rho(1)},\dots,f_{\rho(n)})$ is time-ordered. (In particular, the images of the morphisms 
of a time-orderable tuple are mutually disjoint.)
\end{itemize}
\end{defi}

\begin{rem}\label{rem:Cauchyextension}
An unpleasant feature of the category $\Loc$ is that Cauchy surfaces
$\Sigma \subset M$ \textit{do not} in general admit extensions
along $\Loc$-morphisms $f:M\to N$. A simple example which illustrates
this fact is given by the inclusion $\Loc$-morphism $\iota_U^V : U\to V$
associated with the following two diamond-shaped open subsets in the $2$-dimensional Minkowski spacetime:
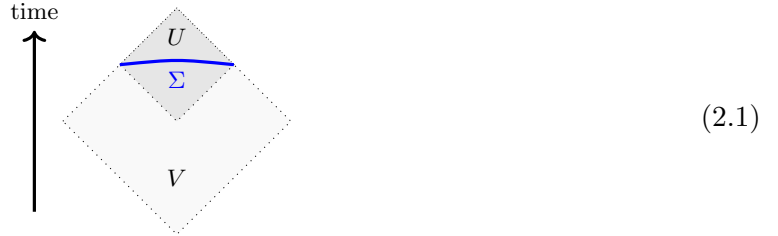
\begin{flalign}\label{eqn:nonrelativelycompact}
\begin{gathered}
\begin{tikzpicture}[scale=1.5]
\draw[fill=gray!5,dotted] (-1,0) -- (0,1) -- (1,0) -- (0,-1) -- (-1,0);
\draw[fill=gray!20,dotted] (-0.5,0.5) -- (0,1) -- (0.5,0.5) -- (0,0) -- (-0.5,0.5);
\draw (0,-0.5) node {{\footnotesize $V$}};
\draw (0,0.75) node {{\footnotesize $U$}}; 
\draw[very thick, blue] (-0.5,0.5) .. controls (0,0.55) .. (0.5, 0.5) node[midway, below] {{\footnotesize $\Sigma$}};
\draw[very thick, ->] (-1.25,-0.8) -- (-1.25,0.8) node[above] {{\footnotesize time}};
\end{tikzpicture}
\end{gathered}
\end{flalign}
It is evident that no Cauchy surface $\Sigma$ of $U$ admits an extension to a Cauchy surface of $V$.
\sk

It is important to note that this issue does not arise for $\Loc$-morphisms 
$f:M\to N$ which are either Cauchy or relatively compact in the sense of Definition 
\ref{def:Lorentz} (a) and (b). For the Cauchy case,
observe that the image $f(\Sigma)\subset N$ of any Cauchy surface $\Sigma\subset M$ 
under any Cauchy $\Loc$-morphism $f:M\to N$ is a Cauchy surface of $N$, hence the 
extension problem is solved trivially. For the relatively compact case,
observe that the closure of the image $\overline{f(\Sigma)}\subset N$ 
of any Cauchy surface $\Sigma\subset M$ under any relatively compact $\Loc$-morphism $f:M\to N$
is compact and achronal, hence by the results 
of Bernal and Sanchez \cite[Theorem 3.8]{BernalSanchez} 
it admits an extension $f(\Sigma)\subseteq \overline{f(\Sigma)}\subseteq \widetilde{\Sigma}
\subset N$ to a Cauchy surface $\widetilde{\Sigma}$ of $N$.
\end{rem}

The ability to extend Cauchy surfaces along $\Loc$-morphisms is crucial 
for proving an equivalence theorem between AQFTs and prefactorization algebras.
This has already been observed in the earlier work \cite{BPScomparison},
where a fix for the issue reported in Remark \ref{rem:Cauchyextension} 
is given by demanding a certain \textit{additivity property} 
for both AQFTs and prefactorization algebras. In the present work
we take a structural approach which consists of restricting 
the category $\Loc$ from Definition \ref{def:Lorentz} to 
a suitable subcategory that does not contain the problematic morphisms. 
This structural approach is more flexible and general than the one of \cite{BPScomparison}
(see Remarks \ref{rem:additivity} and \ref{rem:additivitytPFA} below) and it is better suited 
for our homotopical considerations in Sections \ref{sec:modelcategorical} and \ref{sec:spacetimewise}.
\begin{defi}\label{def:Locrc}
We denote by $\Loc^\rc\subseteq \Loc$ the wide subcategory 
consisting of all objects of $\Loc$ and of all 
$\Loc$-morphisms $f:M\to N$ which are either Cauchy or relatively compact
in the sense of Definition \ref{def:Lorentz} (a) and (b).
\end{defi}
\begin{rem}\label{rem:Locrccomposition}
Note that compositions of Cauchy and relatively compact
$\Loc$-morphisms follow the pattern
\begin{subequations}
\begin{flalign}
(\text{Cauchy})\circ (\text{Cauchy}) \,&=\, (\text{Cauchy})\quad,\\ 
(\text{Cauchy})\circ (\text{relatively compact}) \,&=\, (\text{relatively compact})\quad,\\
(\text{relatively compact})\circ (\text{Cauchy}) \,&=\, (\text{relatively compact})\quad,\\
(\text{relatively compact})\circ (\text{relatively compact}) \,&=\, (\text{relatively compact})\quad.
\end{flalign}
\end{subequations}
Since the identity $\Loc$-morphisms $\id_M : M\to M$ are Cauchy, for all $M\in \Loc$, 
it follows that $\Loc^\rc\subseteq \Loc$ forms indeed a subcategory.
\end{rem}

\subsection{\label{subsec:AQFT}Algebraic quantum field theories}
AQFTs admit an elegant and powerful description 
in terms of algebras over suitable operads. 
This observation originated in \cite{BSWoperad} and it 
was developed further in various directions in
\cite{BSWinvolutive,BSWhomotopy,BruinsmaSchenkel,Yau,BPSWcategorified,Carmona,BCStimeslice}.
See also \cite{BSreview,BSchapter} for brief and non-technical introductions.
The precise version of the AQFT operad we use in our work is 
based on the subcategory $\Loc^{\rc}\subseteq \Loc$ 
from Definition \ref{def:Locrc} and it is defined as follows.
\begin{defi}\label{def:AQFToperad}
The \textit{AQFT operad}\footnote{The overline notation
$\ovr{\Loc^\rc}$ refers to the concept of orthogonal categories introduced in \cite{BSWoperad}.
In the present case, the relevant orthogonal category $\ovr{\Loc^\rc}:= (\Loc^\rc,\perp)$ is
given by the category $\Loc^{\rc}$ from Definition \ref{def:Locrc} and the orthogonality relation
$\perp$ which is defined by causal disjointness from Definition \ref{def:Lorentz} (c).} 
$\O_{\ovr{\Loc^\rc}}^{}$ is the 
colored operad which is defined by the following data:
\begin{itemize}
\item[(i)] The objects are the objects of $\Loc^\rc$.

\item[(ii)] The set of operations from $\und{M} = (M_1,\dots, M_n)$
to $N$ is the quotient set
\begin{flalign}
\O_{\ovr{\Loc^\rc}}^{}\big(\substack{N \\ \und{M}}\big)\,:=\,
\bigg(\Sigma_n \times \prod_{i=1}^n \Loc^\rc(M_i,N)\bigg)\Big/\!\sim_{\perp}^{}\quad,
\end{flalign}
where $\Sigma_n$ denotes the permutation group on $n$ letters and
$\Loc^\rc(M_i,N)$ denotes the set of $\Loc^\rc$-morphisms from $M_i$ to $N$.
Two elements are equivalent $(\sigma,\und{f}) 
\sim_{\perp}^{} (\sigma^\prime,\und{f}^\prime)$ if and only if 
$\und{f} = \und{f}^\prime$ and the right permutation
$\sigma\sigma^{\prime -1} : \und{f}\sigma^{-1}\to \und{f}\sigma^{\prime -1}$
is generated by transpositions of adjacent causally disjoint pairs of morphisms.

\item[(iii)] The composition of $[\sigma,\und{f}] : \und{M}\to N$ 
with $[\sigma_i,\und{g}_i] : \und{K}_i\to M_i$, for $i=1,\dots,n$, is defined by
\begin{subequations}
\begin{flalign}
[\sigma , \und{f}]\, [ \und{\sigma}, \und{\und{g}}]\,:=\,
\big[\sigma(\sigma_1,\dots,\sigma_n), \und{f}\,\und{\und{g}}\,\big]\,:\,\und{\und{K}}\,\longrightarrow\,N\quad,
\end{flalign}
where $\sigma(\sigma_1,\dots,\sigma_n)$ denotes
the composition in the unital associative operad and
\begin{flalign}
 \und{f}\,\und{\und{g}} \,:=\, \big(f_1\,g_{11},\dots, f_1\,g_{1k_1},\dots, f_n\,g_{n1},\dots,f_{n}\,g_{n k_n}\big)
\end{flalign}
\end{subequations}
is given by compositions in the category $\Loc^\rc$.

\item[(iv)] The identity operations are $[e,\id_N^{}] : N\to N$, where $e\in\Sigma_1$ is the identity permutation.

\item[(v)] The permutation action of $\sigma^\prime\in\Sigma_n$ on $[\sigma,\und{f}] : \und{M}\to N$ is given by
\begin{flalign}
[\sigma,\und{f}]\cdot \sigma^\prime \,:=\, [\sigma\sigma^\prime,\und{f}\sigma^\prime] \,:\, \und{M}\sigma^{\prime}~\longrightarrow~
 N\quad,
\end{flalign}
where $\und{f}\sigma^\prime = (f_{\sigma^\prime(1)},\dots,f_{\sigma^\prime(n)})$ 
and $\und{M}\sigma^{\prime}= (M_{\sigma^\prime(1)},\dots, M_{\sigma^{\prime}(n)})$ denote 
the permuted tuples and $\sigma\sigma^\prime$ is given by the group operation of the permutation group $\Sigma_n$.
\end{itemize}
\end{defi}

Let us fix any bicomplete closed symmetric monoidal category $\TT$ as the target category 
in which our AQFTs take values. The typical choice in the context of 
$1$-categorical AQFTs is given by the bicomplete closed symmetric monoidal 
category $\Vec_\bbC$ of complex vector spaces.
\begin{defi}\label{def:AQFT}
The category of \textit{AQFTs over $\Loc^\rc$} with values in a
bicomplete closed symmetric monoidal category $\TT$ is defined as the category
\begin{flalign}
\AQFT^\rc\,:=\, \Alg_{\O_{\ovr{\Loc^\rc}}^{}}^{}\big(\TT\big)
\end{flalign}
of $\TT$-valued algebras over the AQFT operad $\O_{\ovr{\Loc^\rc}}^{}$
from Definition \ref{def:AQFToperad}. We denote by 
\begin{flalign}
\AQFT^{\rc,W}\,\subseteq\, \AQFT^\rc
\end{flalign}
the full subcategory consisting of all AQFTs satisfying the time-slice axiom,
i.e.\ the operad algebra (i.e.\ multifunctor)
$\AAA: \O_{\ovr{\Loc^\rc}}^{}\to\TT$ sends every $1$-ary operation $[e,f]:M\to N$ 
in $\O_{\ovr{\Loc^\rc}}^{}$ which corresponds to a Cauchy morphism $f:M\to N$ in $\Loc^\rc$
to an isomorphism in $\TT$.
\end{defi}
\begin{rem}\label{rem:AQFT}
This operadic definition of $\AQFT^\rc$ is equivalent to the more traditional categorical 
one in terms of functors $\AAA : \Loc^\rc \to \Alg_{\mathsf{uAs}}(\TT)$ from the category
$\Loc^\rc$ to the category of unital associative algebras 
in $\TT$ which satisfy the Einstein causality axiom.
The time-slice axiom from Definition \ref{def:AQFT} agrees with the usual time-slice axiom.
We refer the reader to \cite[Section 3]{BSWoperad} for more details on these points.
Whenever convenient, we will make use of these identifications without further comments.
\end{rem}

To relate our novel concept of AQFTs over $\Loc^\rc$ from Definition \ref{def:AQFT} 
with the more familiar concept of AQFTs over $\Loc$, let us recall that 
$\Loc^\rc \subseteq \Loc$ is a subcategory by Definition \ref{def:Locrc}. 
Therefore, there exists an evident multifunctor $\O_{\ovr{\Loc^\rc}}\to \O_{\ovr{\Loc}}$.
The corresponding functor 
\begin{flalign}\label{eqn:res-rc}
(-)\vert^\rc \,:\, \AQFT \,:=\, \Alg_{\O_{\ovr{\Loc}}}\big(\TT\big) ~\longrightarrow~ \AQFT^\rc
\end{flalign}
restricting AQFTs over $\Loc$ to AQFTs over $\Loc^\rc$ is manifestly 
faithful because $\Loc^\rc$ and $\Loc$ have the same objects. 
We shall now show that the functor \eqref{eqn:res-rc} acts fully 
faithfully on the full subcategory $\AQFT^\mathrm{add} \subseteq \AQFT$ 
of \textit{additive} AQFTs over $\Loc$, 
consisting of all objects $\AAA \in \AQFT$ 
satisfying the following additivity property, 
which we recall from \cite{BPScomparison}: 
For every $M\in\Loc$, the canonical map
\begin{flalign}\label{eqn:additivity}
\xymatrix{
\colim\Big(\AAA\vert_M : \mathbf{RC}_M\to \Alg_{\mathsf{uAs}}(\TT)\Big)~\ar[r]^-{\cong} &~\AAA(M)
}
\end{flalign}
is an isomorphism, where $\mathbf{RC}_M$ denotes the directed set
of all relatively compact and causally convex opens $U\subseteq M$ 
with the evident partial order given by subset inclusion.
The interpretation of the additivity property is that the global
value $\AAA(M)$ of the AQFT $\AAA$ on $M$ is completely determined by its local
values $\AAA(U)$ in all relatively compact and causally convex opens $U\subseteq M$.
\begin{propo}\label{prop:additivity}
The restriction
\begin{flalign}\label{eqn:add-res-rc}
\xymatrix{
(-)\vert^\rc\,:\,\AQFT^\mathrm{add}\ar[r]^-{\mathrm{f.f.}}~&~\AQFT^\rc
}
\end{flalign}
of the functor \eqref{eqn:res-rc} to the full subcategory 
$\AQFT^\mathrm{add} \subseteq \AQFT$
of additive AQFTs over $\Loc$ is fully faithful.
\end{propo}
\begin{proof}
As we already observed, the functor \eqref{eqn:res-rc} is faithful because 
$\Loc^\rc$ and $\Loc$ have the same objects. In particular, its restriction 
\eqref{eqn:add-res-rc} is faithful too. In order to prove that it is also full, 
given any $\AAA,\BBB\in \AQFT^\mathrm{add}$ and any morphism 
$\zeta : \AAA\vert^{\rc} \Rightarrow \BBB\vert^{\rc}$ in $\AQFT^\rc$
between the restrictions $\AAA\vert^{\rc} ,\BBB\vert^{\rc}\in \AQFT^\rc$, we have to construct
a morphism $\widetilde{\zeta}:\AAA\Rightarrow\BBB$ in $\AQFT^\mathrm{add}$
such that its restriction satisfies $\widetilde{\zeta}\vert^\rc = \zeta$.
Using again that $\Loc^\rc$ and $\Loc$ have the same objects, a candidate
for $\widetilde{\zeta}$ is given by the components $\widetilde{\zeta}_M := \zeta_M$, for
all $M\in\Loc$. It remains to show naturality of these components
for all $\Loc$-morphisms $f: M\to N$ (not necessarily Cauchy or relatively compact).
For any $U\in \mathbf{RC}_M$, we can expand the naturality square for $f$ according to
\begin{flalign}\label{eqn:extendednaturality}
\begin{gathered}
\xymatrix{
\ar[d]_-{\zeta_U}
\AAA(U)\ar[r]^-{\AAA(\iota_U^M)}~&~\ar[d]_-{\zeta_M} \AAA(M) \ar[r]^-{\AAA(f)}~&~\AAA(N)\ar[d]_-{\zeta_N}\\
\BBB(U)\ar[r]_-{\BBB(\iota_U^M)}~&~\BBB(M)\ar[r]_-{\BBB(f)}~&~\BBB(N)
}
\end{gathered}\qquad.
\end{flalign}
Observe that the left and outer squares commute because the inclusion
$\iota_U^M: U\to M$ and the composite $f\,\iota_U^M : U\to N$ are both relatively compact, 
i.e.\ they are morphisms in $\Loc^\rc\subseteq \Loc$. By universality of the 
colimit entering the additivity property \eqref{eqn:additivity}, this implies that the 
naturality square for $f$ commutes too.
\end{proof}

We shall now characterize the essential image of the fully faithful 
functor \eqref{eqn:add-res-rc}, and thus single out precisely those 
AQFTs over $\Loc^\rc$ which correspond to additive AQFTs over $\Loc$. 
For this purpose, in analogy with the case of AQFTs over $\Loc$, 
let us introduce the full subcategory
$\AQFT^{\rc,\mathrm{add}} \subseteq \AQFT^\rc$ of \textit{additive} AQFTs 
over $\Loc^\rc$, consisting of all objects $\AAA \in \AQFT^\rc$ 
satisfying the following additivity property: 
For every $M\in\Loc^\rc$, the canonical map
\begin{flalign}\label{eqn:additivity-rc}
\xymatrix{
\colim\Big(\AAA\vert_M : \mathbf{RC}^\rc_M\to \Alg_{\mathsf{uAs}}(\TT)\Big)~\ar[r]^-{\cong} &~\AAA(M)
}
\end{flalign}
is an isomorphism, where $\mathbf{RC}^\rc_M$ denotes the directed set
of all relatively compact and causally convex opens $U\subseteq M$ 
with the partial order given by subset inclusions 
that are either Cauchy or relatively compact.
We observe immediately that the directed sets $\mathbf{RC}^\rc_M$ and 
$\mathbf{RC}_M$ have the same underlying set and, moreover, the obvious 
order-preserving map $\mathbf{RC}^\rc_M \to \mathbf{RC}_M$ is a final functor, 
because all pairs of objects $U_1,U_2 \in \mathbf{RC}_M$ 
are included as relatively compact subsets of some other object 
$V \in \mathbf{RC}_M$, i.e.\ the inclusions $U_1 \subseteq V \supseteq U_2$ 
are relatively compact. In particular, the additivity property for AQFTs over $\Loc$ 
may be detected equivalently by restricting the colimit \eqref{eqn:additivity} 
along the final functor $\mathbf{RC}^\rc_M \to \mathbf{RC}_M$. 
This entails that the functor \eqref{eqn:add-res-rc} factors through 
the full subcategory inclusion $\AQFT^{\rc,\mathrm{add}} \subseteq \AQFT^\rc$. 
\begin{propo}\label{prop:additivity2}
The corestriction
\begin{flalign}\label{eqn:add-res-rcadd}
\xymatrix{
(-)\vert^\rc\,:\,\AQFT^\mathrm{add}\ar[r]^-{\sim}~&~\AQFT^{\rc,\mathrm{add}}
}
\end{flalign}
of the fully faithful functor \eqref{eqn:add-res-rc} to the full subcategory 
$\AQFT^{\rc,\mathrm{add}} \subseteq \AQFT^\rc$
of additive AQFTs over $\Loc^\rc$ is an equivalence of categories. 
In other words, $\AQFT^{\rc,\mathrm{add}}$ is the essential image of 
the functor \eqref{eqn:add-res-rc}
\end{propo}
\begin{proof}
The functor \eqref{eqn:add-res-rcadd} is fully faithful because it is defined 
as the corestriction of a fully faithful functor along 
a full subcategory inclusion. To prove that it is also essentially surjective, 
and hence an equivalence of categories, given an object
$\AAA\in \AQFT^{\rc,\mathrm{add}}$, we construct an object 
$\widetilde{\AAA} \in \AQFT^{\mathrm{add}}$ and an isomorphism 
$\widetilde{\AAA} \vert^\rc \cong \AAA$ in $\AQFT^{\rc,\mathrm{add}}$, 
which we interpret as witnessing the fact that $\widetilde{\AAA}$ 
extends $\AAA$ to morphisms in $\Loc$ beyond $\Loc^\rc$, 
namely neither Cauchy nor relatively compact.
Since $\Loc$ and $\Loc^\rc$ have the same objects,
we set $\widetilde{\AAA}(M) := \AAA(M)$, for all $M \in \Loc$. Furthermore, since the morphisms 
$f: M \to N$ in $\Loc$ that are either Cauchy or relatively compact are 
precisely the morphisms of $\Loc^\rc$, for those we set 
$\widetilde{\AAA}(f) := \AAA(f)$. To define also the morphisms 
$\widetilde{\AAA}(f): \widetilde{\AAA}(M) \to \widetilde{\AAA}(N)$ 
in $\Alg_{\mathsf{uAs}}(\TT)$, for $f: M \to N$ in $\Loc$ that is neither 
Cauchy nor relatively compact, we observe that, for any object 
$U \in \mathbf{RC}^\rc_M$, one has the evident $f$-induced isomorphism 
$f\vert_U^{f(U)}: U \to f(U)$ in $\Loc^\rc$ and the morphism 
$\iota_{f(U)}^N: f(U) \to N$ in $\Loc^\rc$ that embeds the relatively compact 
and causally convex open $f(U) \subseteq N$ in $N$. Therefore, 
using the $\Loc^\rc$-functoriality of $\AAA$, one constructs the morphism
\begin{flalign}
\xymatrix@C=4em{
\widetilde{\AAA}(U) \,=\, \AAA(U) \ar[r]^-{\AAA(f\vert_U^{f(U)})} ~&~ 
\AAA(f(U)) \ar[r]^-{\AAA(\iota_{f(U)}^N)} ~&~ \AAA(N) \,=\, \widetilde{\AAA}(N)
}
\end{flalign}
in $\Alg_{\mathsf{uAs}}(\TT)$, which is furthermore natural with respect to 
the relatively compact inclusions $\iota_U^V: U \to V$ in $\mathbf{RC}^\rc_M$. 
Hence, using the additivity property of $\AAA \in \AQFT^{\rc,\mathrm{add}}$ 
and the universal property of the colimit in \eqref{eqn:additivity-rc}, 
one defines the morphism 
$\widetilde{\AAA}(f): \widetilde{\AAA}(M) \to \widetilde{\AAA}(N)$ 
in $\Alg_{\mathsf{uAs}}(\TT)$ also for $f: M \to N$ in $\Loc$ that is neither 
Cauchy nor relatively compact. Using again the additivity property of $\AAA$ 
and the universal property of the colimit, one checks that 
$\widetilde{\AAA}: \Loc \to \Alg_{\mathsf{uAs}}(\TT)$ is a functor, 
that furthermore inherits both the Einstein causality axiom 
and the additivity property from $\AAA$, hence 
$\widetilde{\AAA} \in \AQFT^\mathrm{add}$ is an additive AQFT over $\Loc$. 
By construction the restriction of $\widetilde{\AAA}$ to $\Loc^\rc$ 
coincides with $\AAA$, namely we can take the identity 
as the isomorphism $\widetilde{\AAA} \vert^\rc \cong \AAA$ 
in $\AQFT^{\rc,\mathrm{add}}$ witnessing that $\widetilde{\AAA}$ extends 
$\AAA$ to morphisms in $\Loc$ that are neither Cauchy nor relatively compact.
\end{proof}

\begin{rem}\label{rem:additivity}
Proposition \ref{prop:additivity} shows that our novel concept of AQFTs 
over $\Loc^\rc$ from Definition \ref{def:AQFT}
is a honest generalization of the additive AQFTs over $\Loc$ from the 
earlier comparison theorem in \cite{BPScomparison}. Furthermore, 
Proposition \ref{prop:additivity2} identifies the 
AQFTs over $\Loc^\rc$ coming from additive AQFTs over $\Loc$ by a 
similar additivity property. 
The proof of Proposition \ref{prop:additivity2} also constructs, for every additive
AQFT $\AAA$ over $\Loc^\rc$, an explicit extension 
$\widetilde{\AAA} \in \AQFT^{\mathrm{add}}$
of $\AAA$ to the additional morphisms in $\Loc \supseteq \Loc^\rc$ 
that are neither Cauchy nor relatively compact.
\end{rem}

An essential technique which we will use frequently in the main 
part of this paper is the decomposition of an AQFT over $\Loc^\rc$
into a compatible family of AQFTs over individual spacetimes 
and, vice versa, the assembly of a compatible family 
of AQFTs over individual spacetimes into an AQFT over $\Loc^\rc$.
These constructions were developed in \cite{BGSHaagKastler}
under the name of Haag-Kastler $2$-functors and one of 
their main benefits is that they allow us to break complicated
global problems on $\Loc$ (or $\Loc^\rc$ in the present case)
into families of simpler problems on individual spacetimes $M$.
We will now recall those aspects of \cite{BGSHaagKastler} which 
are relevant for the present work and at the same time adapt them to
our context given by the category $\Loc^\rc$ from Definition \ref{def:Locrc}.
\begin{defi}\label{def:O_M}
For each $M\in\Loc^\rc$, we denote by
\begin{flalign}\label{eqn:O_M}
\O_M\,\subseteq\,\O_{\ovr{\Loc^\rc}}
\end{flalign}
the suboperad of the AQFT operad from Definition \ref{def:AQFToperad}
whose objects are all $U\in \O_{\ovr{\Loc^\rc}}$
corresponding to causally convex opens $U\subseteq M$ which
are either Cauchy or relatively compact and whose operations
are of the form $[\sigma,\iota_{\und{U}}^V] :\und{U}\to V$,
where $\iota_{\und{U}}^V = (\iota_{U_1}^V,\dots,\iota_{U_n}^V)$ is a tuple
of inclusion $\Loc^\rc$-morphisms. The category of \textit{AQFTs 
over $\Loc^\rc/M$} is defined as the category
\begin{flalign}
\HK(M)\,:=\,\Alg_{\O_M}\big(\TT\big)
\end{flalign}
of $\TT$-valued algebras over this suboperad. We further denote by
\begin{flalign}
\HK^W(M)\,\subseteq\, \HK(M)
\end{flalign}
the full subcategory of all AQFTs over $\Loc^\rc/M$ satisfying the time-slice axiom.
\end{defi}

For any $\Loc^\rc$-morphism $f : M\to N$, there exists 
an evident multifunctor
(denoted with abuse of notation by the same symbol)
$f:\O_M\to \O_N$ which is given by taking images under $f$. This induces a pullback functor
$f^\ast : \HK^{(W)}(N)\to \HK^{(W)}(M)$ between the corresponding 
categories of AQFTs (satisfying the time-slice axiom) over the individual spacetimes. 
These data assemble into a \emph{strict} $2$-functor 
(see e.g.\ \cite[Definition 1.4.1]{Hovey}), called \textit{Haag-Kastler $2$-functor}
\begin{flalign}\label{eqn:HK2functor}
\HK^{(W)}\,:\,\big(\Loc^\rc\big)^\op~\longrightarrow~\CAT\quad,
\end{flalign}
which assigns to each $M\in\Loc^\rc$ the corresponding
category $\HK^{(W)}(M)\in \CAT$ of AQFTs over $\Loc^\rc/M$ (satisfying the time-slice axiom)
and to each $\Loc^\rc$-morphism $f:M\to N$ the pullback functor
$\HK^{(W)}(f) := f^\ast : \HK^{(W)}(N)\to \HK^{(W)}(M)$. 
Of course, here we consider $\Loc^\rc$ with its obvious $2$-category structure, 
where the only $2$-morphisms are identities.
\begin{defi}\label{def:HKpoints}
The \textit{category of points} of the Haag-Kastler $2$-functor
$\HK^{(W)}$ is defined as the bicategorical limit
\begin{flalign}
\HK^{(W)}(\pt)\,:=\, \mathrm{bilim}\Big(\HK^{(W)}:\big(\Loc^\rc\big)^\op \to \CAT\Big)\,\in\,\CAT\quad.
\end{flalign}
\end{defi}
\begin{rem}\label{rem:HKpoints}
A more geometric description of the category of points $\HK^{(W)}(\pt)$ 
is as the category of pseudo-natural transformations and modifications 
from the constant strict $2$-functor 
$\Delta \mathbf{1}\colon \big(\Loc^\rc\big)^\op \to \CAT $ 
on the terminal category $\mathbf{1}$ (the point) to the Haag-Kastler 
$2$-functor $\HK$, see \cite[Definition 3.4]{BGSHaagKastler}. 
The notation comes from the fact that $\HK^{(W)}(M)$ can be identified with 
the category of pseudo-natural transformations $y(M) \to \HK$ and modifications, 
where $y(M)\colon \big(\Loc^\rc\big)^\op \to \Set\to \CAT$ denotes the Yoneda 
functor associated to $M\in \Loc^{\rc}$. Doing so, one interprets 
$\HK^{(W)}(M)$ as the category of $M$-points of the Haag-Kastler $2$-functor 
$\HK^{(W)}$ and, in the same spirit, $\HK^{(W)}(\pt)$ as its category of points.
\sk

As explained in  \cite[Remark 3.5]{BGSHaagKastler}, 
the category of points $\HK^{(W)}(\pt)$ also admits the following explicit description
in terms of compatible families of AQFTs:
\begin{itemize}
\item An object in $\HK^{(W)}(\pt)$ is a tuple $\big(\{\AAA_M\},\{\alpha_f\}\big)$
consisting of a family of AQFTs (satisfying the time-slice axiom) over individual spacetimes
$\AAA_M\in\HK^{(W)}(M)$, for all $M\in\Loc^\rc$, and a family of $\HK^{(W)}(M)$-isomorphisms
$\alpha_f : \AAA_M\Rightarrow f^\ast(\AAA_N)$, for all $\Loc^\rc$-morphisms $f:M\to N$,
such that the diagrams
\begin{flalign}
\begin{gathered}
\xymatrix@C=3em{
\ar@{=>}[d]_-{\alpha_{gf}}\AAA_M\ar@{=>}[r]^-{\alpha_f}~&~f^\ast(\AAA_N)\ar@{=>}[d]^-{f^\ast(\alpha_g)}  
~&~\AAA_M \ar@{=>}[r]^-{\alpha_{\id_M}}\ar@{=>}[rd]_-{\id_{\AAA_M}}~&~\id_M^\ast(\AAA_M)\ar@{=}[d]\\
(gf)^\ast(\AAA_L)\ar@{=}[r] ~&~ f^\ast g^\ast(\AAA_L)  ~&~ ~&~\AAA_M
}
\end{gathered}
\end{flalign}
in $\HK^{(W)}(M)$ commute, for all composable $\Loc^\rc$-morphisms 
$f:M\to N$ and $g: N\to L$ and all objects $M\in \Loc^\rc$.

\item A morphism $\{\zeta_M\}: \big(\{\AAA_M\},\{\alpha_f\}\big)\Rightarrow \big(\{\BBB_M\},\{\beta_f\}\big)$
in $\HK^{(W)}(\pt)$ is a family of $\HK^{(W)}(M)$-morphisms $\zeta_M : \AAA_M\Rightarrow\BBB_M$,
for all $M\in\Loc^\rc$, such that the diagram
\begin{flalign}
\begin{gathered}
\xymatrix@C=3em{
\ar@{=>}[d]_-{\zeta_M}\AAA_M \ar@{=>}[r]^-{\alpha_f}~&~f^\ast(\AAA_N)\ar@{=>}[d]^-{f^\ast(\zeta_N)}\\
\BBB_M \ar@{=>}[r]_-{\beta_f}~&~f^\ast(\BBB_N)
}
\end{gathered}
\end{flalign}
in $\HK^{(W)}(M)$ commutes, for all $\Loc^\rc$-morphisms $f:M\to N$. \qedhere
\end{itemize}
\end{rem}

There exists an equivalence between the category of points $\HK^{(W)}(\pt)$ and 
the category  $\AQFT^{\rc(,W)}$ of AQFTs over $\Loc^\rc$ from Definition \ref{def:AQFT} 
which is exhibited by the decomposition and assembly functors from
\cite[Constructions 3.6 and 3.7]{BGSHaagKastler}. For completeness, 
let us briefly sketch their descriptions. 
The \textit{decomposition functor}
\begin{subequations}\label{eqn:dc}
\begin{flalign}\label{eqn:dc1}
\dc\,:\,\AQFT^{\rc} ~\longrightarrow~\HK(\pt)
\end{flalign}
assigns to an object $\AAA\in\AQFT^\rc$ the object
in $\HK(\pt)$ which is specified by the tuple
\begin{flalign} 
\dc(\AAA)\,:=\,\big(\big\{\AAA\vert_M\big\},\big\{\AAA\vert_M\Rightarrow f^\ast(\AAA\vert_N)\big\}\big)\quad,
\end{flalign}
where $\AAA\vert_M\in \HK(M)$ is defined by restricting $\AAA$ to the suboperad \eqref{eqn:O_M}
and the $\HK(M)$-isomorphism $\AAA\vert_M\Rightarrow f^\ast(\AAA\vert_N)$ is defined
by the components $\AAA(f\vert_U): \AAA(U)\to \AAA(f(U))$, for all $U\in\O_M$,
where $f\vert_U : U\to f(U)$ denotes the domain and codomain restriction of the 
$\Loc^\rc$-morphism $f:M\to N$. The action of the functor \eqref{eqn:dc1}
on an $\AQFT^\rc$-morphism $\zeta: \AAA\Rightarrow \BBB$ is defined by restriction
\begin{flalign}
\dc(\zeta)\,:=\,\big\{\zeta\vert_M\big\}\,:\,\dc(\AAA)~\Longrightarrow~\dc(\BBB)
\end{flalign}
\end{subequations}
to the suboperads \eqref{eqn:O_M}.
The \textit{assembly functor}
\begin{subequations}\label{eqn:as}
\begin{flalign}\label{eqn:as1}
\as\,:\,\HK(\pt)~\longrightarrow~\AQFT^\rc
\end{flalign}
assigns to an object $\big(\{\AAA_M\},\{\alpha_f\}\big)\in\HK(\pt)$
the object in $\AQFT^\rc$ which is defined by the functor (recall the 
equivalence from Remark \ref{rem:AQFT})
\begin{flalign}
\as\big(\{\AAA_M\},\{\alpha_f\}\big)\,:\,\Loc^\rc~&\longrightarrow~ \Alg_{\mathsf{uAs}}(\TT)\quad,\\
\nn M~&\longmapsto~\AAA_M(M)\quad,\\
\nn \big(f:M\to N\big)~&\longmapsto~\Big(\xymatrix@C=3.5em{
\AAA_M(M)\ar[r]^-{(\alpha_f)_M}~&~\AAA_N(f(M))\ar[r]^-{\AAA_N(\iota_{f(M)}^N)}~&~\AAA_N(N)
}\Big)\quad.
\end{flalign}
The Einstein causality axiom for this functor
can be verified as in \cite[Construction 3.7]{BGSHaagKastler}.
The action of the functor \eqref{eqn:as1} on a
$\HK(\pt)$-morphism $\{\zeta_M\}:  \big(\{\AAA_M\},\{\alpha_f\}\big)\Rightarrow
\big(\{\BBB_M\},\{\beta_f\}\big)$ is given by the natural transformation 
which is defined by the components
\begin{flalign}
\as\big(\{\zeta_M\}\big)_M\,:=\,\Big(\xymatrix@C=3.5em{
\AAA_M(M)\ar[r]^-{(\zeta_M)_M}~&~\BBB_M(M)
}\Big)\quad,
\end{flalign}
\end{subequations}
for all $M\in\Loc^\rc$.
\begin{theo}\label{theo:dcasAQFT}
The decomposition \eqref{eqn:dc} and assembly \eqref{eqn:as} functors
are quasi-inverse to each other, hence they exhibit an equivalence of 
categories
\begin{flalign}
\AQFT^\rc\,\simeq\,\HK(\pt)\quad.
\end{flalign}
Furthermore, these functors preserve the respective time-slice axioms
and hence they induce also an equivalence of categories
\begin{flalign}
\AQFT^{\rc,W}\,\simeq\,\HK^W(\pt)\quad.
\end{flalign}
\end{theo}
\begin{proof}
One directly verifies that the composition $\as \circ \dc =\id $ is 
equal to the identity functor.
A natural isomorphism $\dc\circ \as\cong \id$ for the other composition
is constructed in \cite[Theorem 3.8]{BGSHaagKastler}.
Explicitly, given any object $\big(\{\AAA_M\},\{\alpha_f\}\big)\in\HK(\pt)$,
the $\HK(\pt)$-isomorphism
 $\dc\big(\as\big(\{\AAA_M\},\{\alpha_f\}\big)\big)\Rightarrow \big(\{\AAA_M\},\{\alpha_f\}\big)$
is defined by the components $(\alpha_{\iota_{U}^{M}})_U : \AAA_U(U)\to \AAA_M(U)$,
for all $U\in\O_M$.
The statement about the preservation of the time-slice axioms is straightforward to check.
\end{proof}

\subsection{\label{subsec:PFA}Prefactorization algebras}
Prefactorization algebras, as introduced by Costello and Gwilliam in 
\cite{CG1,CG2}, have a broader scope than AQFTs in the sense that they
can be defined in various geometric contexts, including topological, 
Riemannian and complex manifolds. In the context of globally hyperbolic
Lorentzian manifolds, which is the focus of our work, it was first observed
in \cite{GwilliamRejzner1} that the factorization products of a prefactorization 
algebra admit an interpretation in terms of the time-ordered products from
more traditional approaches to quantum field theory. This observation was sharpened
in \cite{BPScomparison} by proposing the concept of time-orderable tuples
(see also Definition \ref{def:Lorentz} (e)) as a Lorentzian geometric refinement 
of the tuples of mutually disjoint open subsets of a manifold used in \cite{CG1,CG2}.
The resulting Lorentzian geometric variant of prefactorization
algebras has been called \textit{time-orderable prefactorization algebras} 
in \cite{BPScomparison}. We will adopt this terminology in our present paper.
\sk

In this subsection we introduce some basic terminology and concepts
for time-orderable prefactorization algebras (tPFAs), following the 
structure of Subsection \ref{subsec:AQFT} for the case of AQFTs. 
The tPFAs in this work will be based
on the subcategory $\Loc^{\rc}\subseteq \Loc$ 
from Definition \ref{def:Locrc} and they are governed by the following operad.
\begin{defi}\label{def:tPFAoperad}
The \textit{tPFA operad} $\tP_{\Loc^\rc}$ is the 
colored operad which is defined by the following data:
\begin{itemize}
\item[(i)] The objects are the objects of $\Loc^\rc$.

\item[(ii)] The set of operations from $\und{M} = (M_1,\dots, M_n)$
to $N$ is the set
\begin{flalign}
\tP_{\Loc^\rc}^{}\big(\substack{N \\ \und{M}}\big)\,:=\,
\bigg\{\und{f}\in \prod_{i=1}^n\Loc^{\rc}(M_i,N)~\Big\vert~ \und{f} \text{ is time-orderable}\bigg\}
\end{flalign}
of all time-orderable tuples of $\Loc^\rc$-morphisms 
in the sense of Definition \ref{def:Lorentz} 
(e).\footnote{Our convention is that all empty tuples and all tuples of length $1$ are
time-orderable. Hence, the operad $\tP_{\Loc^\rc}$ contains for every object $N\in\Loc^\rc$
a unique arity $0$ operation $\varnothing \to N$ and to every $\Loc^\rc$-morphism
$f:M\to N$ is associated an arity $1$ operation which we will denote by the same symbol.}

\item[(iii)] The composition of $\und{f} : \und{M}\to N$ 
with $\und{g}_i : \und{K}_i\to M_i$, for $i=1,\dots,n$, is defined by
\begin{flalign}
 \und{f}\,\und{\und{g}} \,:=\, \big(f_1\,g_{11},\dots, f_1\,g_{1k_1},\dots, f_n\,g_{n1},\dots,f_{n}\,g_{n k_n}\big)\quad,
\end{flalign}
where the individual compositions are performed in the category $\Loc^\rc$.

\item[(iv)] The identity operations are $\id_N^{} : N\to N$.

\item[(v)] The permutation action of $\sigma^\prime\in\Sigma_n$ on $\und{f} : \und{M}\to N$ is given by
\begin{flalign}
\und{f}\sigma^\prime\, :\, \und{M}\sigma^{\prime}~\longrightarrow~ N\quad,
\end{flalign}
where $\und{f}\sigma^\prime = (f_{\sigma^\prime(1)},\dots,f_{\sigma^\prime(n)})$ 
and $\und{M}\sigma^{\prime}= (M_{\sigma^\prime(1)},\dots, M_{\sigma^{\prime}(n)})$ denote 
the permuted tuples.
\end{itemize}
Verifying that these data define an operad requires some technical
results for time-orderable tuples which have been proven in \cite[Lemma 4.3]{BPScomparison}.
\end{defi}
\begin{defi}\label{def:tPFA}
The category of \textit{tPFAs over $\Loc^\rc$} with values in a
bicomplete closed symmetric monoidal category $\TT$ is defined as the category
\begin{flalign}
\tPFA^\rc\,:=\, \Alg_{\tP_{\Loc^\rc}^{}}^{}\big(\TT\big)
\end{flalign}
of $\TT$-valued algebras over the tPFA operad $\tP_{\Loc^\rc}^{}$
from Definition \ref{def:tPFAoperad}. We denote by 
\begin{flalign}
\tPFA^{\rc,W}\,\subseteq\, \tPFA^\rc
\end{flalign}
the full subcategory consisting of all tPFAs satisfying the time-slice axiom,
i.e.\ the operad algebra (i.e.\ multifunctor)
$\FFF: \tP_{\Loc^\rc}^{}\to\TT$ sends every $1$-ary operation $f:M\to N$ 
in $\tP_{\Loc^\rc}^{}$ which corresponds to a Cauchy morphism in $\Loc^\rc$
to an isomorphism in $\TT$.
\end{defi}
\begin{rem}\label{rem:additivitytPFA}
The same conclusions of Remark \ref{rem:additivity} hold true for tPFAs. 
The analog of Proposition \ref{prop:additivity} for tPFAs
shows that the concept of tPFAs over $\Loc^\rc$ from Definition \ref{def:tPFA}
is a honest generalization of the additive tPFAs over $\Loc$ from the 
earlier comparison theorem in \cite{BPScomparison}. Furthermore, 
the analog of Proposition \ref{prop:additivity2} for tPFAs identifies 
the tPFAs over $\Loc^\rc$ coming from additive tPFAs over $\Loc$ 
by a similar additivity property. 
\sk

The key point to establish the tPFA-analogs of Propositions 
\ref{prop:additivity} and \ref{prop:additivity2} 
is again to observe that there exists an evident multifunctor
$\tP_{\Loc^\rc}\to \tP_{\Loc}$, where $\tP_{\Loc}$ is the
operad obtained by replacing in Definition \ref{def:tPFAoperad} 
the category $\Loc^\rc$ by the larger category $\Loc$. This multifunctor
is surjective on objects, hence the associated restriction functor
$(-)\vert^\rc : \tPFA\to \tPFA^{\rc}$ is faithful. 
In analogy with Proposition \ref{prop:additivity}, restricting this functor 
to the full subcategory of additive tPFAs, we obtain a fully faithful functor
\begin{flalign}\label{eqn:add-res-rc-tPFA}
\xymatrix{
(-)\vert^\rc\,:\,\tPFA^\mathrm{add}\ar[r]^-{\mathrm{f.f.}}~&~\tPFA^\rc\quad,
}
\end{flalign}
which means that additive tPFAs over $\Loc$ are a full subcategory of tPFAs over $\Loc^\rc$.
To prove this statement, we note that the
analog of the extended naturality square \eqref{eqn:extendednaturality}
is given in the present case by
\begin{flalign}\label{eqn:extendednaturalitytPFA}
\begin{gathered}
\xymatrix{
\ar[d]_-{\zeta_{\und{U}}}
\FFF(\und{U})\ar[r]^-{\FFF(\iota_{\und{U}}^{\und{M}})}~&~\ar[d]_-{\zeta_{\und{M}}}\FFF(\und{M}) 
\ar[r]^-{\FFF(\und{f})}~&~\FFF(N)\ar[d]_-{\zeta_N}\\
\GGG(\und{U})\ar[r]_-{\GGG(\iota_{\und{U}}^{\und{M}})}~&~\GGG(\und{M})\ar[r]_-{\GGG(\und{f})}~&~\GGG(N)
}
\end{gathered}\qquad,
\end{flalign}
where $\und{U}=(U_1,\dots,U_n)\in\mathbf{RC}_{\und{M}}$ is a tuple of objects $U_i\in\mathbf{RC}_{M_i}$
and the corresponding inclusion $\Loc^\rc$-morphisms  are denoted by
$\iota_{\und{U}}^{\und{M}} = (\iota_{U_1}^{M_1},\dots,\iota_{U_n}^{M_n})$.
To invoke universality of the colimit entering the additivity property,
one uses that 
\begin{flalign}
\nn \FFF(\und{M})\,&=\,\bigotimes_{i=1}^n\FFF(M_i)\, \cong \,
\bigotimes_{i=1}^n\colim_{U_i\in \mathbf{RC}_{M_i}}^{}\Big(\FFF(U_i)\Big)\\
\,&\cong\, \colim_{\und{U}\in\mathbf{RC}_{\und{M}}} \bigg(\bigotimes_{i=1}^n \FFF(U_i)\bigg)\,=\, 
\colim_{\und{U}\in\mathbf{RC}_{\und{M}}}\Big(\FFF(\und{U})\Big)\quad,
\end{flalign}
where the first isomorphism in the second line follows from our hypothesis that 
the symmetric monoidal category $\TT$ is closed, which implies 
that the monoidal product preserves colimits.
\sk

Furthermore, introducing the full subcategory 
$\tPFA^{\rc,\mathrm{add}} \subseteq \tPFA^\rc$ of additive tPFAs over 
$\Loc^\rc$ by mimicking the additivity property \eqref{eqn:additivity-rc} 
for AQFTs over $\Loc^\rc$ and using the same finality argument, one finds that 
the fully faithful functor \eqref{eqn:add-res-rc-tPFA} factors 
as the equivalence of categories 
\begin{flalign}\label{eqn:add-res-rc-add-tPFA}
\xymatrix{
(-)\vert^\rc\,:\,\tPFA^\mathrm{add}\ar[r]^-{\sim}~&~\tPFA^{\rc,\mathrm{add}}
}
\end{flalign}
followed by the full subcategory inclusion 
$\tPFA^{\rc,\mathrm{add}} \subseteq \tPFA^\rc$, which identifies the 
tPFAs over $\Loc^\rc$ which arise from additive tPFAs over $\Loc$ 
by a similar additivity property. To prove that \eqref{eqn:add-res-rc-add-tPFA} 
is indeed an equivalence of categories, one extends the proof strategy of Proposition 
\ref{prop:additivity2} to higher arity operations using again our hypothesis 
that the symmetric monoidal category $\TT$ is closed, 
which implies that the monoidal product preserves colimits.
\end{rem}

The Haag-Kastler $2$-functor machinery from \cite{BGSHaagKastler} 
can be adapted to the context of tPFAs, which allows us to decompose
a tPFA over $\Loc^\rc$ into a compatible family of tPFAs over individual 
spacetimes and, vice versa, assemble a compatible family of tPFAs 
over individual spacetimes into a tPFA over $\Loc^\rc$.
Since these techniques have not yet been spelled out in the literature,
we will do so in the remaining part of this subsection. We shall 
refer to the analog of the Haag-Kastler $2$-functor in the context
of (t)PFAs as the \textit{Costello-Gwilliam $2$-functor}.
This terminology is inspired by the fact that the main focus of \cite{CG1,CG2} is on (pre)factorization
algebras which are defined over a fixed manifold $M$, in analogy to
the original work \cite{HaagKastler} of Haag and Kastler on AQFT.
\begin{defi}\label{def:tP_M}
For each $M\in\Loc^\rc$, we denote by
\begin{flalign}\label{eqn:tP_M}
\tP_M\,\subseteq\,\tP_{\Loc^\rc}
\end{flalign}
the suboperad of the tPFA operad from Definition \ref{def:tPFAoperad}
whose objects are all $U\in \tP_{\Loc^\rc}$
corresponding to causally convex opens $U\subseteq M$ which
are either Cauchy or relatively compact and whose operations
are of the form $\iota_{\und{U}}^V :\und{U}\to V$,
where $\iota_{\und{U}}^V = (\iota_{U_1}^V,\dots,\iota_{U_n}^V)$ is a time-orderable
tuple of inclusion $\Loc^\rc$-morphisms. 
The category of \textit{tPFAs over $\Loc^\rc/M$} is defined as the category
\begin{flalign}
\CG(M)\,:=\,\Alg_{\tP_M}\big(\TT\big)
\end{flalign}
of $\TT$-valued algebras over this suboperad. We further denote by
\begin{flalign}
\CG^W(M)\,\subseteq\, \CG(M)
\end{flalign}
the full subcategory of all tPFAs over $\Loc^\rc/M$ satisfying the time-slice axiom.
\end{defi}

For any $\Loc^\rc$-morphism $f : M\to N$, there exists 
an evident multifunctor
(denoted with abuse of notation by the same symbol)
$f:\tP_M\to \tP_N$ which is given by taking images under $f$. This induces a pullback functor
$f^\ast : \CG^{(W)}(N)\to \CG^{(W)}(M)$ between the corresponding 
categories of tPFAs (satisfying the time-slice axiom) over the individual spacetimes.
These data assemble into a strict $2$-functor, called \textit{Costello-Gwilliam $2$-functor}
\begin{flalign}\label{eqn:CG2functor}
\CG^{(W)}\,:\,\big(\Loc^\rc\big)^\op~\longrightarrow~\CAT \quad,
\end{flalign}
which assigns to each $M\in\Loc^\rc$ the corresponding
category $\CG^{(W)}(M)\in \CAT$ of tPFAs over $\Loc^\rc/M$ (satisfying the time-slice axiom)
and to each $\Loc^\rc$-morphism $f:M\to N$ the pullback functor
$\CG^{(W)}(f) := f^\ast : \CG^{(W)}(N)\to \CG^{(W)}(M)$.
\begin{defi}\label{def:CGpoints}
The \textit{category of points} of the Costello-Gwilliam $2$-functor
$\CG^{(W)}$ is defined as the bicategorical limit
\begin{flalign}
\CG^{(W)}(\pt)\,:=\, \mathrm{bilim}\Big(\CG^{(W)}:\big(\Loc^\rc\big)^\op \to \CAT\Big)\,\in\,\CAT\quad.
\end{flalign}
\end{defi}
\begin{rem}\label{rem:CGpoints}
In analogy to Remark \ref{rem:HKpoints}, 
the category of points $\CG^{(W)}(\pt)$ admits the following explicit description
in terms of compatible families of tPFAs:
\begin{itemize}
\item An object in $\CG^{(W)}(\pt)$ is a tuple $\big(\{\FFF_M\},\{\phi_f\}\big)$
consisting of a family of tPFAs (satisfying the time-slice axiom) over individual spacetimes
$\FFF_M\in\CG^{(W)}(M)$, for all $M\in\Loc^\rc$, and a family of $\CG^{(W)}(M)$-isomorphisms
$\phi_f : \FFF_M\Rightarrow f^\ast(\FFF_N)$, for all $\Loc^\rc$-morphisms $f:M\to N$,
such that the diagrams
\begin{flalign}
\begin{gathered}
\xymatrix@C=3em{
\ar@{=>}[d]_-{\phi_{gf}}\FFF_M\ar@{=>}[r]^-{\phi_f}~&~f^\ast(\FFF_N)\ar@{=>}[d]^-{f^\ast(\phi_g)}  
~&~\FFF_M \ar@{=>}[r]^-{\phi_{\id_M}}\ar@{=>}[rd]_-{\id_{\FFF_M}}~&~\id_M^\ast(\FFF_M)\ar@{=}[d]\\
(gf)^\ast(\FFF_L)\ar@{=}[r] ~&~ f^\ast g^\ast(\FFF_L)  ~&~ ~&~\FFF_M
}
\end{gathered}
\end{flalign}
in $\CG^{(W)}(M)$ commute, for all composable $\Loc^\rc$-morphisms 
$f:M\to N$ and $g: N\to L$ and all objects $M\in \Loc^\rc$.

\item A morphism $\{\zeta_M\}: \big(\{\FFF_M\},\{\phi_f\}\big)\Rightarrow \big(\{\GGG_M\},\{\chi_f\}\big)$
in $\CG^{(W)}(\pt)$ is a family of $\CG^{(W)}(M)$-morphisms $\zeta_M : \FFF_M\Rightarrow\GGG_M$,
for all $M\in\Loc^\rc$, such that the diagram
\begin{flalign}
\begin{gathered}
\xymatrix@C=3em{
\ar@{=>}[d]_-{\zeta_M}\FFF_M \ar@{=>}[r]^-{\phi_f}~&~f^\ast(\FFF_N)\ar@{=>}[d]^-{f^\ast(\zeta_N)}\\
\GGG_M \ar@{=>}[r]_-{\chi_f}~&~f^\ast(\GGG_N)
}
\end{gathered}
\end{flalign}
in $\CG^{(W)}(M)$ commutes, for all $\Loc^\rc$-morphisms $f:M\to N$. \qedhere
\end{itemize}
\end{rem}

Similarly to the case of AQFTs in \eqref{eqn:dc} and \eqref{eqn:as}, there exist 
decomposition and assembly functors for tPFAs, which we shall denote by the same symbols
as it will be clear from the context if these functors act on AQFTs or tPFAs.
The \textit{decomposition functor}
\begin{subequations}\label{eqn:dctPFA}
\begin{flalign}\label{eqn:dctPFA1}
\dc\,:\,\tPFA^{\rc} ~\longrightarrow~\CG(\pt)
\end{flalign}
assigns to an object $\FFF\in\tPFA^\rc$ the object
in $\CG(\pt)$ which is specified by the tuple
\begin{flalign} 
\dc(\FFF)\,:=\,\big(\big\{\FFF\vert_M\big\},\big\{\FFF\vert_M\Rightarrow f^\ast(\FFF\vert_N)\big\}\big)\quad,
\end{flalign}
where $\FFF\vert_M\in \CG(M)$ is defined by restricting $\FFF$ to the suboperad \eqref{eqn:tP_M}
and the $\CG(M)$-isomorphism $\FFF\vert_M\Rightarrow f^\ast(\FFF\vert_N)$ is defined
by the components $\FFF(f\vert_U): \FFF(U)\to \FFF(f(U))$, for all $U\in\tP_M$,
where $f\vert_U : U\to f(U)$ denotes the domain and codomain restriction of the 
$\Loc^\rc$-morphism $f:M\to N$. The action of the functor \eqref{eqn:dctPFA1}
on a $\tPFA^\rc$-morphism $\zeta: \FFF\Rightarrow \GGG$ is defined by restriction
\begin{flalign}
\dc(\zeta)\,:=\,\big\{\zeta\vert_M\big\}\,:\,\dc(\FFF)~\Longrightarrow~\dc(\GGG)
\end{flalign}
\end{subequations}
to the suboperads \eqref{eqn:tP_M}.
The \textit{assembly functor}
\begin{subequations}\label{eqn:astPFA}
\begin{flalign}\label{eqn:astPFA1}
\as\,:\,\CG(\pt)~\longrightarrow~\tPFA^\rc
\end{flalign}
assigns to an object $\big(\{\FFF_M\},\{\phi_f\}\big)\in\CG(\pt)$
the object in $\tPFA^\rc$ which is defined by the multifunctor 
\begin{flalign}
\as\big(\{\FFF_M\},\{\phi_f\}\big)\,:\,\tP_{\Loc^\rc}~&\longrightarrow~ \TT\quad,\\
\nn M~&\longmapsto~\FFF_M(M)\quad,\\
\nn \big(\und{f}:\und{M}\to N\big)~&\longmapsto~\Big(\xymatrix@C=3.5em{
\FFF_{\und{M}}(\und{M})\ar[r]^-{(\phi_{\und{f}})_{\und{M}}}~&~\FFF_N(\und{f(M)})\ar[r]^-{\FFF_N(\iota_{\und{f(M)}}^N)}~&~\FFF_N(N)
}\Big)\quad,
\end{flalign}
where $(\phi_{\und{f}})_{\und{M}}:= \bigotimes_{i=1}^n (\phi_{f_i})_{M_i} :
\FFF_{\und{M}}(\und{M}) := \bigotimes_{i=1}^n\FFF_{M_i}(M_i)\to \bigotimes_{i=1}^n\FFF_{N}(f_i(M_i))=:
\FFF_N(\und{f(M)})$. 
The action of the functor \eqref{eqn:astPFA1} on a 
$\CG(\pt)$-morphism $\{\zeta_M\}:  \big(\{\FFF_M\},\{\phi_f\}\big)\Rightarrow
\big(\{\GGG_M\},\{\chi_f\}\big)$ is given by the multinatural transformation 
which is defined by the components
\begin{flalign}
\as\big(\{\zeta_M\}\big)_M\,:=\,\Big(\xymatrix@C=3.5em{
\FFF_M(M)\ar[r]^-{(\zeta_M)_M}~&~\GGG_M(M)
}\Big)\quad,
\end{flalign}
\end{subequations}
for all $M\in\Loc^\rc$.
\begin{theo}\label{theo:dcastPFA}
The decomposition \eqref{eqn:dctPFA} and assembly \eqref{eqn:astPFA} functors
are quasi-inverse to each other, hence they exhibit an equivalence of 
categories
\begin{flalign}
\tPFA^\rc\,\simeq\,\CG(\pt)\quad.
\end{flalign}
Furthermore, these functors preserve the respective time-slice axioms
and hence they induce also an equivalence of categories
\begin{flalign}
\tPFA^{\rc,W}\,\simeq\,\CG^W(\pt)\quad.
\end{flalign}
\end{theo}
\begin{proof}
One directly verifies that the composition $\as \circ \dc =\id $ is 
equal to the identity functor.
A natural isomorphism $\dc\circ \as\cong \id$ for the other composition
is constructed similarly to the one in the Haag-Kastler case \cite[Theorem 3.8]{BGSHaagKastler}.
Explicitly, given any object $\big(\{\FFF_M\},\{\phi_f\}\big)\in\CG(\pt)$,
the $\CG(\pt)$-isomorphism $\dc\big(\as\big(\{\FFF_M\},\{\phi_f\}\big)\big)\Rightarrow \big(\{\FFF_M\},\{\phi_f\}\big)$
is defined by the components $(\phi_{\iota_{U}^{M}})_U : \FFF_U(U)\to \FFF_M(U)$,
for all $U\in\tP_M$.
The statement about the preservation of the time-slice axioms is straightforward to check.
\end{proof}

%%%%%%%%%%%%%%%%%%%%%%%%%%%%%%%%%%%%%%%%%%%%%%%%
%%%%%%%%%%%%%%%%%%%%%%%%%%%%%%%%%%%%%%%%%%%%%%%%

\section{\label{sec:1categorical}A generalized 1-categorical equivalence theorem}
In this section we prove an equivalence theorem between AQFTs and tPFAs over $\Loc^\rc$,
both satisfying the time-slice axiom, in the case where the bicomplete 
closed symmetric monoidal target category $\TT$ is a $1$-category.
As a consequence of our observations in Remarks \ref{rem:additivity} and \ref{rem:additivitytPFA},
this result will be more general than the previous equivalence theorem 
in \cite{BPScomparison} which covers only additive theories. 
The main innovation of the present section is that we present a new
proof strategy, based on the Haag-Kastler/Costello-Gwilliam 
$2$-functor techniques from Section \ref{sec:prelim}, which 
admits a very useful generalization to the context where the target $\TT$ is a
symmetric monoidal model category. These homotopical aspects will be studied 
later in Sections \ref{sec:modelcategorical} and \ref{sec:spacetimewise}.
\sk

The key ingredient for our equivalence theorem
is the following multifunctor which allows us 
to compare AQFTs and tPFAs at the operadic level
from Definitions \ref{def:AQFToperad} and \ref{def:tPFAoperad}. The existence of this multifunctor 
has already been observed in \cite[Remark 5.2]{BPScomparison}.
\begin{defi}\label{def:comparisonoperadmorphism}
The \textit{tPFA/AQFT-comparison multifunctor} is defined by
\begin{flalign}\label{eqn:Phimorphism}
\Phi\,:\,\tP_{\Loc^{\rc}}~&\longrightarrow~\O_{\ovr{\Loc^{\rc}}}\quad,\\
\nn  M~&\longmapsto~M\quad,\\
\nn \big(\und{f}:\und{M}\to N\big)~&\longmapsto~\big([\rho^{-1},\und{f}]:\und{M}\to N\big)\quad,
\end{flalign}
where $\rho\in\Sigma_n$ is any choice\footnote{The AQFT operation
$[\rho^{-1},\und{f}]:\und{M}\to N$ is independent of the specific choice of time-ordering
permutation $\rho$ for the time-orderable tuple $\und{f}$. This is a consequence of \cite[Lemma 4.2 (iii)]{BPScomparison}
and the definition of the equivalence relation $\sim_\perp$ in the AQFT operad
from Definition \ref{def:AQFToperad} (ii).} 
of time-ordering permutation for the time-orderable tuple 
of $\Loc^\rc$-morphisms $\und{f} = (f_1:M_1\to N,\dots,f_n:M_n\to N)$.
\end{defi}

Pullback of operad algebras along the multifunctor \eqref{eqn:Phimorphism}
defines a functor
\begin{flalign}\label{eqn:Phipull}
\Phi^\ast\,:\, \AQFT^\rc ~\longrightarrow~\tPFA^\rc
\end{flalign}
which allows us to compare AQFTs and tPFAs over $\Loc^\rc$.
Since $\Phi$ preserves Cauchy morphisms, 
this pullback functor restricts to the full subcategories
\begin{flalign}\label{eqn:PhipullW}
\Phi^\ast\,:\, \AQFT^{\rc,W}~\longrightarrow~\tPFA^{\rc,W}
\end{flalign}
consisting of AQFTs and tPFAs that satisfy
the time-slice axiom from Definitions \ref{def:AQFT} and \ref{def:tPFA}.
\sk

The goal of this section is to prove that the functor 
\eqref{eqn:PhipullW} exhibits an equivalence of categories.
We use the Haag-Kastler/Costello-Gwilliam $2$-functor techniques
from Section \ref{sec:prelim} to reduce this problem  
to a family of simpler problems on individual spacetimes $M\in\Loc^\rc$.
Let us start with observing that \eqref{eqn:Phimorphism} restricts
to a family of multifunctors
\begin{flalign}\label{eqn:PhiM}
\Phi_M \,:\, \tP_M~\longrightarrow~\O_M\quad,
\end{flalign}
for all $M\in\Loc^\rc$, between the suboperads from Definitions \ref{def:O_M} and \ref{def:tP_M}.
This family is natural in the sense that the diagram
\begin{flalign}
\begin{gathered}
\xymatrix@C=3em{
\ar[d]_-{\Phi_M}\tP_M \ar[r]^-{f}~&~ \tP_N\ar[d]^-{\Phi_N}\\
\O_M \ar[r]_-{f}~&~\O_N
}
\end{gathered}
\end{flalign}
of multifunctors commutes, for all $\Loc^\rc$-morphisms $f:M\to N$,
where we recall that both $f: \tP_M\to \tP_N$ and $f:\O_M\to \O_N$ are defined 
by taking images under $f$. These data define via pullback the components of 
a $2$-natural transformation
\begin{flalign}\label{eqn:Phipull2functor}
\Phi^\ast \,:\,\HK~\Longrightarrow~\CG 
\end{flalign}
from the Haag-Kastler $2$-functor \eqref{eqn:HK2functor}
to the Costello-Gwilliam $2$-functor \eqref{eqn:CG2functor}.
Since the multifunctors in \eqref{eqn:PhiM} preserve 
the respective sets of Cauchy morphisms, \eqref{eqn:Phipull2functor}
restricts to a $2$-natural transformation
\begin{flalign}\label{eqn:Phipull2functorW}
\Phi^\ast \,:\,\HK^W~\Longrightarrow~\CG^W 
\end{flalign}
between the Haag-Kastler/Costello-Gwilliam $2$-functors which encode the time-slice axiom.
\begin{lem}\label{lem:comparisondiagram}
The diagrams of categories and functors 
\begin{flalign}\label{eqn:comparisondiagram}
\begin{gathered}
\xymatrix@C=3em{
\ar[d]_-{\dc} \AQFT^{\rc(,W)} \ar[r]^-{\Phi^\ast}~&~\tPFA^{\rc(,W)}\ar[d]^-{\dc}\\
\HK^{(W)}(\pt)\ar[r]_-{(\Phi^\ast)_\ast}~&~\CG^{(W)}(\pt)
}
\end{gathered}
\end{flalign}
commute, where $(\Phi^\ast)_\ast := \bilim(\Phi^\ast)$
denote the functors between the categories of points (see
Definitions \ref{def:HKpoints} and \ref{def:CGpoints})
which are induced by the $2$-natural transformations 
\eqref{eqn:Phipull2functor} and \eqref{eqn:Phipull2functorW}.
\end{lem}
\begin{proof}
This is a simple check using the explicit expression for 
the tPFA/AQFT-comparison multifunctor $\Phi$ from Definition 
\ref{def:comparisonoperadmorphism} and the decomposition functors
given by \eqref{eqn:dc} and \eqref{eqn:dctPFA}.
\end{proof}

Using Theorems \ref{theo:dcasAQFT} and \ref{theo:dcastPFA},
we obtain that both vertical arrows in the diagram 
\eqref{eqn:comparisondiagram} are equivalences of categories.
This implies that our problem of proving that the top horizontal arrow
$\Phi^\ast:\AQFT^{\rc,W}\to\tPFA^{\rc,W}$ is an equivalence reduces
to proving that the bottom horizontal arrow $(\Phi^\ast)_\ast : \HK^W(\pt)\to \CG^W(\pt)$
between the categories of points is one. The latter would follow if we can show that
the $2$-natural transformation $\Phi^\ast:\HK^W \Rightarrow\CG^W$ between the Haag-Kastler/Costello-Gwilliam
$2$-functors is a $2$-natural equivalence, i.e.\ the components
\begin{flalign}\label{eqn:PhiMpull}
\Phi_M^\ast\,:\,\HK^W(M)~\longrightarrow~\CG^W(M)
\end{flalign}
are equivalences of categories, for all $M\in\Loc^\rc$.
We will now prove that is indeed the case.
\begin{theo}\label{theo:equivalenceM}
For each $M\in\Loc^\rc$, the functor \eqref{eqn:PhiMpull} 
exhibits an equivalence of categories. Even stronger,
this functor admits a strict inverse which will be constructed in the proof below.
\end{theo}
\begin{proof}
The inverse functor $\CG^W(M)\to \HK^W(M)$ can be defined by using the
same constructions as in \cite[Section 3]{BPScomparison}. Let us start
with observing that \eqref{eqn:PhiM} restricts to
an isomorphism $\Phi_M^1 : \tP_M^1\stackrel{\cong}{\longrightarrow}\O_M^1$
on the subcategories of $1$-ary operations $\tP_M^1 \subseteq \tP_M$ and $\O_M^1\subseteq \O_M$. 
Given any tPFA $\FFF_M\in \CG^W(M)$, we will construct an AQFT $\AAA_M\in\HK^W(M)$ 
whose underlying functor is given by $\AAA_M := \FFF_M\vert^{1} \circ (\Phi_M^1)^{-1} : \O_M^1\to\TT$.
Using Remark \ref{rem:AQFT}, this amounts to 
1.)~defining a unital associative algebra structure
on $\AAA_M(U)$, for all $U\in \O_M^1$, 
2.)~verifying naturality of these algebra structures with respect
to $\O_M^1$-morphisms, and 3.)~verifying the Einstein causality axiom.
Step 1.)~is the content of \cite[Proposition 3.4]{BPScomparison}. Step
2.)~is carried out in \cite[Lemma 3.5]{BPScomparison} for the case
of relatively compact morphisms, and similar arguments apply to Cauchy morphisms.
Step 3.)~is carried out in \cite[Lemma 3.9]{BPScomparison}. 
The verification that the resulting functor $\CG^W(M)\to \HK^W(M)$
is inverse to \eqref{eqn:PhiMpull} is similar to \cite[Theorem 5.1]{BPScomparison}.
\end{proof}

Combining the above results, we can state and prove
the following global tPFA/AQFT equivalence theorem.
\begin{theo}\label{theo:1categoricalequivalence}
Let $\TT$ be any bicomplete closed symmetric monoidal $1$-category. 
The functor
\begin{flalign}
\Phi^\ast\,:\, \AQFT^{\rc,W}~\stackrel{\sim}{\longrightarrow}~\tPFA^{\rc,W}
\end{flalign}
given by pullback of operad algebras 
along the tPFA/AQFT-comparison multifunctor \eqref{eqn:Phimorphism}
exhibits an equivalence between the category $\AQFT^{\rc,W}$
of $\TT$-valued AQFTs over $\Loc^\rc$ satisfying the time-slice axiom 
(Definition \ref{def:AQFT}) and the category $\tPFA^{\rc,W}$
of $\TT$-valued tPFAs over $\Loc^\rc$ satisfying the time-slice axiom 
(Definition \ref{def:tPFA}).
\end{theo}
\begin{proof}
Theorem \ref{theo:equivalenceM} shows that the spacetime-wise comparison
functor $\Phi_M^\ast : \HK^W(M)\to\CG^W(M)$ is an equivalence of categories,
for all $M\in\Loc^\rc$. This implies that the $2$-natural transformation
$\Phi^\ast : \HK^{W}\Rightarrow \CG^W$ between the Haag-Kastler and Costello-Gwilliam 
$2$-functors from \eqref{eqn:Phipull2functorW} is a $2$-natural equivalence,
hence its bicategorical limit $(\Phi^\ast)_\ast : \HK^{W}(\pt)\to \CG^W(\pt)$ is an equivalence of categories.
The result then follows from the commutative square in \eqref{eqn:comparisondiagram}
together with the fact that both vertical arrows are equivalences 
of categories by Theorems \ref{theo:dcasAQFT} and \ref{theo:dcastPFA}.
\end{proof}

%%%%%%%%%%%%%%%%%%%%%%%%%%%%%%%%%%%%%%%%%%%%%%%%
%%%%%%%%%%%%%%%%%%%%%%%%%%%%%%%%%%%%%%%%%%%%%%%%

\section{\label{sec:modelcategorical}Homotopical reduction to spacetime-wise problems}
In this section we consider AQFTs and tPFAs which take values in 
the category
\begin{flalign}
\TT\,:=\, \Ch_R
\end{flalign}
of cochain complexes of modules over a commutative, associative and unital algebra $R$ 
over a field $\bbK$ of characteristic zero. The 
motivation behind considering cochain complexes is to encode
homotopical phenomena of quantum field theories which are particularly present in gauge theories,
see e.g.\ \cite{FredenhagenRejzner,BBSLinearYM,BMScomparison} for AQFTs 
and \cite{CG1,CG2} for prefactorization algebras.
\sk

The crucial difference between the present context and Section \ref{sec:1categorical}
is that cochain complexes are higher-categorical objects which should be compared by 
quasi-isomorphisms instead of isomorphisms. This can be encoded by endowing
the category $\Ch_R$ with its projective model structure, see e.g.\ \cite[Chapter 2.3]{Hovey}
and \cite[Section 1.1]{Riehl}. Recall that in this model category a $\Ch_R$-morphism
$f : V\to W$ is a weak equivalence if it is a quasi-isomorphism,
a fibration if it is degree-wise surjective, and a cofibration if it has the left
lifting property against all acyclic fibrations. It is well-known that the projective
model structure is compatible with the standard closed symmetric monoidal structure on cochain complexes
in the sense that $\Ch_R$ is a \textit{symmetric monoidal model category}, see e.g.\ \cite[Chapter 4.2]{Hovey}.
For later use, let us note that the symmetric monoidal model category $\Ch_R$
is combinatorial and tractable, i.e.\ its underlying category is locally presentable
and the model structure is cofibrantly generated with generating (acyclic) cofibrations having cofibrant domains.
These niceness properties will become relevant below.
\sk

The goal of this section is to prove that the 
global homotopical equivalence problem 
for $\Ch_R$-valued AQFTs and tPFAs over $\Loc^\rc$
satisfying the homotopy time-slice axiom can be reduced to
a family of simpler spacetime-wise homotopical
equivalence problems for $\Ch_R$-valued AQFTs and tPFAs over $\Loc^\rc/M$
satisfying the homotopy time-slice axiom, for all $M\in\Loc^\rc$.

\subsection{\label{subsec:modelstructures}(Semi-)model structures}
In this subsection we endow the categories of $\TT=\Ch_R$-valued AQFTs and tPFAs
from Section \ref{sec:prelim} with suitable (semi-)model structures which allow 
us to describe homotopical phenomena of quantum field theories. 
\sk

The first, and simplest, kind of model structures that we require in our work are the standard
projective model structures on $\Ch_R$-valued algebras over operads. These 
exist under our hypothesis that the underlying field $\bbK\subseteq R$ is of characteristic $0$.
\begin{theo}[\cite{Hinich1,Hinich2}]\label{theo:projectivemodel}
For every colored operad $\Q$, the category $\Alg_\Q\big(\Ch_R\big)$
of $\Ch_R$-valued $\Q$-algebras carries the projective model structure in which a morphism
$\zeta: \AAA\Rightarrow \BBB$ is a weak equivalence (respectively, fibration) if each
component $\zeta_M : \AAA(M)\to \BBB(M)$ is a weak equivalence (respectively, fibration)
in $\Ch_R$, for all objects $M\in\Q$. A morphism is a cofibration if it
has the left lifting property against all acyclic fibrations. 
We denote the resulting projective model category by the same symbol
$\Alg_\Q\big(\Ch_R\big)$ and note that it is combinatorial and tractable.
\end{theo}
An immediate but important consequence of such model structures 
is the following result, see e.g.\ \cite{Hinich1,Hinich2}.
\begin{propo}\label{prop:Quillenprojective}
For every multifunctor $F : \Q\to \P$, the adjunction
\begin{flalign}
\xymatrix{
F_!\,:\, \Alg_\Q\big(\Ch_R\big) \ar@<0.75ex>[r]~&~\ar@<0.75ex>[l]\Alg_\P\big(\Ch_R\big)\,:\,F^\ast
}\quad,
\end{flalign}
which is given by pullback of operad algebras $F^\ast$ and operadic left Kan extension $F_!$,
is a Quillen adjunction $F_! \dashv F^\ast$ 
between the projective model categories from Theorem \ref{theo:projectivemodel}.
\end{propo}
\begin{ex}\label{ex:projectivemodel}
Applying Theorem \ref{theo:projectivemodel} to the context of
Section \ref{sec:prelim}, we obtain the following model categories
of AQFTs and tPFAs:
\begin{itemize}
\item $\AQFT^\rc$ denotes the 
projective model category of $\Ch_R$-valued AQFTs over $\Loc^\rc$.

\item $\HK(M)$ denotes the 
projective model category of $\Ch_R$-valued AQFTs over $\Loc^\rc/M$.

\item $\tPFA^\rc$ denotes the 
projective model category of $\Ch_R$-valued tPFAs over $\Loc^\rc$.

\item $\CG(M)$ denotes the 
projective model category of $\Ch_R$-valued tPFAs over $\Loc^\rc/M$.
\end{itemize}
Using also Proposition \ref{prop:Quillenprojective}, 
the tPFA/AQFT-comparison multifunctor from Definition \ref{def:comparisonoperadmorphism}
defines a Quillen adjunction
\begin{subequations}\label{eqn:PhiQuillen}
\begin{flalign}
\xymatrix{
\Phi_!\,:\, \tPFA^\rc \ar@<0.75ex>[r]~&~\ar@<0.75ex>[l]\AQFT^\rc\,:\,\Phi^\ast
}
\end{flalign}
and its restrictions \eqref{eqn:PhiM} to individual spacetimes
define a family of Quillen adjunctions
\begin{flalign}
\xymatrix{
\Phi_{M\,!}\,:\, \CG(M) \ar@<0.75ex>[r]~&~\ar@<0.75ex>[l]\HK(M)\,:\,\Phi_M^\ast
}\quad,
\end{flalign}
\end{subequations}
for all $M\in\Loc^\rc$.
\end{ex}

The second kind of (semi-)model structures that we require in our work
are left Bousfield localizations (see e.g.\ \cite[Chapter 4.1]{Balchin} for 
an excellent survey) of the projective model categories from Theorem \ref{theo:projectivemodel}.
Such model structures are more subtle because they are only guaranteed to
exist under additional hypotheses on the original model structure, most notably left properness, 
see e.g.\ \cite[Proposition 4.1.4]{Balchin}. The projective model 
categories from Theorem \ref{theo:projectivemodel} 
are in general \textit{not} left proper\footnote{In the special case where $R=\bbK$ is a field of characteristic
zero, it was previously claimed in \cite[Proposition 4.10]{Carmona} that
the projective model structure on the category $\Alg_{\Q}(\Ch_\bbK)$ of $\Ch_\bbK$-valued
algebras over \textit{any} colored operad $\Q$ is left proper. This claim is in general not true
and it requires additional assumptions on the operad $\Q$.}, 
hence we can not apply the standard theory of left Bousfield localizations of model categories. 
A solution to this issue is to work within the more
flexible framework of \textit{semi-model categories},
see e.g.\ \cite{Barwick,Fresse,BataninWhite,CarmonaSemi}.
A semi-model category carries the same structures as a model category,
i.e.\ classes of weak equivalences, fibrations and cofibrations,
but these classes are required to satisfy slightly weaker lifting
and factorization axioms; more specifically, those are restricted to morphisms 
with cofibrant domain. See in particular \cite[Definition 2.1]{BataninWhite}
for a detailed description of the concept of semi-model category 
that we use in our work.
The main advantage of working in the context of semi-model categories is that
left Bousfield localizations exist under less restrictive
hypotheses \cite[Theorem A]{BataninWhite} which are satisfied by the 
projective model categories from Theorem \ref{theo:projectivemodel}.
\begin{theo}[\cite{BataninWhite,CFM,Carmona}]\label{theo:Bousfieldmodel}
Let $\Q$ be any colored operad and $W\subseteq \mathrm{Mor}^1(\Q)$
any subset of the set of $1$-ary operations in $\Q$.
Denote by $\Alg_\Q\big(\Ch_R\big)$
the projective model category of $\Ch_R$-valued $\Q$-algebras from 
Theorem \ref{theo:projectivemodel}.
\begin{itemize}
\item[(a)] There exists a subset $\widehat{W}\subseteq \mathrm{Mor}\big(\Alg_\Q\big(\Ch_R\big)\big)$
such that an object $\AAA\in \Alg_\Q\big(\Ch_R\big)$ is $\widehat{W}$-local if and only if
the multifunctor $\AAA : \Q\to \Ch_R$ sends every $1$-ary operation in 
$W$ to a weak equivalence in $\Ch_R$.

\item[(b)] The left Bousfield localization $\mathcal{L}_{\widehat{W}} \Alg_\Q\big(\Ch_R\big)$
of the projective model category $\Alg_\Q\big(\Ch_R\big)$ at the set of morphisms
$\widehat{W}$ from item (a) exists as a combinatorial and tractable semi-model category.
The fibrant objects in this semi-model category are precisely those
multifunctors $\AAA : \Q\to \Ch_R$ which send every $1$-ary operation in 
$W$ to a weak equivalence in $\Ch_R$.
\end{itemize}
\end{theo}
\begin{rem}\label{rem:Bousfieldmodel}
For concreteness, let us note that the set of morphisms 
$\widehat{W}\subseteq \mathrm{Mor}\big(\Alg_\Q\big(\Ch_R\big)\big)$ 
from item (a) can be constructed explicitly as follows. Denote by $\Q^1\subseteq \Q$ the subcategory
of $1$-ary operations in the colored operad $\Q$. Using the covariant Yoneda embedding
\begin{flalign}
\mathsf{y}(-)\,:\,\big(\Q^1\big)^\op~\longrightarrow~\Fun\big(\Q^1,\Ch_R\big)~~,\quad 
M~\longmapsto~\mathsf{y}(M) \,=\, \Q^1(M,-)\otimes R \quad,
\end{flalign}
we can assign to a $1$-ary operation $f:M\to N$ in $\Q$ 
a natural transformation $\mathsf{y}(f) : \mathsf{y}(N)\Rightarrow\mathsf{y}(M)$
between functors from $\Q^1$ to $\Ch_R$. We can further shift the cohomological
degree by any integer $r\in\bbZ$ and consider the natural transformation
$\mathsf{y}(f)[r] : \mathsf{y}(N)[r]\Rightarrow\mathsf{y}(M)[r]$. 
Using operadic left Kan extension $j_! : \Fun\big(\Q^1,\Ch_R\big)\to \Alg_\Q\big(\Ch_R\big)$
along the inclusion multifunctor $j : \Q^1\hookrightarrow \Q$, we can define the
multinatural transformation between $\Ch_R$-valued $\Q$-algebras
\begin{flalign}\label{eqn:What}
j_!\big(\mathsf{y}(f)[r]\big)\, :\, j_!\big( \mathsf{y}(N)[r]\big)~\Longrightarrow~j_!\big(\mathsf{y}(M)[r]\big)\quad.
\end{flalign} 
The subset $\widehat{W}\subseteq \mathrm{Mor}\big(\Alg_\Q\big(\Ch_R\big)\big)$ 
consists of all morphisms of the form \eqref{eqn:What}, where the $1$-ary operation $(f:M\to N)\in W$
runs over the set $W$ and $r\in\bbZ$ runs over all integers.
\end{rem}
\begin{rem}
It is important to emphasize that the underlying category
of the left Bousfield localization $\mathcal{L}_{\widehat{W}}\Alg_\Q\big(\Ch_R\big)$ 
is precisely the same as the underlying category of the projective model category $\Alg_\Q\big(\Ch_R\big)$.
The mechanism how the semi-model category $\mathcal{L}_{\widehat{W}}\Alg_\Q\big(\Ch_R\big)$
selects those multifunctors $\AAA : \Q\to \Ch_R$ which send
every $1$-ary operation in $W$ to a weak equivalence in $\Ch_R$
is more indirect and sophisticated: 
In the semi-model category $\mathcal{L}_{\widehat{W}}\Alg_\Q\big(\Ch_R\big)$ 
one enlarges the class of weak equivalences
to what is called $\widehat{W}$-equivalences (see e.g.\ \cite[Definition 4.1.1]{Balchin})
and keeps the class of cofibrations the same as the projective cofibrations.
This necessarily reduces the class of fibrations, which are determined by suitable 
lifting properties as in \cite[Theorem 4.2]{BataninWhite}, such that the fibrant objects in 
$\mathcal{L}_{\widehat{W}}\Alg_\Q\big(\Ch_R\big)$ are precisely the 
multifunctors $\AAA : \Q\to \Ch_R$ sending 
every $1$-ary operation in $W$ to a weak equivalence in $\Ch_R$.
\end{rem}
The result in Proposition \ref{prop:Quillenprojective} admits the 
following upgrade to left Bousfield localizations.
\begin{propo}\label{prop:QuillenBousfield}
Let $\Q$ and $\P$ be two colored operads equipped with subsets 
$W\subseteq \mathrm{Mor}^1(\Q)$ and $S \subseteq \mathrm{Mor}^1(\P)$ of $1$-ary operations.	
Let $F : \Q\to \P$ be any multifunctor which preserves
these subsets, i.e.\ $F(W)\subseteq S$. Then the Quillen adjunction from Proposition 
\ref{prop:Quillenprojective} induces a Quillen adjunction 
\begin{flalign}
\xymatrix{
F_!\,:\, \mathcal{L}_{\widehat{W}}\Alg_\Q\big(\Ch_R\big) \ar@<0.75ex>[r]~&~\ar@<0.75ex>[l]
\mathcal{L}_{\widehat{S}}\Alg_\P\big(\Ch_R\big)\,:\,F^\ast
}
\end{flalign}
between the left Bousfield localized semi-model categories from Theorem \ref{theo:Bousfieldmodel}.
\end{propo}
\begin{proof}
Using the characterization of Quillen adjunctions for left Bousfield
localizations of combinatorial and tractable semi-model categories
from \cite[Lemma 3.5]{CarmonaCriterion}, we have to prove that the right adjoint functor $F^\ast$ sends
every $\widehat{S}$-local object $\BBB\in \Alg_\P\big(\Ch_R\big)$ 
to a $\widehat{W}$-local object $F^\ast(\BBB)\in \Alg_\Q\big(\Ch_R\big)$.
Using item (a) of Theorem \ref{theo:Bousfieldmodel}, we have that $\BBB:\P\to\Ch_R$ sends
every $1$-ary operation in $S$ to a weak equivalence in $\Ch_R$, and we have to show that 
$F^\ast(\BBB)=\BBB\,F :\Q\to\Ch_R$ sends every $1$-ary operation in $W$ to a weak equivalence in $\Ch_R$.
This follows directly from our hypothesis that $F(W)\subseteq S$ preserves the subsets of $1$-ary operations.
\end{proof}
\begin{ex}\label{ex:Bousfieldmodel}
Applying Theorem \ref{theo:Bousfieldmodel} to the context of 
Section \ref{sec:prelim}, we obtain the following semi-model categories
whose fibrant objects are precisely the AQFTs or tPFAs which 
send every Cauchy morphism to a weak equivalence in $\Ch_R$, i.e.\ 
they satisfy the homotopy time-slice axiom:
\begin{itemize}
\item $\mathcal{L}_{\widehat{W}}\AQFT^\rc$ denotes the semi-model category
whose fibrant objects are $\Ch_R$-valued AQFTs over $\Loc^\rc$ 
satisfying the homotopy time-slice axiom.

\item $\mathcal{L}_{\widehat{W}_M}\HK(M)$ denotes the semi-model category
whose fibrant objects are $\Ch_R$-valued AQFTs over $\Loc^\rc/M$ 
satisfying the homotopy time-slice axiom.

\item $\mathcal{L}_{\widehat{W}}\tPFA^\rc$ denotes the semi-model category
whose fibrant objects are $\Ch_R$-valued tPFAs over $\Loc^\rc$ 
satisfying the homotopy time-slice axiom.

\item $\mathcal{L}_{\widehat{W}_M}\CG(M)$ denotes the semi-model category
whose fibrant objects are $\Ch_R$-valued tPFAs over $\Loc^\rc/M$ 
satisfying the homotopy time-slice axiom.
\end{itemize}
Using also Proposition \ref{prop:QuillenBousfield}, the tPFA/AQFT-comparison
Quillen adjunctions from Example \ref{ex:projectivemodel}
induce Quillen adjunctions
\begin{subequations}\label{eqn:PhiQuillenW}
\begin{flalign}
\xymatrix{
\Phi_!\,:\, \mathcal{L}_{\widehat{W}}\tPFA^\rc \ar@<0.75ex>[r]~&~\ar@<0.75ex>[l] \mathcal{L}_{\widehat{W}}\AQFT^\rc\,:\,\Phi^\ast
}
\end{flalign}
and
\begin{flalign}
\xymatrix{
\Phi_{M\,!}\,:\, \mathcal{L}_{\widehat{W}_M}\CG(M) \ar@<0.75ex>[r]~&~\ar@<0.75ex>[l]\mathcal{L}_{\widehat{W}_M}\HK(M)\,:\,\Phi_M^\ast
}\quad,
\end{flalign}
\end{subequations}
for all $M\in\Loc^\rc$, between the left Bousfield localized semi-model categories
which encode tPFAs and AQFTs satisfying the homotopy time-slice axiom.
\end{ex}

\subsection{\label{subsec:homotopicalpoints}Homotopical points of $2$-functors}
In this subsection we develop a homotopical refinement of the categories
of points from Definitions \ref{def:HKpoints} and \ref{def:CGpoints}.
In our present context, the Haag-Kastler \eqref{eqn:HK2functor}
and Costello-Gwilliam \eqref{eqn:CG2functor} $2$-functors do not 
assign ordinary categories in $\CAT$, but they assign the (semi-)model categories 
from Examples \ref{ex:projectivemodel} and \ref{ex:Bousfieldmodel}.
We take this additional structure into account by replacing the $2$-category
$\CAT$ with the following $2$-category.
\begin{defi}\label{def:Comb}
We denote by $\CombR$ the $2$-category whose objects are combinatorial
and tractable semi-model categories, morphisms are right Quillen functors 
and $2$-morphisms are natural transformations.
\end{defi}

The Haag-Kastler and Costello-Gwilliam $2$-functors obtained
from the (semi-)model categories in Examples \ref{ex:projectivemodel} 
and \ref{ex:Bousfieldmodel} are thus $2$-functors of the form
\begin{subequations}\label{eqn:HKCGCombR}
\begin{flalign}
&& \HK\,:~&\, \big(\Loc^\rc\big)^\op \longrightarrow~\CombR\quad , 
~& \CG\,:~&\, \big(\Loc^\rc\big)^\op \longrightarrow~\CombR\quad , &\\
&& \mathcal{L}_{\widehat{W}}\HK\,:~&\, \big(\Loc^\rc\big)^\op \longrightarrow~\CombR\quad, 
~& \mathcal{L}_{\widehat{W}}\CG\,:~&\, \big(\Loc^\rc\big)^\op \longrightarrow~\CombR\quad. &
\end{flalign}
\end{subequations}
The fact that these $2$-functors assign a right Quillen functor (the pullback functor $f^\ast$) 
to every $\Loc^\rc$-morphism $f:M\to N$ follows from 
Propositions \ref{prop:Quillenprojective} and \ref{prop:QuillenBousfield}.
Since (semi-)model categories are naturally 
compared by Quillen equivalences, in contrast
to equivalences of plain categories, one has to refine the 
bicategorical limits in the Definitions \ref{def:HKpoints}
and \ref{def:CGpoints} of the categories of points to
homotopy limits. In this subsection we will achieve this goal 
by adapting ideas of Barwick \cite[Application II]{Barwick}, 
extending them to the realm of semi-model categories 
and providing some additional technical results.
\sk

In order to avoid repetitions, let us consider in this subsection an arbitrary contravariant
$2$-functor $\mathsf{X} : \CC^\op\to \CombR$ from a small category $\CC$ to the 
$2$-category $\CombR$ from Definition \ref{def:Comb}.
As a first step towards computing the homotopy limit of this $2$-functor,
we introduce the category of right sections of $\mathsf{X}$, 
which can be also understood as the oplax $2$-limit of $\mathsf{X}$.
\begin{defi}\label{def:rightsections}
The \textit{category of right sections} $\SectR(\mathsf{X})$ 
of a $2$-functor $\mathsf{X} : \CC^\op\to \CombR$
consists of the following objects and morphisms:
\begin{itemize}
\item An object in $\SectR(\mathsf{X})$ is a tuple $(\{x_M\},\{\psi_f\})$
consisting of a family of objects $x_M\in \mathsf{X}(M)$, for all $M\in\CC$,
and a family of $\mathsf{X}(M)$-morphisms $\psi_f : x_M \to \mathsf{X}(f)(x_N)$,
for all $\CC$-morphisms $f:M\to N$, such that the diagrams
\begin{flalign}
\begin{gathered}
\xymatrix@C=3em{
\ar[d]_-{\psi_{gf}}x_M \ar[r]^-{\psi_f}~&~\mathsf{X}(f)(x_N)\ar[d]^-{\mathsf{X}(f)(\psi_g)}
~&~ \ar[dr]_-{\id_{x_M}}x_M \ar[r]^-{\psi_{\id_M}}~&~\mathsf{X}(\id_M)(x_M)\ar@{=}[d]\\
\mathsf{X}(gf)(x_L)\ar@{=}[r]~&~\mathsf{X}(f)\mathsf{X}(g)(x_L) ~&~~&~x_M
}
\end{gathered}
\end{flalign}
in $\mathsf{X}(M)$ commute, for all composable 
$\CC$-morphisms $f: M\to N$ and $g: N\to L$
and all objects $M\in\CC$.

\item A morphism $\{\zeta_M\} : (\{x_M\},\{\psi_f\})\to (\{x_M^\prime\},\{\psi_f^\prime\})$
in $\SectR(\mathsf{X})$ is a family of $\mathsf{X}(M)$-morphisms $\zeta_M : x_M\to x_M^\prime$,
for all $M\in\CC$, such that the diagram
\begin{flalign}
\begin{gathered}
\xymatrix@C=3em{
\ar[d]_-{\zeta_M} x_M \ar[r]^-{\psi_f}~&~ \mathsf{X}(f)(x_N) \ar[d]^-{\mathsf{X}(f)(\zeta_N)}\\
x_M^\prime \ar[r]_-{\psi^\prime_f}~&~\mathsf{X}(f)(x_N^\prime)
}
\end{gathered}
\end{flalign}
in $\mathsf{X}(M)$ commutes, for all $\CC$-morphisms $f: M\to N$.
\end{itemize}
\end{defi}

\begin{rem}\label{rem:SectR}
The reader may have recognized the similarity between 
$\SectR(\mathsf{X})$ and the categories of points
from Remarks \ref{rem:HKpoints} and \ref{rem:CGpoints}.
We however would like to emphasize the crucial difference
that the $\mathsf{X}(M)$-morphisms  $\psi_f : x_M \to \mathsf{X}(f)(x_N)$
in Definition \ref{def:rightsections} are 
\textit{not} required to be isomorphisms, in contrast to
the $\HK^{(W)}(M)$-\textit{iso}morphisms $\alpha_f : \AAA_M\Rightarrow f^\ast(\AAA_N)$
in $\HK^{(W)}(\mathrm{pt})$ and the $\CG^{(W)}(M)$-\textit{iso}morphisms 
$\phi_f : \FFF_M\Rightarrow f^\ast(\FFF_N)$
in $\CG^{(W)}(\mathrm{pt})$. 
Of course, $\HK^{(W)}(\pt)$ and $\CG^{(W)}(\pt)$ sit inside 
$\SectR(\HK^{(W)})$ and $\SectR(\CG^{(W)})$, respectively, 
as full subcategories.
\end{rem}

The following result is a generalization of \cite[Theorem 2.28]{Barwick}
to our context of semi-model categories.
\begin{propo}\label{prop:SectRprojective}
For any $2$-functor $\mathsf{X} : \CC^\op\to \CombR$ with $\CC$ a small $1$-category, the category
$\SectR(\mathsf{X})$ of right sections from Definition \ref{def:rightsections} 
admits a combinatorial and tractable semi-model category structure, called the projective 
semi-model structure, in which a morphism
$\{\zeta_M\} : (\{x_M\},\{\psi_f\})\to (\{x_M^\prime\},\{\psi_f^\prime\})$
is a weak equivalence (respectively, fibration) if each component
$\zeta_M : x_M\to x_M^\prime$ is a weak equivalence (respectively, fibration)
in $\mathsf{X}(M)$, for all $M\in\CC$. 
We denote the resulting projective semi-model category of 
right sections by the same symbol $\SectR(\mathsf{X})$.
For every cofibration $\{\zeta_M\} : 
(\{x_M\},\{\psi_f\})\to (\{x_M^\prime\},\{\psi_f^\prime\})$
with cofibrant domain in this semi-model structure, 
the components $\zeta_M : x_M\to x_M^\prime$
are cofibrations in $\mathsf{X}(M)$, for all $M\in\CC$. 
\end{propo}
\begin{proof}
Let us start with observing that there exists an adjunction
\begin{flalign}
\xymatrix{
F\,:\, \prod\limits_{M\in\CC}\mathsf{X}(M) \ar@<0.75ex>[r]~&~\ar@<0.75ex>[l]\SectR(\mathsf{X})\,:\,U
}
\end{flalign}
whose right adjoint functor $U:\SectR(\mathsf{X})\to \prod_{M\in\CC}\mathsf{X}(M)$ 
assigns the underlying components of right sections, i.e.\ $(\{x_M\},\{\psi_f\})\mapsto\{x_M\} $
and $\{\zeta_M\}\mapsto \{\zeta_M\}$. Indeed, since limits and colimits are computed component-wise
in $\SectR(\mathsf{X})$, the functor $U$ preserves both limits and colimits, hence a left adjoint functor
$F$ exists as a consequence of the special adjoint functor theorem for locally presentable categories.
We endow $\prod_{M\in\CC}\mathsf{X}(M)$ with the product semi-model structure
and observe that taking unions $I:= \coprod_{M\in\CC}I_M$ and $J:=\coprod_{M\in\CC}J_M$ 
of the component-wise generating (acyclic) cofibrations $(I_M,J_M)$ defines generating
(acyclic) cofibrations $(I,J)$ for $\prod_{M\in\CC}\mathsf{X}(M)$. 
Note that this semi-model structure on  $\prod_{M\in\CC}\mathsf{X}(M)$ is tractable. 
\sk

We will now verify that a simplified version of the transfer criterion in \cite[Theorem 3.3]{Fresse2},
taking into account local presentability of our categories and tractability of the product semi-model structure,
holds true for our example. This proves existence of the projective semi-model structure on 
$\SectR(\mathsf{X})$ and implies the properties stated in this proposition.
\sk

Similarly to \cite[Theorems 11.3.1 and 11.3.2]{Hirschhorn}, we must check that $U$
sends $F(I)$-cell complexes (respectively, $F(J)$-cell complexes) with cofibrant 
domains to cofibrations (respectively, acyclic cofibrations) in $\prod_{M\in\CC}\mathsf{X}(M)$ 
since the small object argument is always available in a locally presentable category for any set of maps. 
As observed by Fresse in \cite[Theorem 3.3]{Fresse2}, due to the fact that $U$ preserves 
colimits over non-empty ordinals, it suffices to show that, for any pushout
\begin{flalign}
\begin{gathered}
\xymatrix{
\ar[d]_-{F(\{i_M\})} F\big(\{x_M\}\big) \ar[r]~&~ \big(\{y_M\},\{\psi_f\}\big)\ar@{-->}[d]^-{\{\zeta_M\}}\\
F\big(\{x^\prime_M\}\big)\ar@{-->}[r]~&~\big(\{y^\prime_M\},\{\psi^\prime_f\}\big)
}
\end{gathered}
\end{flalign}
in $\SectR(\mathsf{X})$, the morphism $U(\{\zeta_M\}) : \{y_M\}\to \{y_M^\prime\}$
is a cofibration (respectively, an acyclic cofibration) in $\prod_{M\in\CC}\mathsf{X}(M)$
whenever $\{i_M\}:\{x_M\} \to \{x^\prime_M\}$ belongs to $I$ 
(respectively, to $J$) and $(\{y_M\},\{\psi_f\})$ is cofibrant in $\prod_{M\in\CC}\mathsf{X}(M)$. 
Note that Fresse's assumption is stronger than ours, but we do not use 
the full strength of his result since we already know that the small object 
argument applies to $F(I)$ and $F(J)$. Using that $U$ preserves colimits 
and that (acyclic) cofibrations are closed under pushouts, 
this problem reduces to showing that the morphism $UF\big(\{i_M\}\big) : UF\big(\{x_M\}\big)
\to UF\big(\{x^\prime_M\}\big)$ is a (acyclic) cofibration in $\prod_{M\in\CC}\mathsf{X}(M)$,
for all generating cofibrations $\{i_M\} : \{x_M\}\to \{x^\prime_M\}$
in $I$ (respectively, in $J$). This can be simplified further by using the family of
adjunctions $p_M^\dagger : \mathsf{X}(M) \rightleftarrows \prod_{M\in\CC}\mathsf{X}(M) : p_M$,
for all $M\in\CC$, whose right adjoints project onto the $M$-component. 
As a consequence of colimit preservation of $U$ and $F$, it suffices
to check the component-wise condition that 
$p_N UFp_M^\dagger(i) :p_N UFp_M^\dagger(x) \to p_N UF p_M^\dagger(x^\prime)$
is a (acyclic) cofibration in $\mathsf{X}(N)$, for all $M,N\in\CC$ and all
$i:x\to x^\prime$ in $I_M$ (respectively, in $J_M$). One shows that
\begin{flalign}
p_N UFp_M^\dagger(i) \,\cong\,\coprod_{f\in\CC(M,N)} \mathsf{X}^\dagger(f)(i)\quad,
\end{flalign}
where $\mathsf{X}^\dagger(f) : \mathsf{X}(M) \rightleftarrows \mathsf{X}(N):\mathsf{X}(f)$
denotes the left adjoint of $\mathsf{X}(f)$. The proof then follows from the fact that
$\mathsf{X}^\dagger(f)$ is a left Quillen functor between semi-model categories, 
hence it preserves (acyclic) cofibrations with cofibrant domains,
and so does the coproduct $\coprod_{f\in\CC(M,N)}$.
\end{proof}

Combining the result about left Bousfield localizations 
of semi-model categories from \cite[Theorem A]{BataninWhite} with 
the constructions in \cite[Theorem 4.38]{Barwick}, we 
obtain the following result.
\begin{theo}[\cite{BataninWhite,Barwick}]\label{theo:homotopicalpoints}
Let $\mathsf{X} : \CC^\op\to \CombR$ be a $2$-functor and
denote by $\SectR(\mathsf{X})$ the projective semi-model category of right sections
from Proposition \ref{prop:SectRprojective}.
\begin{itemize}
\item[(a)] There exists a subset $S\subseteq \mathrm{Mor}\big(\SectR(\mathsf{X})\big)$
such that an object $(\{x_M\},\{\psi_f\})\in \SectR(\mathsf{X})$ is $S$-local if and only if 
it is homotopy cartesian, i.e., for every $\CC$-morphism $f: M\to N$,
the composite $\mathsf{X}(M)$-morphism 
\begin{flalign}\label{eqn:homotopycartesian}
\xymatrix@C=2.5em{
x_M \ar[r]^-{\psi_f} ~&~ \mathsf{X}(f)(x_N) \ar[r]^-{\mathsf{X}(f)(r)}~&~\mathsf{X}(f)R(x_N)\,=:\,\bbR\mathsf{X}(f)(x_N)
}
\end{flalign} 
is a weak equivalence,
where $r:x_N\to R(x_N) $ is any choice of fibrant replacement in $\mathsf{X}(N)$.

\item[(b)] The left Bousfield localization 
\begin{flalign}\label{eqn:homotopicalpoints}
\mathsf{X}\{\mathrm{pt}\}\,:=\, \mathcal{L}_{S}\SectR(\mathsf{X})
\end{flalign}
of the projective semi-model category $\SectR(\mathsf{X})$ 
at the set of morphisms from item (a) exists as a combinatorial
and tractable semi-model category. The fibrant objects in this 
semi-model category are precisely the homotopy cartesian
right sections $(\{x_M\},\{\psi_f\})$ with $x_M\in\mathsf{X}(M)$ fibrant
objects, for all $M\in\CC$.
\end{itemize}
\end{theo}

\begin{rem}\label{rem:homotopicalpoints}
Note that the fibrant
objects $(\{x_M\},\{\psi_f\})\in\mathsf{X}\{\pt\}$ can be characterized equivalently 
by the property that $\psi_f : x_M \to \mathsf{X}(f)(x_N)$
is a weak equivalence in $\mathsf{X}(M)$, for all $M\in\CC$, because one can choose the trivial 
fibrant replacements $r=\id: x_N \to R(x_N) = x_N$ in this case. 
This property is a homotopical generalization of the isomorphism
property in the ordinary categories of points from Remarks \ref{rem:HKpoints} and \ref{rem:CGpoints}.
This justifies to call the left Bousfield localization
$\mathsf{X}\{\pt\}$ from \eqref{eqn:homotopicalpoints} 
the \textit{semi-model category of homotopical points} of $\mathsf{X}$.
\end{rem}

We conclude this subsection by studying the behavior of the semi-model
categories $\SectR(\mathsf{X})$ and $\mathsf{X}\{\pt\}$ under pseudo-natural
right Quillen equivalences.
Let $\Theta : \mathsf{X}\Rightarrow \mathsf{X}^\prime$ be a pseudo-natural transformation
(in the usual sense of bicategories) between 
two $2$-functors $\mathsf{X},\mathsf{X}^\prime : \CC^\op \to  \CombR$. 
By Definition \ref{def:Comb} of the $2$-category $\CombR$, 
each component $\Theta_M : \mathsf{X}(M)\to \mathsf{X}^\prime(M)$
is a right Quillen functor, hence it admits a left adjoint
$\Theta_M^\dagger: \mathsf{X}^\prime(M)\to \mathsf{X}(M)$.
These left adjoint functors assemble into a pseudo-natural transformation 
$\Theta^\dagger : \mathsf{X}^\prime\Rightarrow \mathsf{X}$
going in the opposite direction of $\Theta : \mathsf{X}\Rightarrow \mathsf{X}^\prime$.
By the construction in \cite[Lemma 2.23]{Barwick}, one obtains an induced 
adjunction 
\begin{flalign}\label{eqn:SectRadjuction}
\xymatrix{
\Theta^\dagger_\ast \,:\, \SectR(\mathsf{X}^\prime) \ar@<0.75ex>[r]~&~\ar@<0.75ex>[l] \SectR(\mathsf{X})\,:\,\Theta_\ast
}
\end{flalign}
between the categories of right sections from Definition \ref{def:rightsections}.
Explicitly, the right adjoint functor $\Theta_\ast$ maps an object
$\big(\{x_M\},\{\psi_f\}\big)$ in $\SectR(\mathsf{X})$
to the object 
\begin{flalign}
\Big(\big\{\Theta_M(x_M)\big\},\, \big\{\Theta_M(\psi_f) : \Theta_M(x_M)\longrightarrow \Theta_M\mathsf{X}(f)(x_N)\cong
\mathsf{X}^\prime(f)\Theta_N(x_N)\big\}\Big)
\end{flalign} 
in $\SectR(\mathsf{X}^\prime)$, 
where the last isomorphism uses the pseudo-naturality structure of $\Theta$.
It further maps a morphism $\{\zeta_M\}$ in $\SectR(\mathsf{X})$ to the morphism
$\{\Theta_M(\zeta_M)\}$ in $\SectR(\mathsf{X}^\prime)$. The left adjoint functor
$\Theta_\ast^\dagger$ is defined similarly. 
\begin{propo}\label{prop:SectRQuillen}
Let $\Theta : \mathsf{X}\Rightarrow \mathsf{X}^\prime$ 
be a pseudo-natural transformation between two $2$-functors 
$\mathsf{X},\mathsf{X}^\prime : \CC^\op \to \CombR$. Then the adjunction
\eqref{eqn:SectRadjuction} is a Quillen adjunction between the projective semi-model
categories from Proposition \ref{prop:SectRprojective}. In the case where 
$\Theta : \mathsf{X}\Rightarrow \mathsf{X}^\prime$ 
is a pseudo-natural right Quillen equivalence, i.e.\ the components
$\Theta_M : \mathsf{X}(M)\to \mathsf{X}^\prime(M)$ are right Quillen equivalences,
for all $M\in\CC$, then \eqref{eqn:SectRadjuction} is a Quillen equivalence.
\end{propo}
\begin{proof}
To prove the first statement, we observe that the right adjoint 
$\Theta_\ast:\SectR(\mathsf{X})\to \SectR(\mathsf{X}^\prime)$ 
preserves fibrations and acyclic fibrations because these are defined component-wise 
in the projective semi-model structure from Proposition \ref{prop:SectRprojective}
and the components $\Theta_M : \mathsf{X}(M)\to \mathsf{X}^\prime(M)$ are 
right Quillen functors, for all $M\in\CC$, by Definition \ref{def:Comb} of $\CombR$.
\sk

To prove the second statement, i.e.\ $\Theta^{\dagger}_{\ast}\dashv \Theta_{\ast}$ 
in \eqref{eqn:SectRadjuction} is a Quillen equivalence  
provided that all components $\Theta_M : \mathsf{X}(M)\to \mathsf{X}^\prime(M)$ are right Quillen equivalences,
we have to verify the following condition: Given any cofibrant object
$(\{x_M^\prime\},\{\psi_f^\prime\})\in \SectR(\mathsf{X}^\prime)$ and any fibrant object
$(\{x_M\},\{\psi_f\})\in \SectR(\mathsf{X})$,
a $\SectR(\mathsf{X}^\prime)$-morphism $\{\zeta_M^\prime\} : 
(\{x_M^\prime\},\{\psi_f^\prime\})\to \Theta_\ast(\{x_M\},\{\psi_f\})$ is a  weak equivalence
if and only if the $\SectR(\mathsf{X})$-morphism 
$\{\zeta^{\prime}_M\!\,^{\dagger}\} : \Theta_\ast^\dagger(\{x_M^\prime\},\{\psi_f^\prime\})\to (\{x_M\},\{\psi_f\})$
which is obtained via the adjunction $\Theta^{\dagger}_{\ast}\dashv \Theta_{\ast}$ is a weak equivalence.
Due to the component-wise definition of weak equivalences in the projective semi-model structure 
from Proposition \ref{prop:SectRprojective}, this is equivalent to checking the following condition:
\begin{flalign}\label{eqn:Quillenequivalencecomponents}
\begin{gathered}
\xymatrix@R=1.5em{
\big(\text{$\zeta_M^\prime : x_M^\prime \to \Theta_M(x_M)$ 
is a weak equivalence in $\mathsf{X}^\prime(M)$}~~\forall M\in\CC\big)\ar@{<=>}[d]\\
\big(\text{$\zeta^{\prime}_M\!\,^{\dagger} : \Theta_M^\dagger(x_M^\prime) \to x_M$ 
is a weak equivalence in $\mathsf{X}(M)$}~~\forall M\in\CC\big)
}
\end{gathered}\quad.
\end{flalign}
In the projective semi-model structure from Proposition \ref{prop:SectRprojective},
the fibrant objects are precisely the component-wise fibrant objects
and the cofibrant objects are in particular component-wise cofibrant.
The condition \eqref{eqn:Quillenequivalencecomponents} then
follows from the hypothesis that $\Theta_M^\dagger \dashv \Theta_M$ is
a Quillen equivalence, for all $M\in\CC$.
\end{proof}

A similar result holds for the semi-model category of homotopical points
from Theorem \ref{theo:homotopicalpoints}.
\begin{theo}\label{theo:homotopicalpointsQuillen}
Let $\Theta : \mathsf{X}\Rightarrow \mathsf{X}^\prime$ 
be a pseudo-natural transformation between two $2$-functors 
$\mathsf{X},\mathsf{X}^\prime : \CC^\op \to \CombR$. Then the Quillen adjunction
from Proposition \ref{prop:SectRQuillen} induces a Quillen adjunction
\begin{flalign}\label{eqn:homotopicalpointsadjuction}
\xymatrix{
\Theta^\dagger_\ast \,:\, \mathsf{X}^\prime\{\pt\} \ar@<0.75ex>[r]~&~\ar@<0.75ex>[l] \mathsf{X}\{\pt\}\,:\,\Theta_\ast
}
\end{flalign}
between the left Bousfield localized semi-model categories from Theorem \ref{theo:homotopicalpoints}.
In the case where $\Theta : \mathsf{X}\Rightarrow \mathsf{X}^\prime$ 
is a pseudo-natural right Quillen equivalence, then \eqref{eqn:homotopicalpointsadjuction} 
is a Quillen equivalence.
\end{theo}
\begin{proof}
To prove the first statement, we use the
characterization of Quillen adjunctions for left Bousfield
localizations of combinatorial and tractable semi-model categories
from \cite[Lemma 3.5]{CarmonaCriterion}. This amounts to verifying 
that the right adjoint functor $\Theta_\ast : \SectR(\mathsf{X})\to\SectR(\mathsf{X}^\prime)$
sends every component-wise fibrant homotopy cartesian right section $(\{x_M\},\{\psi_f\})\in \SectR(\mathsf{X})$
to a homotopy cartesian right section in $\SectR(\mathsf{X}^\prime)$. 
As a consequence of component-wise fibrancy, the fibrant replacements in \eqref{eqn:homotopycartesian} 
can be chosen to be trivial $r=\id$, hence $\psi_f : x_M\to \mathsf{X}(f)(x_N)$ is a weak equivalence
in $\mathsf{X}(M)$, for all $\CC$-morphisms $f:M\to N$.
Using the explicit description of $\Theta_\ast$ which is given below \eqref{eqn:SectRadjuction},
we observe that $\Theta_\ast(\{x_M\},\{\psi_f\})\in \SectR(\mathsf{X}^\prime)$
is homotopy cartesian because $\Theta_M(\psi_f) : \Theta_M(x_M)\to 
\Theta_M\mathsf{X}(f)(x_N)\cong \mathsf{X}(f)\Theta_N(x_N)$ is a weak equivalence
in $\mathsf{X}^\prime(M)$, for all $\CC$-morphisms $f:M\to N$, and $\Theta_N(x_N)\in\mathsf{X}^\prime(N)$
is fibrant, for all $N\in\CC$, as a consequence of $\Theta_M$, $\Theta_N$ 
and $\mathsf{X}(f)$ being right Quillen functors.
\sk

Let us now prove the second claim. Recall
from Proposition \ref{prop:SectRQuillen} that, for 
$\Theta : \mathsf{X}\Rightarrow \mathsf{X}^\prime$ a pseudo-natural 
right Quillen equivalence, the adjunction
$\Theta^\dagger_\ast : \SectR(\mathsf{X}^\prime) \rightleftarrows \SectR(\mathsf{X}) : 
\Theta_\ast$ is a Quillen equivalence between the projective semi-model categories. Recall 
further from Theorem \ref{theo:homotopicalpoints} that $\mathsf{X}\{\pt\} = \mathcal{L}_S\SectR(\mathsf{X})$
and $\mathsf{X}^\prime\{\pt\} =\mathcal{L}_{S^\prime}\SectR(\mathsf{X}^\prime)$ are defined 
as left Bousfield localizations at appropriate sets of morphisms which are specified explicitly
in \cite[Theorem 4.38]{Barwick}, see also the paragraph below. Using the result in \cite[Theorem 3.3.20]{Hirschhorn},
which also holds true for semi-model categories, we have a Quillen equivalence 
$\Theta^\dagger_\ast : \mathsf{X}^\prime\{\pt\} \rightleftarrows 
\mathcal{L}_{\mathbb{L}\Theta^\dagger_\ast(S^\prime)}\SectR(\mathsf{X}) : \Theta_\ast$, where
$\mathbb{L}\Theta^\dagger_\ast(S^\prime)\subseteq \mathrm{Mor}\big(\SectR(\mathsf{X})\big)$ denotes
the image of the localizing set $S^\prime \subseteq \mathrm{Mor}\big(\SectR(\mathsf{X}^\prime)\big)$
under the left derived functor $\mathbb{L}\Theta^\dagger_\ast$. Hence, the result would follow if we can prove
that the two sets $\mathbb{L}\Theta^\dagger_\ast(S^\prime)$ and $S$ determine the same left Bousfield localization,
i.e.\ $\mathbb{L}\Theta^\dagger_\ast(S^\prime)$-local objects are homotopy cartesian right sections.
\sk

To complete the proof, we have to recall from \cite[Theorem 4.38]{Barwick} the definition of the 
localizing sets $S^{(\prime)} \subseteq \mathrm{Mor}\big(\SectR(\mathsf{X}^{(\prime)})\big)$.
As a first step, we observe that there exists, for each object $M\in\CC$, a Quillen adjunction
\begin{flalign}
\xymatrix{
\pi_M^\dagger \,:\,\mathsf{X}(M) \ar@<0.75ex>[r]~&~\ar@<0.75ex>[l] \SectR(\mathsf{X})
\,:\,\pi_M
}
\end{flalign}
whose right Quillen functor $\pi_M$ projects a right section $(\{x_M\},\{\psi_f\})\in 
\SectR(\mathsf{X})$ to its $M$-component $x_M\in\mathsf{X}(M)$. 
Given two objects $M,N\in\CC$ and any $x\in \mathsf{X}(M)$, one has that
\begin{flalign}
\pi_N\pi_M^{\dagger}(x)\,\cong\,\coprod_{f\in\CC(M,N)} \mathsf{X}^\dagger(f)(x)\quad,
\end{flalign}
where $\mathsf{X}^\dagger(f) : \mathsf{X}(M) \rightleftarrows \mathsf{X}(N):\mathsf{X}(f)$
denotes the left adjoint of $\mathsf{X}(f)$. Hence, for every $\CC$-morphism $f:M\to N$,
there exists a canonical inclusion $\mathsf{X}(N)$-morphism $\mathsf{X}^\dagger(f)(x) \to \pi_N\pi_M^{\dagger}(x)$
and we denote its adjunct $\mathsf{X}(M)$-morphism by 
$r_{f,x} : \pi_N^\dagger \mathsf{X}^\dagger(f)(x)\to \pi_M^{\dagger}(x)$. The localizing set 
$S\in \mathrm{Mor}\big(\SectR(\mathsf{X})\big)$ is defined by
\begin{flalign}\label{eqn:Sset}
S\,:=\,\Big\{r_{f,x} : \pi_N^\dagger \mathsf{X}^\dagger(f)(x)\to \pi_M^{\dagger}(x) ~\Big\vert~ (f:M\to N)\in\mathrm{Mor}(\CC)\,,~x\in \mathsf{G}(M)\Big\}\quad,
\end{flalign}
where $\mathsf{G}(M)\subseteq \mathsf{X}(M)$ is any choice of cofibrant homotopy generators of $\mathsf{X}(M)$.
The localizing set $S^\prime \in \mathrm{Mor}\big(\SectR(\mathsf{X}^\prime)\big)$ is defined similarly by inserting
$^{\prime}$ at the relevant places. Using pseudo-naturality of $\Theta : \mathsf{X}\Rightarrow \mathsf{X}^\prime$
and cofibrancy of the source and target of all $S^\prime$-morphisms, one finds that the image
\begin{flalign}
\mathbb{L}\Theta^\dagger_\ast(S^\prime)\,\simeq\, \Big\{r_{f,\Theta_M^\dagger(x^\prime)} : \pi_N^\dagger \mathsf{X}^\dagger(f)\Theta_M^\dagger(x^\prime)\to \pi_M^{\dagger}\Theta_M^\dagger(x^\prime) ~\Big\vert~ (f:M\to N)\in\mathrm{Mor}(\CC)\,,~x^\prime
\in \mathsf{G}^\prime(M)\Big\}
\end{flalign}
is weakly equivalent to a set of the form \eqref{eqn:Sset}, with 
$x\in \Theta_M^\dagger\big(\mathsf{G}^\prime(M)\big)$ now running over 
the image of the cofibrant homotopy generators of $\mathsf{X}^\prime(M)$.
Because $\Theta_M^\dagger$ is by hypothesis a left Quillen equivalence, for all
$M\in\CC$, it follows that $\Theta_M^\dagger\big(\mathsf{G}^\prime(M)\big)$ is a choice
of cofibrant homotopy generators of $\mathsf{X}(M)$. Hence, $\mathbb{L}\Theta^\dagger_\ast(S^\prime)$
and $S$ determine the same left Bousfield localization, which completes the proof.
\end{proof}

\subsection{\label{subsec:homotopicaldcas}Homotopical decomposition and assembly}
In this subsection we develop a homotopical generalization of the decomposition
and assembly functors for AQFTs (see \eqref{eqn:dc} and \eqref{eqn:as})
and for tPFAs (see \eqref{eqn:dctPFA} and \eqref{eqn:astPFA}). Our first 
observation is that these functors extend to the category of right sections from
Definition \ref{def:rightsections} because invertibility of the morphisms $\alpha_f$ and $\phi_f$
is not required. (Recall from Remark \ref{rem:SectR} that this invertibility
is the only difference between the categories of points and the categories of right sections.)
Hence, we obtain decomposition and assembly functors
\begin{subequations}\label{eqn:dcasSect}
\begin{flalign}
&& \dc\,:~&\, \AQFT^\rc \longrightarrow~\SectR(\HK)\quad, 
~& 	\dc\,:~&\, \tPFA^\rc \longrightarrow~\SectR(\CG)\quad , &\\
&& \as \,:~&\, \SectR(\HK)\longrightarrow~\AQFT^\rc\quad, 
~& \as \,:~&\, \SectR(\CG)\longrightarrow~\tPFA^\rc\quad. &
\end{flalign}
\end{subequations}
For the compositions of these functors we have that
\begin{flalign}\label{eqn:dcascompositions}
\as\circ \dc \,=\,\id \quad,\qquad \dc\circ \as \,\Longrightarrow\, \id\quad,
\end{flalign}
where the natural transformations are the ones constructed in Theorems \ref{theo:dcasAQFT} and \ref{theo:dcastPFA}.
It is worthwhile to observe that the latter are \textit{not} natural isomorphisms
because the morphisms $\alpha_f$ and $\phi_f$ are not necessarily invertible
in the categories of right sections. In other words, the generalized decomposition and assembly functors
do \textit{not} induce ordinary equivalences of categories. The main result of this subsection 
is that these functors induce Quillen equivalences
between the relevant semi-model categories, see Corollary \ref{cor:dcQuillenequivalence} 
and Theorem \ref{theo:dcQuillenequivalenceW} below.
\begin{lem}\label{lem:dcasQuillen}
\begin{itemize}
\item[(a)] All four functors in \eqref{eqn:dcasSect} are right Quillen functors
for the projective model structures from Example \ref{ex:projectivemodel} 
and the projective (semi-)model structures on $\SectR(\HK)$ and $\SectR(\CG)$ 
from Proposition \ref{prop:SectRprojective}.

\item[(b)] The right Quillen functors from item (a) induce right Quillen functors
\begin{subequations}\label{eqn:dcashomotopicalpoints}
\begin{flalign}
&& \dc\,:~&\, \AQFT^\rc \longrightarrow~\HK\{\pt\}\quad, 
~& 	\dc\,:~&\, \tPFA^\rc \longrightarrow~\CG\{\pt\}\quad, &\\
&& \as \,:~&\, \HK\{\pt\}\longrightarrow~\AQFT^\rc\quad, 
~& \as \,:~&\, \CG\{\pt\}\longrightarrow~\tPFA^\rc\quad, &
\end{flalign}
\end{subequations}
for the left Bousfield localized semi-model structures on
$\HK\{\pt\}$ and $\CG\{\pt\}$ from Theorem \ref{theo:homotopicalpoints}.

\item[(c)] Restricting the right Quillen functors \eqref{eqn:dcashomotopicalpoints}
to the full subcategories of fibrant objects, the natural transformations in \eqref{eqn:dcascompositions} 
are natural weak equivalences. Hence, these restricted right Quillen functors are weakly quasi-inverse to each other.
\end{itemize}
\end{lem}
\begin{proof}
We start with observing that all four functors in \eqref{eqn:dcasSect} preserve
limits and filtered colimits because these are computed component-wise in the underling category $\TT=\Ch_R$.
Hence, by the special adjoint functor theorem for locally presentable categories, they are right adjoint functors.
Item (a) is a direct consequence of the fact that fibrations and acyclic fibrations are defined component-wise
in the projective (semi-)model structures. 
\sk

To prove item (b), we use \cite[Lemma 3.5]{CarmonaCriterion}, i.e.\ we have to show that the decomposition functors
send fibrant objects to component-wise fibrant homotopy cartesian right sections and that
the assembly functors send component-wise fibrant homotopy cartesian right sections to fibrant objects.
Note that the second statement is automatic because all objects in the projective model categories from 
Theorem \ref{theo:projectivemodel} are fibrant. To show the first statement, recall the definition
of the decomposition functor \eqref{eqn:dc} for AQFTs and observe that 
\begin{flalign}
\xymatrix{
\AAA\vert_M \ar@{=>}[r]^-{\cong}~&~f^\ast(\AAA\vert_N)\,=\,\bbR f^\ast(\AAA\vert_N)
}
\end{flalign}
is an isomorphism (hence a weak equivalence) in $\HK(M)$, for all $\Loc^\rc$-morphisms $f:M\to N$,
where in the last step we used that $\AAA\vert_N\in\HK(N)$ is a fibrant object.
The same argument applies to the decomposition functor \eqref{eqn:dctPFA} for tPFAs.
\sk

Item (c): The components of these natural transformations are specified in Theorems \ref{theo:dcasAQFT}
and \ref{theo:dcastPFA}. They are given by $(\alpha_{\iota_{U}^{M}})_U : \AAA_U(U)\to \AAA_M(U)$
for AQFTs and by $(\phi_{\iota_{U}^{M}})_U : \FFF_U(U)\to \FFF_M(U)$ for tPFAs.
These morphisms are weak equivalences as a consequence of the component-wise fibrancy and the 
homotopy cartesian property from Theorem \ref{theo:homotopicalpoints}, see also Remark \ref{rem:homotopicalpoints}.
\end{proof}
\begin{cor}\label{cor:dcQuillenequivalence}
The decomposition and assembly functors \eqref{eqn:dcashomotopicalpoints}
are right Quillen equivalences between the projective model 
categories $\AQFT^\rc$ and $\tPFA^\rc$ from Example \ref{ex:projectivemodel}
and the semi-model categories of homotopical points $\HK\{\pt\}$ and $\CG\{\pt\}$
from Theorem \ref{theo:homotopicalpoints}.
\end{cor}
\begin{proof}
It suffices to spell out the details for AQFTs since the proof for tPFAs is identical.
\sk

By Lemma \ref{lem:dcasQuillen} (b), the decomposition and assembly functors are right Quillen functors.
To show that they are right Quillen equivalences, we have to prove that the induced functors
\begin{subequations}
\begin{flalign}
\mathsf{Ho}(\dc)\,:~&\,\mathsf{Ho}\big(\AQFT^\rc\big) \longrightarrow~\mathsf{Ho}\big(\HK\{\pt\}\big)\quad,\\
\mathsf{Ho}(\as)\,:~&\,\mathsf{Ho}\big(\HK\{\pt\}\big)\longrightarrow~ \mathsf{Ho}\big(\AQFT^\rc\big)
\end{flalign}
\end{subequations}
exhibit equivalences between the homotopy categories. This follows from our result in Lemma \ref{lem:dcasQuillen} (c)
and the fact that the homotopy category $\mathsf{Ho}(\mathbf{M})\simeq\mathsf{Ho}(\mathbf{M}_f) $ 
of any semi-model category $\mathbf{M}$ can be determined (up to equivalence) from the full 
subcategory $\mathbf{M}_f\subseteq \mathbf{M}$ of fibrant objects.
\end{proof}

The result of Corollary \ref{cor:dcQuillenequivalence} generalizes to
the case of AQFTs and tPFAs satisfying the homotopy time-slice axiom
from Example \ref{ex:Bousfieldmodel}.
\begin{theo}\label{theo:dcQuillenequivalenceW}
The decomposition and assembly functors
\begin{subequations}\label{eqn:dcashomotopicalpointsL}
\begin{flalign}
&& \dc\,:~&\, \mathcal{L}_{\widehat{W}}\AQFT^\rc \longrightarrow~\mathcal{L}_{\widehat{W}}\HK\{\pt\}\quad, 
~& 	\dc\,:~&\, \mathcal{L}_{\widehat{W}}\tPFA^\rc \longrightarrow~\mathcal{L}_{\widehat{W}}\CG\{\pt\}\quad, &\\
&& \as \,:~&\, \mathcal{L}_{\widehat{W}}\HK\{\pt\}\longrightarrow~\mathcal{L}_{\widehat{W}}\AQFT^\rc\quad, 
~& \as \,:~&\, \mathcal{L}_{\widehat{W}}\CG\{\pt\}\longrightarrow~\mathcal{L}_{\widehat{W}}\tPFA^\rc\quad, &
\end{flalign}
\end{subequations}
are right Quillen equivalences between the left Bousfield localized semi-model 
categories $\mathcal{L}_{\widehat{W}}\AQFT^\rc$ and $\mathcal{L}_{\widehat{W}}\tPFA^\rc$ 
from Example \ref{ex:Bousfieldmodel}
and the semi-model categories of homotopical points 
$\mathcal{L}_{\widehat{W}}\HK\{\pt\}$ and $\mathcal{L}_{\widehat{W}}\CG\{\pt\}$
from Theorem \ref{theo:homotopicalpoints} of the $2$-functors
$\mathcal{L}_{\widehat{W}}\HK$ and $\mathcal{L}_{\widehat{W}}\CG$
from \eqref{eqn:HKCGCombR}.
\end{theo}
\begin{proof}
It suffices to spell out the details for AQFTs since the proof for tPFAs is identical.
\sk

Let us start with observing that the semi-model categories
$\mathcal{L}_{\widehat{W}}\HK\{\pt\}$ and $\SectR(\HK)$ have 
the same cofibrations. (This follows by construction 
of left Bousfield localizations and the fact that cofibrations 
(respectively, acyclic fibrations)
are characteried by the left (respectively, right) lifting property 
against acyclic fibrations (respectively, cofibrations) 
\cite[Lemma 1.7]{Barwick}.)
Furthermore, the weak equivalences of $\mathcal{L}_{\widehat{W}}\HK\{\pt\}$ contain
the weak equivalences of $\SectR(\HK)$. This implies that
$\mathcal{L}_{\widehat{W}}\HK\{\pt\}$ is a left Bousfield localization of $\SectR(\HK)$.
\sk

With the above observation, we can provide a proof by using
the characterization of Quillen adjunctions for left Bousfield
localizations of combinatorial and tractable semi-model categories
from \cite[Lemma 3.5]{CarmonaCriterion}. For this it is crucial to recall
that the fibrant objects in $\mathcal{L}_{\widehat{W}}\AQFT^\rc$ are
AQFTs $\AAA\in \AQFT^\rc$ over $\Loc^\rc$ which satisfy the homotopy time-slice axiom,
see Theorem \ref{theo:Bousfieldmodel} and Example \ref{ex:Bousfieldmodel}. 
On the other hand, the fibrant objects in $\mathcal{L}_{\widehat{W}}\HK\{\pt\}$ are right sections
$\big(\{\AAA_M\},\{\alpha_f\}\big)\in \SectR(\HK)$ such that 1.)~$\AAA_M\in\HK(M)$ satisfies
the homotopy time-slice axiom, for all $M\in\Loc^\rc$, and 2.) $\alpha_f : \AAA_M\Rightarrow f^\ast(\AAA_N)$
is a weak equivalence in the projective model structure
$\HK(M)$, for all $\Loc^\rc$-morphisms $f:M\to N$.
It is now straightforward to check that the decomposition \eqref{eqn:dc}
and assembly \eqref{eqn:as} functors both preserve fibrant objects, hence
they define right Quillen functors
\begin{flalign}
\dc\,:\, \mathcal{L}_{\widehat{W}}\AQFT^\rc ~\longrightarrow~\mathcal{L}_{\widehat{W}}\HK\{\pt\}\quad\text{and}\quad
\as\,:\, \mathcal{L}_{\widehat{W}}\HK\{\pt\} ~\longrightarrow~ \mathcal{L}_{\widehat{W}}\AQFT^\rc
\end{flalign}
as a consequence of Lemma \ref{lem:dcasQuillen} (a) and \cite[Lemma 3.5]{CarmonaCriterion}.
Using Lemma \ref{lem:dcasQuillen} (c), one finds that the restrictions of these right Quillen functors
to the full subcategories of fibrant objects are weakly quasi-inverse to each other. The proof
then follows from the same argument as in Corollary \ref{cor:dcQuillenequivalence}.
\end{proof}

\subsection{\label{subsec:reductionmodelcat}Main result}
Recall from Example \ref{ex:Bousfieldmodel} that the tPFA/AQFT-comparison
multifunctor $\Phi : \tP_{\Loc^\rc}\to\O_{\ovr{\Loc^\rc}}$ 
from Definition \ref{def:comparisonoperadmorphism} defines
a right Quillen functor
\begin{flalign}\label{eqn:globalproblem}
\Phi^\ast\,:\,\mathcal{L}_{\widehat{W}}\AQFT^\rc ~\longrightarrow~
\mathcal{L}_{\widehat{W}}\tPFA^\rc
\end{flalign}
from the semi-model category $\mathcal{L}_{\widehat{W}}\AQFT^\rc$ 
whose fibrant objects are $\Ch_R$-valued AQFTs over $\Loc^\rc$ 
satisfying the homotopy time-slice axiom
to the semi-model category $\mathcal{L}_{\widehat{W}}\tPFA^\rc$ 
whose fibrant objects are $\Ch_R$-valued tPFAs over $\Loc^\rc$ 
satisfying the homotopy time-slice axiom.
Similarly, for each object $M\in\Loc^\rc$, the spacetime-wise 
tPFA/AQFT-comparison multifunctor $\Phi_M : \tP_{M}\to\O_{M}$ 
defines a right Quillen functor 
\begin{flalign}\label{eqn:spacetimewiseproblem}
\Phi_M^\ast\,:\, \mathcal{L}_{\widehat{W}_M}\HK(M) 
~\longrightarrow~ \mathcal{L}_{\widehat{W}_M}\CG(M)
\end{flalign}
for $\Ch_R$-valued AQFTs and tPFAs over $\Loc^\rc/M$ satisfying the homotopy time-slice axiom.
\begin{theo}\label{theo:homotopicalreduction}
Suppose that the right Quillen functors \eqref{eqn:spacetimewiseproblem} are right Quillen
equivalences, for all $M\in\Loc^\rc$. Then the right Quillen functor
\eqref{eqn:globalproblem} is a right Quillen equivalence too.
\end{theo}
\begin{rem}
In simpler words: The global homotopical equivalence problem 
for $\Ch_R$-valued AQFTs and tPFAs over $\Loc^\rc$
satisfying the homotopy time-slice axiom can be reduced to
a family of simpler spacetime-wise homotopical
equivalence problems for $\Ch_R$-valued AQFTs and tPFAs over $\Loc^\rc/M$
satisfying the homotopy time-slice axiom, for all $M\in\Loc^\rc$. 
In Subsection \ref{subsec:simplification}, the latter spacetime-wise 
homotopical equivalence problem will be simplified further by leveraging 
a spacetime-wise strictification theorem for the homotopy time-slice axiom 
in the case of AQFTs, see Theorem \ref{theo:O_MCLF} and Corollary \ref{cor:simplerspacetimewise}. 
\end{rem}
\begin{proof}
Similarly to Lemma \ref{lem:comparisondiagram}, we 
have a commutative diagram
\begin{flalign}\label{eqn:comparisondiagramhomotopical}
\begin{gathered}
\xymatrix@C=3em{
\ar[d]_-{\dc} \mathcal{L}_{\widehat{W}}\AQFT^{\rc} \ar[r]^-{\Phi^\ast}
~&~\mathcal{L}_{\widehat{W}}\tPFA^{\rc}\ar[d]^-{\dc}\\
\mathcal{L}_{\widehat{W}}\HK\{\pt\}\ar[r]_-{(\Phi^\ast)_\ast}~&~\mathcal{L}_{\widehat{W}}\CG\{\pt\}
}
\end{gathered}
\end{flalign}
of right Quillen functors between semi-model categories, where
$(\Phi^\ast)_\ast$ denotes the right adjoint of the 
construction in \eqref{eqn:SectRadjuction} applied to the $2$-natural transformation $\Phi^\ast : 
\mathcal{L}_{\widehat{W}}\HK\Rightarrow\mathcal{L}_{\widehat{W}}\CG$
whose components are given by \eqref{eqn:spacetimewiseproblem}.
The two vertical arrows in  \eqref{eqn:comparisondiagramhomotopical} 
are right Quillen equivalences as a consequence of Theorem \ref{theo:dcQuillenequivalenceW}.
Using our hypothesis that \eqref{eqn:spacetimewiseproblem} is a right Quillen equivalence, 
for all $M\in\Loc^\rc$, it follows from Theorem \ref{theo:homotopicalpointsQuillen} that the bottom
horizontal arrow in  \eqref{eqn:comparisondiagramhomotopical} is a right Quillen equivalence too.
This implies that the top horizontal arrow in  \eqref{eqn:comparisondiagramhomotopical} is a right Quillen
equivalence, which completes the proof.
\end{proof}

%%%%%%%%%%%%%%%%%%%%%%%%%%%%%%%%%%%%%%%%%%%%%%%%
%%%%%%%%%%%%%%%%%%%%%%%%%%%%%%%%%%%%%%%%%%%%%%%%

\section{\label{sec:spacetimewise}Towards a spacetime-wise homotopical equivalence theorem}
In the light of our reduction Theorem \ref{theo:homotopicalreduction},
it suffices to address the spacetime-wise homotopical equivalence problem,
i.e.\ the question whether \eqref{eqn:spacetimewiseproblem} is a right Quillen equivalence
for every spacetime $M\in\Loc^\rc$, in order to deduce a global homotopical
equivalence theorem between $\Ch_R$-valued AQFTs and tPFAs over $\Loc^\rc$
satisfying the homotopy time-slice axiom. In this section we will 
present non-trivial progress towards proving a spacetime-wise
homotopical equivalence theorem and point out the remaining technical challenges.
\sk

It is important to emphasize that the rather direct proof strategy
we used in the $1$-categorical spacetime-wise equivalence Theorem \ref{theo:equivalenceM},
which is based on the results in \cite{BPScomparison}, does \textit{not} 
admit an evident generalization to the present homotopical context. The reason
is that the construction of the inverse functor $\CG^W(M)\to \HK^W(M)$
manifestly makes use of the strict time-slice axiom to define
the unital associative algebra structures of an AQFT by strictly 
inverting structure maps of the tPFA which correspond to Cauchy morphisms. 
In the homotopical context, one has to work instead with the weaker
homotopy time-slice axiom, which only allows one to quasi-invert
these structure maps, leading to a tower of homotopy coherence data that
is hard to control. This suggests that solving the spacetime-wise
homotopical equivalence problem requires more abstract machinery which is
able to control such homotopy coherences.
We discuss possible proof strategies 
and point out a simplification which is given by a
strictification theorem for the homotopy time-slice axiom of AQFTs
over $\Loc^\rc/M$.

\subsection{\label{subsec:simplification}Simplification by strictifying the AQFT homotopy time-slice axiom}
To simplify the problem of proving that the right Quillen functor \eqref{eqn:spacetimewiseproblem} 
is a right Quillen equivalence, we shall show in this subsection
that the semi-model category $\mathcal{L}_{\widehat{W}_M}\HK(M)$ encoding 
$\Ch_R$-valued AQFTs over $\Loc^\rc/M$ satisfying the \textit{homotopy} 
time-slice axiom is Quillen equivalent to a model category $\HK^W(M)$
of $\Ch_R$-valued AQFTs over $\Loc^\rc/M$ satisfying the \textit{strict} 
time-slice axiom, for all $M\in\Loc^\rc$. The precise definition of the latter
is as follows: Using the concept of localizations of $\Set$-valued operads,
see e.g.\ \cite[Definition 2.12]{BCStimeslice}, one obtains that the full subcategory
$\HK^W(M)\subseteq \HK(M)$ of $\Ch_R$-valued AQFTs over $\Loc^\rc/M$ satisfying the strict
time-slice axiom from Definition \ref{def:O_M} can be presented (up to equivalence) as the category 
\begin{flalign}\label{eqn:AlgOMlocalization}
\HK^W(M)\,\simeq\, \Alg_{\O_M[W_M^{-1}]}\big(\Ch_R\big)
\end{flalign}
of $\Ch_R$-valued algebras over the localization $\O_M[W_M^{-1}]$ 
of the operad $\O_M$ at the set of $1$-ary operations $W_M\subseteq \mathrm{Mor}^1(\O_M)$ 
given by all Cauchy morphisms in $\O_M$. We endow the category of operad algebras \eqref{eqn:AlgOMlocalization} 
with the projective model structure from Theorem \ref{theo:projectivemodel}
and denote the resulting projective model category by the same symbol $\HK^W(M)$.
\sk

By definition of localization of operads \cite[Definition 2.12]{BCStimeslice}, 
we have a localization multifunctor $L_M : \O_M \to \O_M[W_M^{-1}]$ which by 
Proposition \ref{prop:Quillenprojective}
induces a Quillen adjunction $L_{M\,!} : \HK(M) \rightleftarrows \HK^W(M) : L_M^\ast$ 
between the projective model categories. This Quillen adjunction induces a
Quillen adjunction
\begin{flalign}\label{eqn:strictificationadjunction}
\xymatrix{
L_{M\,!}\,:\, \mathcal{L}_{\widehat{W}_M}\HK(M) \ar@<0.75ex>[r]~&~\ar@<0.75ex>[l] \HK^W(M)\,:\,L_M^\ast
}
\end{flalign}
between the projective model category $\HK^W(M)$ of AQFTs satisfying the strict time-slice axiom
and the left Bousfield localized semi-model category  $\mathcal{L}_{\widehat{W}_M}\HK(M)$ from 
Example \ref{ex:Bousfieldmodel} which describes AQFTs satisfying the homotopy time-slice axiom.
This claim follows by using \cite[Lemma 3.5]{CarmonaCriterion} and the fact that the right
adjoint $L_M^\ast$ sends every object in $\HK^W(M)$ to an object in $\HK(M)$ which satisfies
the strict, and hence also the homotopy, time-slice axiom, i.e.\ this object is $\widehat{W}_M$-local
in the sense of Theorem \ref{theo:Bousfieldmodel}.
\sk

It is important to stress that a Quillen adjunction as in \eqref{eqn:strictificationadjunction}
may \textit{not} be a Quillen equivalence because the homotopy time-slice axiom
is a priori richer than the strict one.
If it is a Quillen equivalence, we say that the homotopy time-slice
axiom admits a strictification. Special instances of such strictification theorems have been 
proven in \cite{BCStimeslice}, but these do not apply to our present case since the localization
$L_M : \O_M \to \O_M[W_M^{-1}]$ is not reflective. We develop in Appendix 
\ref{app:operadiclocalization} a new and more general strictification theorem 
which is based on the concept of operadic calculus of left fractions and applies to the present case.
This leads to the following main result of this subsection.
\begin{theo}\label{theo:O_MCLF}
For each $M\in\Loc^\rc$, the pair $(\O_M,W_M)$ consisting of the spacetime-wise
AQFT operad $\O_M$ from Definition \ref{def:O_M} and the subset of $1$-ary 
operations $W_M\subseteq \mathrm{Mor}^1(\O_M)$ 
given by all Cauchy morphisms in $\O_M$ admits an operadic calculus of left fractions in the sense
of Definition \ref{def:OpCLF}. In particular, the Quillen 
adjunction \eqref{eqn:strictificationadjunction} is a Quillen equivalence,
for all $M\in\Loc^\rc$.
\end{theo}
\begin{proof}
The last part of the statement follows from Proposition 
\ref{prop:CLFhomotopical} as soon as the pair $(\O_M,W_M)$ admits 
an operadic calculus of left fractions. To prove that this is indeed the case, 
we check the properties from Definition \ref{def:OpCLF}. 
\sk

Properties (1) and (2) are obvious 
from the Definition \ref{def:O_M} of $\O_M$ and the fact that $W_M$ 
consists of Cauchy morphisms. 
To check also property (4), recall that the slice category $\Loc^\rc/M$ is thin, 
i.e.\ there exists at most one morphism between any two objects, 
hence the $n$-ary operations $[\sigma]: \und{U} \to V$ in $\O_M$ are just equivalence classes 
of permutations $\sigma \in \Sigma_n$ under the equivalence relation $\sim_\perp$ 
induced by causal disjointness, see also Definition \ref{def:AQFToperad}. 
Then property (4) follows from the fact that Cauchy morphisms preserve and detect causal disjointness: 
For $U_1 \subseteq V_1$ and $U_2 \subseteq V_2$ Cauchy morphisms in $\Loc^{\rc}/M$, 
the two subsets $V_1 \subseteq M$ and $V_2 \subseteq M$ 
are causally disjoint if and only if $U_1 \subseteq M$ and $U_2 \subseteq M$ are causally disjoint. 
(This follows from the fact that the Cauchy development, whose definition is 
recalled at the beginning of Appendix \ref{app:Lorentz}, preserves and detects causal disjointness.) 
\sk

To complete the proof we have to check also property (3). 
Consider an $n$-ary operation $[\sigma]: \und{U} \to V$ in $\O_M$ 
and a family $U_i \subseteq U_i^\prime$ of Cauchy inclusions, for all $i = 1, \ldots, n$. 
The goal is to construct an $n$-ary operation $[\sigma]: \und{U}^\prime \to V^\prime$ and 
a Cauchy inclusion $V \subseteq V^\prime$ (commutativity of the square of inclusions is automatic). 
This amounts to constructing a causally convex open $V^\prime \subseteq M$ which is either Cauchy or relatively compact 
such that, for all $i = 1, \ldots, n$, one has zig-zags $U_i^\prime \subseteq V^\prime \supseteq V$, 
where $U_i^\prime \subseteq V^\prime$ is either Cauchy or relatively compact and $V \subseteq V^\prime$ is Cauchy. 
By Lemma \ref{lem:zigzag} this is equivalent to checking that the inclusions 
$U_i^\prime \subseteq D_M(V)$ into the Cauchy development of $V$
are either Cauchy or relatively compact, for all $i=1,\ldots,n$.
We make a case distinction. 
If $V \subseteq M$ is Cauchy, it follows that $U_i^\prime \subseteq M = D_M(V)$ 
is either Cauchy or relatively compact (because $U_i^\prime \in \O_M$). 
Otherwise $V \subseteq M$ is not Cauchy, hence it is relatively compact, which implies
that $U_i^\prime \subseteq M$ is relatively compact too. 
(Indeed, if $U_i^\prime \subseteq M$ were not relatively compact, 
it would be Cauchy, hence $U_i \subseteq U_i^\prime \subseteq M$ would be Cauchy too, 
leading to a contradiction with the inclusion $U_i \subseteq V$ 
and the hypothesis that $V \subseteq M$ is not Cauchy.) 
Let us show that each $U^\prime_i \subseteq D_M(V)$ 
is either Cauchy or relatively compact by making another case distinction. 
If $U_i \subseteq V$ is Cauchy, 
we deduce $D_M(U_i^\prime) = D_M(U_i) = D_M(V)$,
which proves that $U_i^\prime \subseteq D_M(V)$ is Cauchy 
(by the properties of the Cauchy development $D_M$). 
Otherwise, $U_i \subseteq V$ is not Cauchy, hence $U_i \subseteq V$ is relatively compact. 
The closure $\overline{U_i} \subseteq V$ is then compact and  
it follows from Lemma \ref{lem:Dclosed} that 
$D_M(\overline{U_i}) \subseteq M$ is closed. 
Therefore, for the compactum $\overline{U_i^\prime} \subseteq M$ one has the inclusion 
$\overline{U_i^\prime} \subseteq \overline{D_M(U_i^\prime)} = \overline{D_M(U_i)} 
\subseteq \overline{D_M(\overline{U_i})} = D_M(\overline{U_i}) \subseteq D_M(V)$, 
which proves that $U_i^\prime \subseteq D_M(V)$ is relatively compact.
\end{proof}

\begin{cor}\label{cor:simplerspacetimewise}
The right Quillen functor \eqref{eqn:spacetimewiseproblem} controlling the spacetime-wise
homotopical equivalence problem is a right Quillen equivalence
if and only if the composite right Quillen functor
\begin{flalign}\label{eqn:simplerspacetimewise}
\xymatrix{
(L_M\Phi_M)^\ast \,:\,\HK^W(M)\ar[r]^-{L_M^\ast}~&~ \mathcal{L}_{\widehat{W}_M}\HK(M) \ar[r]^-{\Phi_M^\ast}~&~ \mathcal{L}_{\widehat{W}_M}\CG(M)
}
\end{flalign}
is a right Quillen equivalence.
\end{cor}

\begin{rem}\label{rem:nostrictification}
Operadic calculi of left fractions also exist for the Cauchy morphisms in the 
variations of the spacetime-wise AQFT operads studied in \cite[Appendix B]{BGSHaagKastler}, 
which implies that strictification theorems for the homotopy time-slice axiom are also available in these cases. 
In stark contrast to this, the Cauchy morphisms in the global 
AQFT operad over $\Loc^\rc$ from Definition \ref{def:AQFToperad}
(or in the global AQFT operad over the larger category $\Loc$ from Definition \ref{def:Lorentz}) 
do \textit{not} admit an operadic calculus of left fractions in spacetime dimensions $m\geq 2$,
so the time-slice strictification problem in the global case remains open.
Furthermore, operadic calculi of left fractions do \textit{not} exist
for the Cauchy morphisms in the global tPFA operad over $\Loc^\rc$ 
from Definition \ref{def:tPFAoperad} and for the Cauchy morphisms in spacetime-wise tPFA operads
from Definition \ref{def:tP_M}, hence there are no evident time-slice strictification theorems 
for tPFAs in both the global and the spacetime-wise case.
\end{rem}

\begin{rem}\label{rem:localissimpler}
We can now explain the precise sense in which the spacetime-wise homotopical tPFA/AQFT 
equivalence problem is simpler than the global homotopical equivalence problem.
Analyzing the right Quillen functor \eqref{eqn:globalproblem} which controls the global
homotopical equivalence problem amounts to comparing 
AQFTs and tPFAs over $\Loc^\rc$ that are both satisfying the homotopy time-slice axiom. 
In this case there are no evident strictification theorems for the homotopy time-slice axiom
on either side. In contrast to this, Corollary \ref{cor:simplerspacetimewise}
implies that the spacetime-wise homotopical equivalence problem 
is equivalent to proving that the composite right Quillen functor \eqref{eqn:simplerspacetimewise}
is a right Quillen equivalence. The latter involves comparing
AQFTs over $\Loc^\rc/M$ satisfying the simpler \textit{strict} time-slice axiom with tPFAs over $\Loc^\rc/M$
satisfying the homotopy time-slice axiom, i.e.\ on the AQFT-side all homotopical phenomena
of the homotopy time-slice axiom have been trivialized by the strictification Theorem \ref{theo:O_MCLF}.
Unfortunately, this simplification is insufficient to apply the explicit constructions
from the $1$-categorical spacetime-wise equivalence Theorem \ref{theo:equivalenceM} 
because there is no evident strictification theorem for the homotopy time-slice axiom on the tPFA-side.
\end{rem}

We conclude this subsection by showing that the localized operad $\O_M[W_M^{-1}]$
whose algebras $\HK^W(M) = \Alg_{\O_M[W_M^{-1}]}\big(\Ch_R\big)$ define the source
of the simplified comparison right Quillen functor \eqref{eqn:simplerspacetimewise}
admits a very explicit description.
\begin{propo}\label{prop:O_Mlocalized}
For each $M\in\Loc^\rc$, a model for the localization $\O_M[W_M^{-1}]$ 
of the operad $\O_M$ from Definition \ref{def:O_M} at the subset of $1$-ary operations $W_M\subseteq \mathrm{Mor}^1(\O_M)$ 
consisting of all Cauchy morphisms in $\O_M$ is given by the colored operad
which is defined by the following data:
\begin{itemize}
\item[(i)] The objects are all causally convex opens $U\subseteq M$ which are either
Cauchy or relatively compact.

\item[(ii)] The set of operations from $\und{U} = (U_1,\dots,U_n)$ to $V$ is
\begin{flalign}
\O_M[W_M^{-1}]\big(\substack{V\\\und{U}}\big)\,:=\,\begin{cases}
\Sigma_n\big/\!\sim_{\perp}^{} ~&,~~\text{if }U_i\subseteq D_M(V) \text{ is Cauchy or relatively compact }\\
&~~~\text{for all } i=1,\dots,n\quad,\\[4pt]
\varnothing ~&,~~\text{else}\quad,
\end{cases}
\end{flalign}
where $\Sigma_n$ denotes the permutation group on $n$ letters and $D_M(V)\subseteq M$ denotes the 
Cauchy development of $V\subseteq M$.  The equivalence relation $\sim_\perp^{}$ is defined as follows:
$\sigma\sim_\perp^{} \sigma^\prime$ if and only if the right permutation $\sigma\sigma^{\prime-1} : 
\und{U}\sigma^{-1}\to\und{U}\sigma^{\prime -1}$ is generated by transpositions of adjacent causally 
disjoint subsets of $M$.

\item[(iii)] The composition of $[\sigma] : \und{U}\to V$ with $[\sigma_i]: \und{O}_i\to U_i $, 
for $i=1,\dots,n$, is defined by
\begin{flalign}
[\sigma]\,[\und{\sigma}]\,:=\, \big[\sigma(\sigma_1,\dots,\sigma_n)\big]\,:\,\und{\und{O}}~\longrightarrow~V\quad,
\end{flalign}
where $\sigma(\sigma_1,\dots,\sigma_n)$ denotes the composition in the unital associative operad.

\item[(iv)] The identity operations are $[e] : U\to U $, where $e\in\Sigma_1$ is the identity permutation.

\item[(v)] The permutation action of $\sigma^\prime\in\Sigma_n$ on $[\sigma] : \und{U}\to V$ is given by
\begin{flalign}
[\sigma]\cdot \sigma^\prime\,:=\, [\sigma\sigma^\prime]\,:\, \und{U}\sigma^{\prime}~\longrightarrow~V \quad,
\end{flalign}
where $\und{U}\sigma^{\prime}= (U_{\sigma^\prime(1)},\dots, U_{\sigma^{\prime}(n)})$ denotes 
the permuted tuple and $\sigma\sigma^\prime$ is given by the group operation of the permutation group $\Sigma_n$.
\end{itemize}
The localization multifunctor is given by
\begin{flalign}
L_M\,:\, \O_M~&\longrightarrow~\O_M[W_M^{-1}]\quad,\\
\nn U~&\longmapsto~U\quad,\\
\nn \big([\sigma,\iota_{\und{U}}^V]: \und{U}\to V\big)~&\longmapsto~\big([\sigma]: \und{U}\to V\big)\quad.
\end{flalign}
\end{propo}
\begin{proof}
It was shown in \cite[Propositions 2.11 and 2.14]{BCStimeslice} that localizations of AQFT operads
are completely determined by localizing their underlying subcategories of $1$-ary operations.
Since $(\O_M,W_M)$ admits an operadic calculus of left fractions by Theorem \ref{theo:O_MCLF}, it follows 
from Remark \ref{rem:OpCLF} that its underlying subcategory of $1$-ary operations admits a calculus of 
left fractions $(\O_M^1,W_M)$. The explicit model for the localized operad $\O_M[W_M^{-1}]$ 
then follows by slightly adapting the computations from \cite[Appendix B]{BGSHaagKastler}
of the underlying localized category and its pushforward orthogonality relation
to our context where all inclusion morphisms $\iota_U^V$ must be either Cauchy or relatively compact. 
The necessary changes are implemented by replacing 
\cite[Proposition B.2]{BGSHaagKastler} with Proposition \ref{prop:localizedrcmor}. 
\end{proof}

\subsection{\label{subsec:openproblem}Remaining open problem and technical challenges}
In this subsection we summarize the main achievements of our present paper
and point out the remaining open problem and technical challenges.
Let us recall that the main goal is to prove a homotopical equivalence 
theorem between $\Ch_R$-valued AQFTs and tPFAs over $\Loc^\rc$ satisfying the homotopy time-slice axiom.
This amounts to proving that the tPFA/AQFT-comparison right Quillen functor in 
\eqref{eqn:globalproblem} is a right Quillen equivalence. In Theorem \ref{theo:homotopicalreduction},
we have reduced the global homotopical equivalence problem to a family of
simpler spacetime-wise homotopical equivalence problems for 
$\Ch_R$-valued AQFTs and tPFAs over $\Loc^\rc/M$ satisfying the homotopy time-slice axiom,
for all $M\in\Loc^\rc$. The reduced problem amounts to proving
that the spacetime-wise tPFA/AQFT-comparison right Quillen functors in 
\eqref{eqn:spacetimewiseproblem} are right Quillen equivalences, for all $M\in\Loc^\rc$.
Using our spacetime-wise strictification Theorem \ref{theo:O_MCLF} for the homotopy 
time-slice axiom of AQFTs, this problem can be simplified further to proving that the
partially strictified spacetime-wise tPFA/AQFT-comparison right Quillen functors in 
\eqref{eqn:simplerspacetimewise} are right Quillen equivalences, for all $M\in\Loc^\rc$.
Note that the key difference between \eqref{eqn:spacetimewiseproblem} and \eqref{eqn:simplerspacetimewise}
is that in the latter case the homotopy time-slice axiom for AQFTs over $\Loc^\rc/M$ has been strictified,
i.e.\ on the AQFT-side all homotopical phenomena of the homotopy time-slice axiom have been trivialized.
\sk

In conclusion, the remaining open problem which has to be solved in order
to obtain a global homotopical equivalence theorem 
between $\Ch_R$-valued AQFTs and tPFAs over $\Loc^\rc$ 
satisfying the homotopy time-slice axiom is the following.
\begin{open}\label{open:openproblem}
Prove that the composite right Quillen functor 
\begin{flalign}\label{eqn:openproblem}
\xymatrix{
(L_M\Phi_M)^\ast \,:\,\HK^W(M)\ar[r]^-{L_M^\ast}~&~ \mathcal{L}_{\widehat{W}_M}\HK(M) \ar[r]^-{\Phi_M^\ast}~&~ \mathcal{L}_{\widehat{W}_M}\CG(M)
}
\end{flalign}
from the projective model category $\HK^W(M)\simeq \Alg_{\O_M[W_M^{-1}]}\big(\Ch_R\big)$
from Theorem \ref{theo:projectivemodel} to the left Bousfield localized semi-model category
$\mathcal{L}_{\widehat{W}_M}\CG(M)$ from Theorem \ref{theo:Bousfieldmodel} 
is a right Quillen equivalence, for all objects $M\in\Loc^\rc$.
The spacetime-wise tPFA/AQFT-comparison multifunctor $\Phi_M : \tP_M\to \O_M$
is defined in \eqref{eqn:PhiM}, see also Definition \ref{def:comparisonoperadmorphism}, 
and an explicit model for the localization multifunctor
$L_M : \O_M\to \O_M[W_M^{-1}]$ is given in Proposition \ref{prop:O_Mlocalized}.
\end{open}

There are various approaches one could follow
in order to prove Open Problem \ref{open:openproblem}.
\begin{description}
\item[Homotopical algebra:] Develop models for the derived functors 
associated with the Quillen adjunction \eqref{eqn:openproblem} and prove
that the derived unit and counit are weak equivalences. Since every object in the projective model category
$\HK^W(M)$ is fibrant, the right derived functor $\bbR(L_M\Phi_M)^\ast = (L_M\Phi_M)^\ast 
: \HK^W(M)\to \mathcal{L}_{\widehat{W}_M}\CG(M)$ can be modeled by the underived right Quillen functor.
For the left derived functor $\mathbb{L}(L_M\Phi_M)_! : \mathcal{L}_{\widehat{W}_M}\CG(M)\to \HK^W(M)$
one can use that the cofibrant objects in the projective model category $\CG(M)$
agree with the ones in its left Bousfield localization $\mathcal{L}_{\widehat{W}_M}\CG(M)$,
hence one could use the bar resolution model for operadic left Kan extensions 
from \cite[Theorem 17.2.7 and Section 13.3]{Fresse}. However, proving that the derived unit 
and counit are weak equivalences seems to be very difficult due to the complexity
of such derived functors, hence we are uncertain if this approach will be successful.

\item[Homotopical operadic localization:] Use \cite[Theorem 3.13]{Carmona}
to observe that Open Problem \ref{open:openproblem} would be solved if
the composite multifunctor $L_M\Phi_M : \tP_M \to \O_M[W_M^{-1}]$ exhibits
a homotopical localization of the spacetime-wise tPFA operad $\tP_M $ over $\Loc^\rc/M$
at the subset of $1$-ary operations $W_M\subseteq \mathrm{Mor}^1(\tP_M)$ 
given by all Cauchy morphisms in $\tP_M$. Denoting by $L_{W_M}\tP_{M}$ such a 
homotopical localization as an operad in simplicial sets, this amounts to showing 
that the induced multifunctor $L_{W_M}\tP_{M}\to  \O_M[W_M^{-1}]$ is a weak 
equivalence of operads in simplicial sets. A mild generalization of Theorem \ref{theo:equivalenceM} 
entails that $L_M\Phi_M$ exhibits the ordinary localization of $\tP_M$ at $W_M$, 
hence it suffices to show that all spaces of operations in $L_{W_M}\tP_{M}$ are discrete, 
i.e.\ they have vanishing higher homotopy groups $\pi_{\geq 1}$. 
Since homotopical localizations of operads are currently not yet 
well-developed and studied, we are uncertain if this approach will 
be successful.\footnote{In \cite{OperadicHammock}, 
the authors propose a candidate for $L_{W_M}\tP_{M}$ 
by generalizing the hammock construction of Dwyer-Kan \cite{DwyerKan}
to the operadic setting. However, such a candidate is not known to satisfy 
the appropriate universal property and hence it is currently unknown 
if it presents the homotopical localization of operads.}

\item[$\infty$-categorical localization:] Use similar
arguments as in Appendix \ref{app:operadiclocalization},
which are based on \cite[Corollary 4.11]{Haugseng}, \cite[Theorem 7.3.1]{WhiteYau}
and \cite[Appendix B]{CalaqueCarmona}, to observe that 
Open Problem \ref{open:openproblem} would be solved if 
the functor $(L_M\Phi_M)^\otimes : \tP_M^\otimes \to \O_M[W_M^{-1}]^\otimes$
between the categories of operators (see Definition \ref{def:operatorcategory})
exhibits an $\infty$-localization of the category of operators
$\tP_{M}^\otimes$ of the spacetime-wise tPFA operad over $\Loc^\rc/M$
with respect to the subset $W_M^\otimes\subseteq \Mor\big(\tP_{M}^\otimes\big)$
from Lemma \ref{lem:operatorlocalization}. Similarly to the previous approach 
and by the same mild generalization of Theorem \ref{theo:equivalenceM}, this amounts 
to showing that all spaces of operations in $L_{W_M^{\otimes}}^{H}\tP_{M}^{\otimes}$ 
are discrete, i.e.\ they have vanishing higher homotopy groups $\pi_{\geq 1}$, 
where $L^H$ denotes the classical hammock localization of Dwyer-Kan \cite{DwyerKan}. 
We believe that this is currently the most promising approach, 
because $\infty$-localization of categories is a well-studied subject and there exist explicit
criteria which allow one to detect $\infty$-localizations. 
\end{description}

Following the third approach, we will now show that
the criterion for $\infty$-localizations from \cite{HinichDK}
is inconclusive for our example at hand.
\begin{propo}\label{prop:HinichCriterion}
The sufficient, but not necessary, criterion
from \cite[Key Lemma 1.3.6]{HinichDK} does not
apply to the functor $(L_M\Phi_M)^\otimes : \tP_M^\otimes \to \O_M[W_M^{-1}]^\otimes$.
Hence, it remains undecided whether or not this functor exhibits an $\infty$-localization
of the category of operators $\tP_M^\otimes$ at the
subset $W_M^\otimes\subseteq \mathrm{Mor}\big(\tP_M^\otimes\big)$
defined in Lemma \ref{lem:operatorlocalization}.
\end{propo}
\begin{proof}
We start with spelling out Hinich's criterion
in our context. Let $[n]\in \Delta$
be the $n$-simplex category, i.e.\ $[n]:=\big(0\to 1\to\cdots\to n\big)$, 
for all $n\in \bbZ^{\geq 0}$.
We write $\Fun\big([n],\O_M[W_M^{-1}]^\otimes\big)^{\cong}$
for the category of all functors $[n]\to \O_M[W_M^{-1}]^\otimes$
with morphisms given by all natural \textit{iso}morphisms. Furthermore, we write
$\Fun\big([n],\tP_M^\otimes\big)^{W_M^\otimes}$
for the category of all functors $[n]\to \tP_M^\otimes$
with morphisms given by all natural transformations whose components 
belong to the subset $W_M^\otimes\subseteq \mathrm{Mor}\big(\tP_M^\otimes\big)$
from Lemma \ref{lem:operatorlocalization}. The functor $(L_M\Phi_M)^\otimes : 
\tP_M^\otimes \to \O_M[W_M^{-1}]^\otimes$ induces via pushforward a family of functors
\begin{flalign}
(L_M\Phi_M)^\otimes_\ast\,:\, \Fun\big([n],\tP_M^\otimes\big)^{W_M^\otimes}~\longrightarrow~
\Fun\big([n],\O_M[W_M^{-1}]^\otimes\big)^{\cong}\quad,
\end{flalign}
for all $n\in \bbZ^{\geq 0}$, because $(L_M\Phi_M)^\otimes$ sends $W_M^\otimes$ to isomorphisms. 
The criterion for $\infty$-localizations in 
\cite[Key Lemma 1.3.6]{HinichDK} amounts to checking that the homotopy
fibers of these functors
\begin{flalign}
\begin{gathered}
\xymatrix{
\ar@{-->}[d]\Big(\Fun\big([n],\tP_M^\otimes\big)^{W_M^\otimes}\Big)_\Psi \ar@{-->}[r]~&~ \Fun\big([n],\tP_M^\otimes\big)^{W_M^\otimes}\ar[d]^-{(L_M\Phi_M)^\otimes_\ast}\\
\pt \ar[r]_-{\Psi}~&~ \Fun\big([n],\O_M[W_M^{-1}]^\otimes\big)^{\cong}
}
\end{gathered}
\end{flalign}
have a weakly contractible nerve, for all $n\in \bbZ^{\geq 0}$ and all 
functors $\Psi : [n]\to \O_M[W_M^{-1}]^\otimes$. 
\sk

We will now show that this criterion fails for $n=1$ by exhibiting examples of 
operations $\Psi : [1]\to \O_M[W_M^{-1}]^\otimes$ such that the homotopy fiber is empty.
Consider any binary operation $[\sigma] : (U_1,U_2)\to M$ in
the localized AQFT operad  $\O_M[W_M^{-1}]$ from Proposition \ref{prop:O_Mlocalized}
such that any choice of Cauchy surfaces $\Sigma_1\subset U_1$ and $\Sigma_2\subset U_2$
intersect $\Sigma_1\cap\Sigma_2\neq \varnothing$. An example 
is given by the following two overlapping regions in the $2$-dimensional Minkowski spacetime:
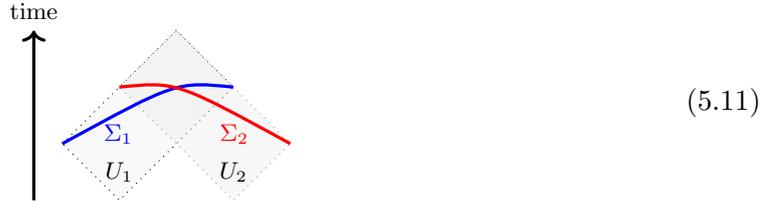
\begin{flalign}\label{eqn:crossingsurfaces}
\begin{gathered}
\begin{tikzpicture}[scale=1.5]
\draw[fill=gray!5,dotted] (-1,0) -- (0,1) -- (0.5,0.5) -- (-0.5,-0.5) -- (-1,0);
\draw[fill=gray!15,dotted,opacity=0.4] (-0.5,0.5) -- (0,1) -- (1,0) -- (0.5,-0.5) -- (-0.5,0.5);
\draw (-0.5,-0.25) node {{\footnotesize $U_1$}};
\draw (0.5,-0.25) node {{\footnotesize $U_2$}}; 
\draw[very thick, blue] (-1,0)  .. controls (0,0.55) .. (0.5, 0.5) node[pos=0.2, below] {{\footnotesize $\Sigma_1$}} ;
\draw[very thick, red] (-0.5,0.5)  .. controls (0,0.55) .. (1, 0) node[pos=0.8, below] {{\footnotesize $\Sigma_2$}} ;
\draw[very thick, ->] (-1.25,-0.5) -- (-1.25,1) node[above] {{\footnotesize time}};
\end{tikzpicture}
\end{gathered}
\end{flalign}
Objects in the homotopy fiber $\big(\Fun\big([1],\tP_M^\otimes\big)^{W_M^\otimes}\big)_{[\sigma]}$ 
over this operation are binary operations $(\iota_{U^\prime_1}^V,\iota_{U^\prime_2}^V) : (U_1^\prime,U_2^\prime)\to V$
in the tPFA operad $\tP_M$ from Definition \ref{def:tP_M} together with $\O_M[W_M^{-1}]^\otimes$-isomorphisms
$(U_1^\prime,U_2^\prime)\cong (U_1,U_2)$ and $V\cong M$ such that the diagram
\begin{flalign}\label{eqn:totuplediagram}
\begin{gathered}
\xymatrix@C=3em{
\ar[d]_-{\cong} (U_1^\prime,U_2^\prime) \ar[r]^-{[\rho^{-1}]}~&~ V\ar[d]^-{\cong}\\
(U_1,U_2) \ar[r]_-{[\sigma]}~&~ M
}
\end{gathered}
\end{flalign}
in $\O_M[W_M^{-1}]^\otimes$ commutes, where $\rho\in \Sigma_2$ is a time-ordering permutation
for the time-orderable tuple $(\iota_{U^\prime_1}^V,\iota_{U^\prime_2}^V)$. Since by our 
hypothesis any choice of Cauchy surfaces $\Sigma_1\subset U_1$ and $\Sigma_2\subset U_2$ intersect,
it is impossible to shrink $(U_1,U_2)$ through Cauchy morphisms to 
a pair of disjoint causally convex open subsets. In particular, there exists
no time-orderable tuple $(\iota_{U^\prime_1}^V,\iota_{U^\prime_2}^V)$ making the 
diagram \eqref{eqn:totuplediagram} commute, hence the homotopy fiber is empty.
\end{proof}

\begin{rem}\label{rem:othercriteria}
We would like to note that there exist also other sufficient, 
but not necessary, criteria to detect $\infty$-localizations,
for instance the criterion in \cite[Appendix A.2]{Scheimbauer} which originated in the work of Ayala and Francis
and the criterion in \cite[Proposition 4.2.18]{Harpaz} which is based on Lurie's concept of 
weak approximations of $\infty$-operads. These criteria are inconclusive for our example due to the
same counterexamples we provided in the proof of Proposition \ref{prop:HinichCriterion}.
\end{rem}

\begin{rem}\label{rem:outlook}
The problem of showing weak contractibility of the nerves of the homotopy fibers 
in Proposition \ref{prop:HinichCriterion} can be rephrased into the following 
equivalent lifting problems
\begin{equation}\label{eqn:liftingproblems}
\begin{tikzcd}
\mathbf{K} \ar[rr] \ar[dd,"!"', ""{name=S}] && \Big(\Fun\big([n],\tP_M^\otimes\big)^{W_M^\otimes}\Big)_\Psi \\\\
\pt \ar[rruu, dashed, ""{name=T, above}] \ar[from=S, to=T, shift left=4, "\substack{~\exists~\\\Leftarrow ~\cdots~ \Rightarrow}" description, phantom]
\end{tikzcd}\qquad,
\end{equation}
where $\mathbf{K}$ is any finite poset and the top horizontal arrow is any functor to the homotopy fiber.
Note that the triangle does not have to commute strictly, but only up to finite zig-zags of natural transformations.
One checks that these lifting problems can be solved for $n=0$, showing that 
all homotopy fibers $\big(\Fun\big([0],\tP_M^\otimes\big)^{W_M^\otimes}\big)_\Psi$ with $[n]=[0]$
are weakly contractible. For $n=1$, our preliminary investigations suggest
that the corresponding lifting problem \eqref{eqn:liftingproblems} can be solved whenever 
the finite poset $\mathbf{K}\neq \varnothing$ is non-empty. 
Our hope is that one can combine such partial lifting results with the conclusion 
of Theorem \ref{theo:equivalenceM}, which implies that the functor 
$(L_M\Phi_M)^\otimes : \tP_M^\otimes \to \O_M[W_M^{-1}]^\otimes$
is a localization of $1$-categories, to solve Open Problem \ref{open:openproblem}.
\end{rem}

%%%%%%%%%%%%%%%%%%%%%%%%%%%%%%%%%%%%%%%%%%%%%%%%
%%%%%%%%%%%%%%%%%%%%%%%%%%%%%%%%%%%%%%%%%%%%%%%%

\section*{Acknowledgments}
We would like to thank Ettore Minguzzi for suggesting to us 
a proof for Lemma \ref{lem:Dclosed}.
The work of M.B.\ is supported in part by the MUR Excellence 
Department Project awarded to Dipartimento di Matematica, 
Universit{\`a} di Genova (CUP D33C23001110001) and it is fostered by 
the National Group of Mathematical Physics (GNFM-INdAM (IT)). 
A.G-S.\ was supported by Royal Society Enhancement Grants (RF\textbackslash ERE\textbackslash 210053
and RF\textbackslash ERE\textbackslash 231077).
A.S.\ gratefully acknowledges the support of 
the Royal Society (UK) through a Royal Society University 
Research Fellowship (URF\textbackslash R\textbackslash 211015)
and Enhancement Grants (RF\textbackslash ERE\textbackslash 210053 and 
RF\textbackslash ERE\textbackslash 231077).

%%%%%%%%%%%%%%%%%%%%%%%%%%%%%%%%%%%%%%%%%%%%%%%%%
%%%%%%%%%%%%%%%%%%%%%%%%%%%%%%%%%%%%%%%%%%%%%%%%%
%
%
%\section*{Data availability statement}
%All data generated or analyzed during this study are contained in this document. 
%
%
%\section*{Conflict of interest statement}
%The authors have no conflict of interest to declare that are relevant to the content of this article. 
%
%%%%%%%%%%%%%%%%%%%%%%%%%%%%%%%%%%%%%%%%%%%%%%%%%
%%%%%%%%%%%%%%%%%%%%%%%%%%%%%%%%%%%%%%%%%%%%%%%%%

\appendix

\section{\label{app:operadiclocalization}Operadic calculus of left fractions}
In this appendix we introduce a concept of operadic calculus of left fractions,
generalizing the corresponding concept from category theory \cite{GabrielZisman}
to operads. We then show that an operadic calculus of left fractions has 
non-trivial homotopical consequences for the corresponding operadic localization problem.
\begin{defi}\label{def:OpCLF}
Let $\Q$ be a $\Set$-valued colored operad and $W\subseteq \mathrm{Mor}^1(\Q)$ a subset of 
the $1$-ary operations in $\Q$. We say that the pair $(\Q,W)$ admits an 
\textit{operadic calculus of left fractions} if the following properties hold true:
\begin{itemize}
\item[(1)] All identity operations $\id_{M} : M\to M$ in $\Q$ belong to $W$.

\item[(2)] Given any two $1$-ary operations $f:M\to N$ and $g:N\to L$ in $W$,
then their operadic composition $g\,f: M\to L$ in $\Q$ belongs to $W$.

\item[(3)] Given any $n$-ary operation $\psi : \und{M}\to N$ in $\Q$ 
and any tuple $\und{w} = (w_1,\dots,w_n): \und{M}\to\und{M}^\prime$ of $1$-ary operations $w_i:M_i\to M_i^\prime$
in $W$, then there exists an $n$-ary operation $\psi^\prime : \und{M}^\prime \to N^\prime$
in $\Q$ and a $1$-ary operation $w^\prime : N\to N^\prime$ in $W$ such that the diagram
\begin{flalign}
\begin{gathered}
\xymatrix{
\ar[d]_-{\und{w}}\und{M}  \ar[r]^-{\psi}~&~ N\ar@{-->}[d]^-{w^\prime}\\
\und{M}^\prime \ar@{-->}[r]_-{\psi^\prime}~&~N^\prime
}
\end{gathered}
\end{flalign}
of operadic compositions in $\Q$ commutes.

\item[(4)] Given any two parallel operations $\psi_1,\psi_2 : \und{M}\to N$ in $\Q$ and
any tuple $\und{w}:\und{M}^\prime\to \und{M}$ of $1$-ary operations in $W$ 
such that $\psi_1 \und{w} = \psi_2\und{w} $ under operadic composition in $\Q$,
then there exists a $1$-ary operation $w^\prime : N\to N^\prime$ in $W$ such that
$w^\prime \psi_1 = w^\prime \psi_2$ under operadic composition in $\Q$. We visualize this graphically by
\begin{flalign}
\begin{gathered}
\xymatrix{
\und{M}^\prime \ar[r]^-{\und{w}}~&~ \und{M} \ar@<0.6ex>[r]^-{\psi_1}\ar@<-0.6ex>[r]_-{\psi_2}~&~N \ar@{-->}[r]^-{w^\prime} ~&~N^\prime
}
\end{gathered}\quad.
\end{flalign}
\end{itemize}
\end{defi}
\begin{rem}\label{rem:OpCLF}
Every operadic calculus of left fractions $(\Q,W)$ induces an ordinary
calculus of left fractions $(\Q^1,W)$ in the sense of \cite{GabrielZisman}
on the subcategory $\Q^1\subseteq \Q$ of $1$-ary operations in $\Q$. The converse 
statement is in general \textit{false} because the square filling and the
coequalization properties in Definition \ref{def:OpCLF} 
are required for operations $\psi,\psi_1,\psi_2$ of any arity $n$.
\end{rem}

We believe that the existence of an operadic calculus of left fractions
implies that the operadic Hammock localization $L_W\O$ is weakly equivalent (as a simplicial operad)
to the ordinary localization $\O[W^{-1}]$ of $\Set$-valued operads, which would provide an operadic
generalization of the result in \cite[Proposition 7.3]{DwyerKan} for 
Dwyer-Kan localizations of categories. In combination with \cite[Theorem 3.13]{Carmona},
this would lead to a proof for the Quillen equivalence statement in Theorem \ref{theo:O_MCLF}. 
Unfortunately, the theory of operadic Hammock localizations is not yet sufficiently developed 
to provide a proof for this claim, which is why we propose an alternative proof strategy
that is based on Lurie's concept of $\infty$-operads \cite[Definition 2.1.1.10]{LurieHA} and 
the results of \cite[Appendix B]{CalaqueCarmona}. (See also \cite[Appendix A.4]{Scheimbauer}
for a similar approach.) A central object in this approach
is the category of operators associated with a colored operad.
\begin{defi}\label{def:operatorcategory}
Let $\Q$ be a $\Set$-valued colored operad. The associated 
\textit{category of operators} $\Q^\otimes$ is the category consisting
of the following objects and morphisms:
\begin{itemize}
\item An object in $\Q^\otimes$ is a pair $\big(\langle n\rangle,\und{M}\big)$
consisting of a pointed finite set $\langle n\rangle := \{0,1,\dots,n\}\in\mathbf{Fin}_\ast$, 
with base point $0\in \langle n\rangle$, and a (possibly empty) 
tuple $\und{M}=(M_1,\dots,M_n)$ of objects $M_i\in\Q$, for all $i=1,\dots,n$.

\item A morphism  $(\phi,\und{\phi}): \big(\langle n\rangle,\und{M}\big)\to
\big(\langle n^\prime\rangle,\und{M}^\prime\big)$ 
in $\Q^\otimes$ is a pair consisting of a base-point preserving function $\phi : 
\langle n\rangle \to \langle n^\prime\rangle$, i.e.\ a $\mathbf{Fin}_\ast$-morphism,
and a tuple $\und{\phi}=(\phi_1,\dots,\phi_{n^\prime})$
of operations $\phi_j : \und{M}_{\phi^{-1}(j)}\to M^\prime_j$ in $\Q$,
for all $j=1,\dots,n^\prime$. Here $\und{M}_{\phi^{-1}(j)}$ denotes the restricted tuple
consisting of all objects $M_i\in\Q$ such that $\phi(i) = j$.
\end{itemize}
Composition of morphisms in $\Q^\otimes$ is given by composition of the underlying
base-point preserving functions and operadic composition of the operations. The identities
in $\Q^\otimes$ are given by $(\id_{\langle n\rangle},\id_{\und{M}}) : 
\big(\langle n\rangle,\und{M}\big)\to \big(\langle n\rangle,\und{M}\big)$, where 
$\id_{\und{M}} = (\id_{M_1},\dots,\id_{M_n})$ is the tuple of identity operations in $\Q$.
The category of operators comes endowed with a canonical functor
$\pi:\Q^\otimes\to \mathbf{Fin}_\ast$ which assigns the underlying pointed finite sets
$\big(\langle n\rangle,\und{M}\big)\mapsto \langle n\rangle$ and their morphisms
$(\phi,\und{\phi})\mapsto \phi$.
\end{defi}

\begin{lem}\label{lem:operatorlocalization}
Suppose that a pair $(\Q,W)$ consisting of a $\Set$-valued colored operad
and a subset $W\subseteq \mathrm{Mor}^1(\Q)$ of the $1$-ary operations in $\Q$
admits an operadic calculus of left fractions. Define the pair $(\Q^\otimes,W^\otimes)$
in terms of the category of operators from Definition \ref{def:operatorcategory}
and the subset $W^\otimes \subseteq\mathrm{Mor}(\Q^\otimes)$ given by all $\Q^\otimes$-morphisms
of the form $(\sigma,\und{w}) :\big(\langle n\rangle,\und{M}\big)\to 
\big(\langle n\rangle,\und{M}^\prime\big) $, where $\sigma :\langle n\rangle\stackrel{\cong}{\to}\langle n\rangle$
is any $\mathbf{Fin}_\ast$-isomorphism and $\und{w} = (w_1,\dots,w_n)$ is any 
tuple of $1$-ary operations $w_i : M_{\sigma^{-1}(i)}\to M_{i}^\prime$ in $W$, for all $i=1,\dots,n$. 
Then the pair $(\Q^\otimes,W^\otimes)$
admits a calculus of left fractions in the sense of \cite{GabrielZisman}.
\end{lem}
\begin{proof}
Using the Definition \ref{def:operatorcategory} of $\Q^\otimes$,
this is a direct and straightforward check.
\end{proof}
Our goal is to show that, for a pair $(\Q,W)$ as in Lemma \ref{lem:operatorlocalization}, 
one can determine the localized operad $\Q[W^{-1}]$ 
by means of the localization $\Q^\otimes[(W^\otimes)^{-1}]$ of the associated 
category of operators $Q^\otimes$ at its subset of morphisms $W^\otimes$, 
which as a consequence of Lemma \ref{lem:operatorlocalization} admits a simple and explicit 
model in terms of the same objects as $\Q^\otimes$ and morphisms given by 
(equivalence classes of) fractions $\big(\langle m \rangle,\und{M}\big) \to
\big(\langle n \rangle,\und{\widetilde{N}} \big) \leftarrow \big(\langle n \rangle,\und{N}\big)$, 
where the right-pointing morphism belongs to $\Q^\otimes$ and the left-pointing 
one belongs to $W^\otimes$. (A detailed description of the relevant equivalence 
relation can be found e.g.\ in \cite[Section I.2]{GabrielZisman}.) 

\begin{lem}\label{lem:catop-vs-loc}
Suppose that a pair $(\Q,W)$ consisting of a $\Set$-valued colored operad
and a subset $W\subseteq \mathrm{Mor}^1(\Q)$ of the $1$-ary operations in $\Q$
admits an operadic calculus of left fractions. Then the localized functor 
$\pi: \Q^\otimes[(W^\otimes)^{-1}] \to \mathbf{Fin}_\ast$ associated with 
the canonical functor $\pi: \Q^\otimes \to \mathbf{Fin}_\ast$ 
from Definition \ref{def:operatorcategory} is equivalent to the category of 
operators associated with the localized colored operad $\Q[W^{-1}]$. 
In particular, the categories $\Q[W^{-1}]^\otimes \simeq \Q^\otimes[(W^\otimes)^{-1}]$ are equivalent. 
\end{lem}
\begin{proof}
To achieve the goal one proceeds in two steps. 
First, one checks that the localized 
category $Q^\otimes[(W^\otimes)^{-1}]$ arises as the category of operators 
associated with a colored operad. Second, one checks 
that the latter colored operad exhibits the localization of $\Q$ at $W$. 
\sk

To carry out the first step, one has to verify that 
$\pi: \Q^\otimes[(W^\otimes)^{-1}] \to \mathbf{Fin}_\ast$ fulfills 
the properties from \cite[Proposition 2.2.11]{allegedly}, which 
characterize categories of operators among all categories over $\mathbf{Fin}_\ast$.
In detail, a functor $p: \EE \to \mathbf{Fin}_\ast$ 
is equivalent to the category of operators associated with a colored operad, 
if and only if 
(a)~there exist $p$-cocartesian lifts for inert morphisms in $\mathbf{Fin}_\ast$, 
(b)~the set $\EE(E,F)_\alpha$ of morphisms $E \to F$ in $\EE$ over $\alpha$ 
is in bijective correspondence with the product 
$\prod_{i \in \langle n \rangle} \EE(E,F_i)_{\rho_i \circ \alpha}$, for all 
morphisms $\alpha: \langle m \rangle \to \langle n \rangle$ in $\mathbf{Fin}_\ast$, 
all objects $E \in \EE$ over $\langle m \rangle \in \mathbf{Fin}_\ast$ and
$F \in \EE$ over $\langle n \rangle \in \mathbf{Fin}_\ast$, 
and all $p$-cocartesian morphisms $F \to F_i$ in $\EE$ 
over the inert morphisms $\rho_i: \langle n \rangle \to \langle 1 \rangle$ 
that send $i$ to $1$ and $j \neq i$ to the base-point $0$, and 
(c) there exists an object $E \in \EE$ over $\langle m \rangle \in \mathbf{Fin}_\ast$ 
and $p$-cocartesian morphisms $E \to E_i$ over $\rho_i$, for all objects $E_1, \ldots, E_n \in \EE$ over $\langle 1 \rangle \in \mathbf{Fin}_\ast$.
Because $\pi: \Q^\otimes \to \mathbf{Fin}_\ast$ is by construction 
the category of operators associated with the colored operad $\Q$, 
it fulfills properties (a-c) with the simple choice of $\pi$-cocartesian lifts 
\begin{equation}\label{eq:cocart-lift}
\widetilde{\rho} := \big( \rho, \und{\id} \big): 
\big( \langle n \rangle, \underline{M} \big) \longrightarrow \big( \langle m \rangle, (M_{\rho^{-1}(1)}, \ldots, M_{\rho^{-1}(m)}) \big)
\end{equation}
in $\Q^\otimes$, for all objects $\big( \langle n \rangle, \underline{M} \big) \in \Q^\otimes$ 
and all inert morphisms $\rho: \langle n \rangle \to \langle m \rangle$ in $\mathbf{Fin}_\ast$. 
Using also the model of the localized category 
$\Q^\otimes[(W^\otimes)^{-1}]$ provided by the calculus of left fractions, 
the choice \eqref{eq:cocart-lift} suggests to consider 
as $\pi$-cocartesian lifts of inert morphisms in the case of 
$\pi: \Q^\otimes[(W^\otimes)^{-1}] \to \mathbf{Fin}_\ast$ the
(equivalence classes of the) fractions 
\begin{equation}\label{eq:cocart-lift-loc}
\xymatrix@C=3em{ 
\big( \langle n \rangle, \underline{M} \big) \ar[r]^-{\widetilde{\rho}} & \big( \langle m \rangle, (M_{\rho^{-1}(1)}, \ldots, M_{\rho^{-1}(m)}) \big) & \big( \langle m \rangle, (\widetilde{M}_{\rho^{-1}(1)}, \ldots, \widetilde{M}_{\rho^{-1}(m)}) \big) \ar[l]_{(\id,\und{\id})} 
}
\end{equation}
in $\Q^\otimes[(W^\otimes)^{-1}]$. 
With the choice \eqref{eq:cocart-lift-loc} one directly checks that 
$\pi: \Q^\otimes[(W^\otimes)^{-1}] \to \mathbf{Fin}_\ast$ 
fulfills properties (a-c) and hence is equivalent to the category of operators 
associated with a colored operad. 
\sk

To carry out the second step, one has to verify that 
$\pi: \Q^\otimes[(W^\otimes)^{-1}] \to \mathbf{Fin}_\ast$ 
is actually equivalent to the category of operators associated 
with a colored operad that fulfills the universal property of 
the localization of the colored operad $\Q$ at $W$. First, one observes that 
functors (over $\mathbf{Fin}_\ast$) defined on $\Q^\otimes[(W^\otimes)^{-1}]$ 
precisely correspond to functors (over $\mathbf{Fin}_\ast$) 
defined on $\Q^\otimes$ that send $W^\otimes$ to isomorphisms. 
Therefore, by \cite[Proposition 2.2.13]{allegedly} one is left with the check 
that the above correspondence restricts to functors preserving 
$\pi$-cocartesian lifts of inert morphisms. This, however, follows from the 
fact that our choice of $\pi$-cocartesian lifts \eqref{eq:cocart-lift-loc} 
for $\pi: \Q^\otimes[(W^\otimes)^{-1}] \to \mathbf{Fin}_\ast$ is exactly 
the image under the localization functor 
$\Q^\otimes \to \Q^\otimes[(W^\otimes)^{-1}]$ of the $\pi$-cocartesian lifts 
\eqref{eq:cocart-lift} for $\pi: \Q^\otimes \to \mathbf{Fin}_\ast$. 
\end{proof}

The main result of this appendix is the following proposition.
\begin{propo}\label{prop:CLFhomotopical}
Suppose that a pair $(\Q,W)$ consisting of a $\Set$-valued colored operad
and a subset $W\subseteq \mathrm{Mor}^1(\Q)$ of the $1$-ary operations in $\Q$
admits an operadic calculus of left fractions. Then the localization
multifunctor $L : \Q\to \Q[W^{-1}]$ of $\Set$-valued operads
(see \cite[Definition 2.12]{BCStimeslice}) induces a Quillen equivalence
\begin{flalign}\label{eqn:L!Lastadjunction}
\xymatrix{
L_!\,:\,\mathcal{L}_{\widehat{W}}\Alg_{\Q}\big(\Ch_R\big) \ar@<0.75ex>[r]~&~\ar@<0.75ex>[l]
\Alg_{\Q[W^{-1}]}\big(\Ch_R\big)\,:\,L^\ast
}
\end{flalign}
between the projective model category $\Alg_{\Q[W^{-1}]}\big(\Ch_R\big)$ from Theorem \ref{theo:projectivemodel}
and the left Bousfield localized semi-model category $\mathcal{L}_{\widehat{W}}\Alg_{\Q}\big(\Ch_R\big)$ 
from Theorem \ref{theo:Bousfieldmodel}.
\end{propo}
\begin{proof}
The statement that \eqref{eqn:L!Lastadjunction} is a Quillen adjunction is a simple consequence
of Proposition \ref{prop:Quillenprojective} and \cite[Lemma 3.5]{CarmonaCriterion}. Indeed,
the right adjoint functor $L^\ast$ sends every object in $\Alg_{\Q[W^{-1}]}\big(\Ch_R\big)$
to an object in $\Alg_{\Q}\big(\Ch_R\big)$ that is $\widehat{W}$-local in the sense of Theorem \ref{theo:Bousfieldmodel}.
\sk

Our proof that \eqref{eqn:L!Lastadjunction} is further a Quillen equivalence 
is considerably more involved and indirect. Using the results of 
\cite[Corollary 4.11]{Haugseng} and \cite[Theorem 7.3.1]{WhiteYau},
one can translate this problem to an equivalent $\infty$-categorical problem for 
algebras over $\infty$-operads in the sense of Lurie \cite{LurieHA}. 
Concretely, the equivalent problem is to prove that the pullback along the localization multifunctor
$L : \Q\to \Q[W^{-1}]$ of $\Set$-valued operads (regarded as a morphism of $\infty$-operads)
induces an equivalence of $\infty$-categories
\begin{flalign}\label{eqn:L!infty}
L^\ast\,:\,\mathcal{A}\mathsf{lg}_{\Q[W^{-1}]}\big(\mathcal{C}\mathsf{h}_R\big)~\stackrel{\sim}{\longrightarrow}
~ \mathcal{A}\mathsf{lg}_{\Q,W}\big(\mathcal{C}\mathsf{h}_R\big)\quad,
\end{flalign}
where $\mathcal{A}\mathsf{lg}_{\Q[W^{-1}]}\big(\mathcal{C}\mathsf{h}_R\big)$ denotes the 
$\infty$-category of $\Q[W^{-1}]$-algebras with values in the symmetric monoidal
$\infty$-category $\mathcal{C}\mathsf{h}_R$ of cochain complexes and 
$\mathcal{A}\mathsf{lg}_{\Q,W}\big(\mathcal{C}\mathsf{h}_R\big)\subseteq
\mathcal{A}\mathsf{lg}_{\Q}\big(\mathcal{C}\mathsf{h}_R\big)$ denotes the 
full $\infty$-subcategory spanned by all $\Q$-algebras
sending $W$ to equivalences. We will now argue that the 
multifunctor $L : \Q\to \Q[W^{-1}]$ exhibits an $\infty$-localization
of $\infty$-operads of $\Q$ at $W$, which implies that \eqref{eqn:L!infty} is 
an equivalence of $\infty$-categories. Using the results of \cite[Appendix B]{CalaqueCarmona} in combination with the equivalence 
$\Q[W^{-1}]^{\otimes}\simeq \Q^\otimes[(W^\otimes)^{-1}]$ 
from Lemma \ref{lem:catop-vs-loc},
this $\infty$-operadic localization statement follows if we can show that the associated functor
$L^\otimes : \Q^\otimes\to \Q[W^{-1}]^{\otimes}\simeq \Q^\otimes[(W^\otimes)^{-1}]$
between the categories of operators from Definition \ref{def:operatorcategory}
exhibits an $\infty$-localization of $\infty$-categories. 
Using now the induced calculus of left fractions for the categories of operators
from Lemma \ref{lem:operatorlocalization}, the proof follows 
from the result in \cite[Proposition 7.3]{DwyerKan}.
\end{proof}

%%%%%%%%%%%%%%%%%%%%%%%%%%%%%%%%%%%%%%%%%%%%%%%%
%%%%%%%%%%%%%%%%%%%%%%%%%%%%%%%%%%%%%%%%%%%%%%%%

\section{\label{app:Lorentz}Lorentzian geometric details}
This appendix is devoted to proving some results of Lorentzian geometric flavor 
which are crucial for Section \ref{sec:spacetimewise}. We shall freely use basic concepts 
and tools from Lorentzian geometry that can be easily found in the literature, 
see e.g.\ the monograph \cite{ONeill}, as well as the excellent reviews 
\cite{Minguzzi} and \cite[Section 1.3 and Appendix A.5]{BGP}.
The Cauchy development $D_M(U)\subseteq M$ of a subset $U \subseteq M$ of Lorentzian manifold $M$
refers to the set of points $p \in M$ such that
every inextendable causal curve through $p$ also intersects $U$.
\sk

The next two results play a key role in the proof of Theorem \ref{theo:O_MCLF}. 
\begin{lem}\label{lem:zigzag}
Let $M$ be a time-oriented globally hyperbolic Lorentzian manifold.
Consider any family of causally convex opens $U_1, \ldots, U_n, V \subseteq M$ 
which are either Cauchy or relatively compact. 
Then the following two conditions are equivalent:
\begin{itemize}
\item[(1)] There exists a causally convex open $V^\prime \subseteq M$ 
which is either Cauchy or relatively compact and which contains $V$ and $U_i$, 
for all $i = 1, \ldots, n$, in such a way that the zig-zags 
$U_i \subseteq V^\prime \supseteq V$ consist of inclusions 
$U_i \subseteq V^\prime$ that are either Cauchy or relatively compact 
and $V \subseteq V^\prime$ that is Cauchy. 

\item[(2)] The inclusions $U_i \subseteq D_M(V)$ are either Cauchy or relatively compact, for all $i = 1, \ldots, n$.
\end{itemize}
\end{lem}
\begin{proof}
$(1)\Rightarrow (2)$: It follows from the properties of the Cauchy development $D_M$ that 
all inclusions $U_i \subseteq V^\prime \subseteq D_M(V^\prime) = D_M(V)$ are 
either Cauchy or relatively compact. 
\sk

$(2)\Rightarrow (1)$: We make a case distinction. 
If $V \subseteq M$ is Cauchy, then (1) holds with the choice $V^\prime = M$. 
Otherwise, $V \subseteq M$ is not Cauchy, hence it is relatively compact. 
As a consequence, we observe that all subsets 
$D_M(U_i) \subseteq D_M(V) \subset M$ are strictly contained in $M$. 
In particular, all $U_i \subseteq M$ are not Cauchy, hence they are relatively compact. 
Without loss of generality, suppose that $U_i \subseteq D_M(V)$ is Cauchy,
for $i = 1, \ldots, k$, and $U_i \subseteq D_M(V)$ is not Cauchy,
hence relatively compact, for $i = k + 1, \ldots, n$.
The causally convex hull
$W$ of the union $U_1 \cup \cdots \cup U_k \cup V$ is a relatively compact 
causally convex open subset in $M$ by \cite[Lemma B.4]{BGSHaagKastler},
and moreover the proof of \cite[Proposition B.2]{BGSHaagKastler} adapts
to show that all inclusions $U_i \subseteq W$, for $i = 1, \ldots, k$, and $V \subseteq W$ are Cauchy.
Since $U_i \subseteq D_M(V)$ is relatively compact for all $i = k + 1, \ldots, n$,
it follows that there exists a relatively compact causally convex open $X \subseteq D_M(V)$ 
such that $U_i \subseteq X$ is relatively compact for all $i = k + 1, \ldots, n$.
(Explicitly, $X$ can be constructed by taking a compact subset of $D_M(V)$ that contains all $U_i$, 
for $i = k + 1, \ldots, n$, covering it by finitely many relatively 
compact opens of $D_M(V)$ and forming the causally convex hull of their union.) 
We observe that $X \subseteq D_M(V) = D_M(W)$, hence 
it follows again from \cite[Appendix B]{BGSHaagKastler} that the causally convex hull
$V^\prime$ of the union $X \cup W$ has the following properties: 
$V^\prime \subseteq M$ is a relatively compact causally convex open and there are inclusions 
$X \subseteq V^\prime$ and $W \subseteq V^\prime$, the latter being Cauchy. 
Summing up, one has that all inclusions $U_i \subseteq W \subseteq V^\prime$, 
for $i = 1, \ldots, k$, and $V \subseteq W \subseteq V^\prime$ 
are Cauchy and all inclusions $U_i \subseteq X \subseteq V^\prime$, 
for $i = k + 1, \ldots, n$, are relatively compact. This shows that condition (1) holds. 
\end{proof}

\begin{lem}\label{lem:Dclosed}
Let $M$ be a time-oriented globally hyperbolic Lorentzian 
manifold. Then the Cauchy development $D_M(K) \subseteq M$ 
of any compact subset $K \subseteq M$ is closed.\footnote{The proof of 
this result was suggested to us by Ettore Minguzzi.}
\end{lem}
\begin{proof}
The Cauchy development $D_M(K) = D^+_M(K) \cup D^-_M(K)$ 
is the union of the future and past Cauchy developments, 
where one considers only past-, respectively future-, inextendible future-directed causal curves. 
We shall now show that $D^+_M(K) \subseteq M$ is closed. 
For this purpose, given any $p \in M \setminus D^+_M(K)$, we construct 
an open neighborhood $O \subseteq M \setminus D^+_M(K)$ of $p$. 
By definition of future Cauchy development, see e.g.\ \cite[Definition 3.1]{Minguzzi}, 
there exists a past-inextendible future-directed causal curve 
$\gamma: (-1,0] \to M$ ending at $p$ that does not meet $K$. 
Since $J_M^+(K) \cap J_M^-(p) \subseteq M$ is compact and $M$ 
is globally hyperbolic, it follows from \cite[Proposition 4.80]{Minguzzi} that there exists $s \in (-1,0)$ 
such that $\gamma(t) \notin J_M^+(K) \cap J_M^-(p)$, for all $t \in (-1,s]$. 
Let $\gamma^\prime$ be the future-directed causal curve obtained by concatenating 
a past-inextendible future-directed timelike curve ending at $\gamma(s)$ with 
the restriction of $\gamma$ to $[s,0] \subseteq (-1,0]$. 
By construction $\gamma^\prime$ contains a timelike segment ending at $\gamma(s)$ 
and, moreover, it does not meet $K$ (otherwise $\gamma(s) \in J_M^+(K) \cap J_M^-(p)$, a contradiction). 
It follows that, for $U \subseteq M \setminus K$ an open neighborhood of 
$\gamma^\prime([s,0]) = \gamma([s,0])$ that does not meet $K$, 
there exists a past-inextendible future-directed timelike curve $\gamma^{\prime\prime}$ 
that ends at $p$ and agrees with $\gamma^\prime$ outside $U$. 
(This is achieved by deforming $\gamma^\prime$ according to \cite[Chapter 10, Proposition 46]{ONeill} 
or \cite[Theorem 2.22 and subsequent comment]{Minguzzi}.) 
Let $V \subseteq M \setminus K$ be a causally convex open neighborhood of $p$ that does not meet $K$ 
and $q \in V$ a point of $\gamma^{\prime\prime}$. Consider $O := I_V^+(q) \subseteq V$ 
as candidate open neighborhood of $p$. By construction, every point of $O$ 
is the endpoint of a past-inextendible future-directed timelike curve that does not meet $K$. 
This means that $O \subseteq M \setminus D^+_M(K)$ does not meet the future Cauchy development of $K$.
A similar argument shows that $D^-_M(K) \subseteq M$ is closed too, 
thus concluding the proof. 
\end{proof}

As a consequence of Theorem \ref{theo:O_MCLF} and Remark \ref{rem:OpCLF}, the pair
$(\O_M^1,W_M) = (\Loc^\rc/M,W_M)$ consisting of the underlying category 
of $1$-ary operations in the AQFT operad $\O_M$ and the subset $W_M\subseteq \mathrm{Mor}(\O_M^1)$ 
of all Cauchy morphisms admits a calculus of left fractions.
Using \cite{GabrielZisman} and \cite{BSWoperad,BCStimeslice}, one obtains 
a model for the localization $L_M: \Loc^\rc/M \to (\Loc^\rc/M)[W_M^{-1}]$ 
and a characterization for the pushforward of the orthogonality relation on $\Loc^\rc/M$ 
given by causal disjointness:
\begin{itemize}
\item The localized category $(\Loc^\rc/M)[W_M^{-1}]$ has the same objects as $\Loc^\rc/M$
and its morphisms $[X]: U \to V$ are equivalence classes of objects 
$X \in \Loc^\rc/M$  with morphisms $(U \subseteq X) \in \Loc^\rc/M$ and $(V \subseteq X) \in W_M$. 
Two such $X, X^\prime \in \Loc^\rc/M$ are equivalent if there exists 
a third $X^{\prime\prime} \in \Loc^\rc/M$ with morphisms 
$(X \subseteq X^{\prime\prime}), (X^\prime \subseteq X^{\prime\prime}) \in \Loc^\rc/M$ 
such that $(V \subseteq X^{\prime\prime}) \in W_M$. 
The composite of $[X]: U \to V$ and $[Y]: V \to W$ 
is given by $[Y] \circ [X] = [Z]: U \to W$, where $(X \subseteq Z) \in \Loc^\rc/M$ and $(Y \subseteq Z) \in W_M$
are obtained from $(V \subseteq X) \in W_M$ and $(V \subseteq Y) \in \Loc^\rc/M$ 
via the calculus of left fractions. (See property (3) in Definition \ref{def:OpCLF}.)

\item The localization functor $L: \Loc^\rc/M \to (\Loc^\rc/M)[W_M^{-1}]$ 
acts as the identity on objects and sends a morphism $(U \subseteq V) \in \Loc^\rc/M$ 
to $([V]: U \to V) \in (\Loc^\rc/M)[W_M^{-1}]$.

\item The pushforward orthogonality relation on $(\Loc^\rc/M)[W_M^{-1}]$ is characterized as follows: 
$([X_1]: U_1 \to V) \perp ([X_2]: U_2 \to V)$ is an orthogonal pair 
if and only if $U_1 \subseteq M$ and $U_2 \subseteq M$ are causally disjoint subsets of $M$.
\end{itemize}
\begin{propo}\label{prop:localizedrcmor}
For any object $M\in\Loc^\rc$, the category $\O_M^1[W_M^{-1}] = (\Loc^\rc/M)[W_M^{-1}]$ is thin, 
i.e.\ there exists at most one morphism between any two objects. Moreover, the unique morphism $U \to V$
exists if and only if the inclusion $U \subseteq D_M(V)$ 
in the Cauchy development of $V$ is either Cauchy or relatively compact. 
\end{propo}
\begin{proof}
Consider any two parallel morphisms $[X], [X^\prime]: U \to V$ in $(\Loc^\rc/M)[W_M^{-1}]$.
From $(V \subseteq X), (V \subseteq X^\prime) \in W_M$ 
and property (3) of Definition \ref{def:OpCLF}, we obtain two inclusion morphisms
$(X \subseteq X^{\prime\prime}), (X^\prime \subseteq X^{\prime\prime}) \in \Loc^\rc/M$ 
such that $(V \subseteq X^{\prime\prime}) \in W_M$.
This witnesses the equality $[X] = [X^\prime]$ according to 
the above model for the localized category $(\Loc^\rc/M)[W_M^{-1}]$.
\sk

For the second claim consider two objects $U, V \in (\Loc^\rc/M)[W_M^{-1}]$. 
According to the above model for the localized category $(\Loc^\rc/M)[W_M^{-1}]$, 
the existence of a morphism $[X]: U \to V$ is equivalent to 
the existence of an object $X \in \Loc^\rc/M$ with morphisms 
$(U \subseteq X) \in \Loc^\rc/M$ and $(V \subseteq X) \in W_M$, 
which by Lemma \ref{lem:zigzag} is equivalent to the condition that 
the inclusion $U \subseteq D_M(V)$ is either Cauchy or relatively compact. 
\end{proof}

%%%%%%%%%%%%%%%%%%%%%%%%%%%%%%%%%%%%%%%%%%%%%%%%
%%%%%%%%%%%%%%%%%%%%%%%%%%%%%%%%%%%%%%%%%%%%%%%%

\section*{Data availability statement}
All data generated or analyzed during this study are contained in this document.

\section*{Conflict of interest statement}
The authors have no conflict of interest to declare that are relevant to the content of this article. 

%%%%%%%%%%%%%%%%%%%%%%%%%%%%%%%%%%%%%%%%%%%%%%%%
%%%%%%%%%%%%%%%%%%%%%%%%%%%%%%%%%%%%%%%%%%%%%%%%

%%%%%%%%%%%%%%%%%%%%%%%%

\end{document}